\documentclass[showpacs,preprintnumbers,amsmath,amssymb]{article}

\usepackage{dsfont}  
\usepackage{fullpage} 
\usepackage{amsmath,amssymb,amsfonts,amsthm}
\usepackage{bbold}
\usepackage{authblk}
\usepackage[english]{babel}
\usepackage{graphicx}
\usepackage{epsfig}

\makeatletter
\def\cont{\mathbin{\dimen0=\ht\strutbox \dimen1=\ht\strutbox \divide\dimen0 by 2
\divide\dimen1 by 4
  \hbox{\vbox{\hrule width\dimen1}\hskip-0.4pt\vrule   height\dimen0}}\,}

\usepackage{chngcntr}
\usepackage{apptools}
\AtAppendix{\counterwithin{proposition}{section}}

\newtheorem{definition}{Definition}  

\newtheorem{proposition}{Proposition}

\usepackage{bm}

\makeatletter

\DeclareMathOperator{\tr}{tr}

\newcommand{\vOm}{\mathsf{\Omega}}
\newcommand{\vR}{\mathsf{R}}

\newcommand{\Real}{\mathbb{R}} 
\newcommand{\sltR}{{\mathfrak sl}(2,\Real)}
\newcommand{\Complex}{\mathbb{C}} 
\newcommand{\sltC}{{\mathfrak sl}(2,\Complex)}
\newcommand{\fg}{\mathfrak{g}}
\newcommand{\fk}{\mathfrak{k}}
\newcommand{\fh}{\mathfrak{h}}

\newcommand{\E}{{\cal E}}
\newcommand{\Phat}{\hat{P}}
\newcommand{\A}[3][ ]{{\overset {#1}{A}}_{\mathfrak{#2}\mathfrak{#3}}}
\newcommand{\sA}[3][ ]{{\overset {#1}{A}}_{\shat{\mathfrak{#2}}\mathfrak{#3}}}

\newcommand{\sAs}[3][ ]{{\overset {#1}{A}}_{\shat{\mathfrak{#2}}\shat{\mathfrak{#3}}}}

\newcommand{\cN}{{\cal N}}

\newcommand{\cD}{{\cal D}}

\newcommand{\cV}{{\cal V}} 

\newcommand{\U}{{\cal U}} 
\newcommand{\cP}{\mathcal P} 
\newcommand{\ag}{\alpha}

\newcommand{\dg}{\delta}

\newcommand{\eg}{\epsilon}

\newcommand{\lam}{\lambda}

\newcommand{\vareg}{\varepsilon}
 
\newcommand{\sg}{\sigma} 
\newcommand{\dI}{\mathbb{d}} 
\newcommand{\di}{\partial}

\newcommand{\be}{\begin{equation}} 
\newcommand{\ee}{\end{equation}} 
\newcommand{\bearr}{\begin{eqnarray}}
\newcommand{\eearr}{\end{eqnarray}}

\newcommand{\Cyc}[1]{
  \vtop{\mathsurround=0pt
  \ialign{##\crcr$\textstyle{\rm Cyc}\strut$\crcr
    \noalign{\kern-0.4ex\nointerlineskip}{\tiny#1}\crcr}}\ }
        
\newcommand{\de}{\delta}
\newcommand{\Om}{\Omega}
\newcommand{\Vt}{\tilde{\mathcal V}}
\newcommand{\V}{\mathcal V}
\newcommand{\sk}{\mathfrak k}
\newcommand{\R}{\mathbb R}
\newcommand{\One}{\mathds{1}}
\newcommand{\slt}{\mathfrak{sl}(2,\mathbb R)}
\newcommand{\sot}{\mathfrak{so}(2)}
\newcommand{\Omg}{\Omega_\fg}
\newcommand{\Omk}{\Omega_\fk}
\newcommand{\Omh}{\Omega_\fh}
\newcommand{\Vh}{\hat{\mathcal{V}}}
\newcommand{\Ve}{\overset 1 \V}
\newcommand{\Vz}{\overset 2 \V}
\newcommand{\Vhe}[1][]{\overset{1}{\Vh}{}^{#1}}
\newcommand{\Vhz}[1][]{\overset{2}{\Vh}{}^{#1}}
\newcommand{\Pe}{\overset{1}{P}}
\newcommand{\To}{\overset{1}{T}}
\newcommand{\Pz}{\overset{2}{P}}
\newcommand{\Tt}{\overset{2}{T}}
\newcommand{\Jh}{\hat{J}}
\newcommand{\oo}[1]{\overset 1{#1}}
\newcommand{\ot}[1]{\overset 2{#1}}
\newcommand{\oot}[1]{\overset {12}{#1}}
\newcommand{\norm}[1]{\parallel\! #1\!\parallel}
\newcommand{\Akh}{A_{\fk\fh}}
\newcommand{\Ahk}{A_{\fh\fk}}
\newcommand{\shat}[1]{\underset{\check{}}{#1}}
\newcommand{\Asgg}{A_{\shat{\fg}\shat{\fg}}}
\newcommand{\ue}{\oo u}
\newcommand{\uz}{\ot u}
\newcommand{\M}{\mathcal M}
\newcommand{\ove}[2]{\overset{#1}{#2}}

\begin{document} 

\title{Integrable structures and the quantization of free null initial data for gravity}
\author[1,2]{\bf Andreas Fuchs \thanks{Email: afuchs@geometrie.tuwien.ac.at}}
\author[1]{\bf Michael P. Reisenberger \thanks{Email: miguel@fisica.edu.uy}}
\affil[1]{Instituto de F\a'{\i}sica, Facultad de Ciencias,\protect\\
Universidad de la Rep\a'ublica Oriental del Uruguay,\protect\\
Igu\a'a 4225, esq. Mataojo, Montevideo, Uruguay\\
\ }

\affil[2]{Institut f\"{u}r Diskrete Mathematik und Geometrie, \protect\\
Technische Universit\"{a}t Wien, \protect\\
Wiedner Hauptstra{\ss}e 8-10/104, A-1040 Vienna, Austria} 

\date{March 25, 2017}

\maketitle

\begin{abstract}
Variables for constraint free null canonical vacuum general relativity are presented which have simple Poisson brackets that facilitate
quantization. Free initial data for vacuum general relativity on a pair of intersecting null hypersurfaces has been known since the 1960s. These consist
of the ``main'' data which are set on the bulk of the two null hypersurfaces, and additional ``surface'' data set only on their intersection 2-surface.
More recently the complete set of Poisson brackets of such data has been obtained. However the complexity of these brackets is an obstacle 
to their quantization. Part of this difficulty may be overcome using methods from the treatment of cylindrically symmetric gravity. Specializing 
from general to cylindrically symmetric solutions changes the Poisson algebra of the null initial data surprisingly little, but cylindrically 
symmetric vacuum general relativity is an integrable system, making powerful tools available. Here a transformation is constructed at the 
cylindrically symmetric level which maps the main initial data to new data forming a Poisson algebra for which an exact deformation
quantization is known. (Although an auxiliary condition on the data has been quantized only in the asymptotically flat case, and a suitable 
representation of the algebra of quantum data by operators on a Hilbert space has not yet been found.)
The definition of the new main data generalizes naturally to arbitrary, symmetryless gravitational fields, with the Poisson brackets 
retaining their simplicity. The corresponding generalization of the quantization is however ambiguous and requires further analysis. 
\end{abstract}

\section{Introduction}

\noindent Free (unconstrained) initial data for General Relativity (GR) on certain 
types of piecewise null hypersurfaces have been known since the 1960s 
\cite{Sachs,Penrose,Bondi,Dautcourt}. In \cite{MR,PRL} the Poisson brackets were 
found for a complete set of free data on a {\em double null sheet}. This is a compact 
hypersurface $\cN$ consisting of two null branches, $\cN_L$ and $\cN_R$, 
swept out by the two congruences of future directed, normal null geodesics  (called 
{\em generators}) emerging from a spacelike 2-disk $S_0$. The two branches are 
truncated on disks $S_L$ and $S_R$ respectively before any of the generators form 
a caustic or cross.  (See Fig.\ \ref{Nfigure}.)

\begin{figure}
\begin{center}

\includegraphics[height=5cm]{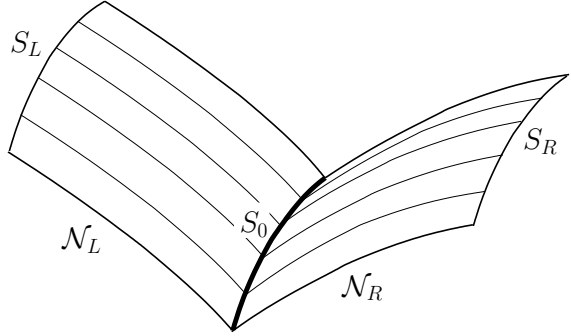}
\caption{A double null sheet in 2+1 dimensional spacetime. In 3+1 dimensions $S_0$ is a disk instead of a line segment, and $\cN_L$ and $\cN_R$ are solid cylinders
instead of 2-surfaces.}
\label{Nfigure}
\end{center}
\end{figure}

One of the chief motivations for calculating these brackets is the hope that they 
might be quantized to yield a canonical quantization of vacuum GR. But, although 
the brackets obtained are not overly complicated, it is by no means obvious how to 
quantize them. Fortunately it seems a large part of the difficulty can be overcome 
by first solving the problem in the simpler context of cylindrically symmetric 
gravitational waves. The Poisson brackets of the main initial datum of 
\cite{MR,PRL}, a complex field $\mu$ called the {\em Beltrami coefficient} on $\cN$, 
are essentially the same in the cylindrically symmetric case as in the general case, 
but cylindrically symmetric gravity is an integrable system which has been studied 
intensively \cite{Belinskii_Zakharov, Maison, Breitenlohner_Maison, Husain}. In particular 
its quantization has been explored in several works: 
\cite{Nicolai91, KorotkinSamtleben, Mena, Kuchar, Ashtekar_Pierri} and others. 

Using ideas from this literature we construct a non-local change of variables which 
replaces $\mu$ by a new field $\E_{ab}$ on $\cN$, a $2 \times 2$ matrix which we will 
call the {\em deformed conformal metric}. The Poisson brackets of $\E$ are simpler 
than those of  $\mu$ and, more importantly, in the cylindrically symmetric case the 
quantization of $\E$ can essentially be read off from the quantization of the closely 
related {\em monodromy matrix}\footnote{
The term ``monodromy matrix'' is often used to denote the holonomy of the Lax connection 
around a curve in space. Here, on the contrary, it is used in the sense of \cite{Breitenlohner_Maison}: 
the monodromy matrix encodes a monodromy in the spectral parameter plane.} 
$\M_{ab}$ given by Korotkin and Samtleben 
\cite{KorotkinSamtleben}. The transformation $\mu \rightarrow \E$ works also for 
gravitational fields without cylindrical symmetry, and simplifies the Poisson 
brackets also in this general context. Even the quantization of $\E$ extends formally 
to the symmetryless case, but unfortunately ambiguous products of delta distributions 
appear in the commutation relations. Perhaps a quantization of the full set of 
cylindrically symmetric initial data, instead of just the main datum,
would help to disambiguate these relations. This will not be attempted here. 

The  quantization of  \cite{KorotkinSamtleben} is natural in that it is adapted to the 
infinite dimensional group of dynamical symmetries of cylindrically symmetric gravity, 
namely the Geroch group \cite{Geroch, Kinnersley, Kinnersley_Chitre1, Kinnersley_Chitre2, 
Kinnersley_Chitre3, Julia85}, and it is complete at the algebraic level. It is 
an exact specification of the associative $*$-algebra of the quantized monodromy matrix 
elements, that is, a specification of the commutators of these data, including all terms 
of higher order in $\hbar$, and of the action of complex conjugation $*$. A unitary 
representation of the algebra by operators on a Hilbert space is, however, not given. (Actually 
such a representation was proposed in \cite{KorotkinSamtleben_PRL}, but unitarity 
was not demonstrated.)

A large part of the present paper is dedicated to obtaining the Poisson brackets of 
$\E$ from the brackets of the null datum $\mu$ given in \cite{PRL}, both in the cylindrically symmetric 
case and in the general case. In the cylindrically symmetric case the brackets we obtain are equivalent 
to the brackets for the monodromy matrix $\M$ found by Korotkin and Samtleben in \cite{KorotkinSamtleben}, 
and quantized by them. This equivalence was expected, but since their brackets were obtained from those of 
spacelike initial data instead of null initial data it is by no means trivial.  

Actually our calculation of the bracket generalizes the result of \cite{KorotkinSamtleben} somewhat even in the 
cylindrically symmetric case, and it closes a logical gap in their calculation.
It generalizes the result of \cite{KorotkinSamtleben} because it does not assume that spacetime is 
asymptotically flat in any sense. Assumptions about the asymptotic geometry of spacetime cannot be implemented as 
restrictions on the data on our null initial data hypersurface, which is compact. We shall see that the quantization 
of \cite{KorotkinSamtleben} of the Poisson algebra can be taken over almost unchanged to the non asymptotically flat 
context. Only the extension to this context of the quantization of an auxiliary condition, $\det \M = 1$, presents 
difficulties (which we will not attempt to resolve here). 

The logical gap that we close in the calculation of \cite{KorotkinSamtleben} is the following: They evaluate the 
bracket of certain fields at coinciding points by taking the limit of the bracket at non-coinciding points as the 
points approach each other. Indeed it is an important result of their work that this limit exists. But of course
such a procedure can in general lead to errors, as it would, for instance, in the case of two canonically conjugate fields.
In the present work the bracket is evaluated directly, without recourse to this point splitting procedure
(and the result of \cite{KorotkinSamtleben} is confirmed).

The remainder of the paper is organized as follows. In section \ref{data0} $\mu$, the main free null datum of 
\cite{PRL}, and $\rho$, the ``area density``, are defined; the Poisson algebra of $\mu$, its complex conjugate 
$\bar{\mu}$, and $\rho$ is reviewed; and the corresponding symmetry reduced data and brackets in the cylindrically 
symmetric model are presented. In section \ref{transformation} $\E$ is defined in terms of the data $\mu$ and $\rho$
in the cylindrically symmetric context. The relation of $\E$ to the variables of \cite{KorotkinSamtleben} is explained 
in section \ref{KS_variables}. Then, in section \ref{brackets}, the Poisson brackets of $\E$ are calculated from 
those of $\mu$, $\bar{\mu}$, and $\rho$ given in section \ref{data0}. The paper closes with a 
presentation of the generalization of our classical results to gravitational fields without cylindrical symmetry 
in section \ref{symmetryless}, and a brief statement of the quantization of the Poisson algebra of the $\E$ obtained 
from the results of \cite{KorotkinSamtleben} in section \ref{quantization}. An appendix on path ordered 
exponentials is included.  

Of course many things are {\em not} done in this paper: The transformation $\mu \rightarrow \E$ is invertible. 
For any $\E$ that is regular in a suitable sense there is a unique Beltrami coefficient $\mu$ that transforms 
to $\E$ and is regular at the symmetry axis. However, since the demonstration of this claim requires the 
definition of a number of structures not needed for the remaining results, it will not be included here.

It should also be emphasized that we will only discuss the Poisson algebra and quantization of a {\em subset} 
of the initial data, including the main datum $\mu$, and not of the whole set of null initial data on $\cN$ defined in 
\cite{MR, PRL, MR13}. 

\section{Free null initial data and Poisson brackets with and without cylindrical symmetry}\label{data0}

In the classical gravitational fields we will consider spacetime will be assumed to be a smooth manifold, 
and the metric on it everywhere smooth and Lorentzian. The sole exception will be at the symmetry axis of 
cylindrically symmetric fields, where other regularity conditions will be imposed. 
The intersection 2-surface $S_0$ of the double null sheet $\cN$ on which initial data is set will be assumed 
to be smoothly embedded in spacetime. These assumptions imply that the generators of $\cN$ are smoothly embedded 
and that the branches $\cN_A$ ($A = L,R$) are also, provided that the truncating surfaces $S_A$ are smooth.

The Beltrami coefficient $\mu$, the main initial datum of \cite{PRL}, encodes the conformal structure of the 
induced metric on $\cN$: If a chart $(x_A, \theta^1, \theta^2)$ is chosen on the branch $\cN_A$ ($A = L,R$) such 
that the $\theta^a$ ($a = 1,2$) are constant along the generators then $\di_{x_A}$ is tangent to the generators 
and hence null and normal to all tangents of $\cN_A$,\footnote{ 
Let $k_A \propto \di_{x_A}$ be the tangent to the generators corresponding to an affine parametrization of these. Clearly $k_A$ is 
normal to $\cN_A$ at $S_0$, since it is normal to both $S_0$ and to itself (being null). Any tangent $t$ to $\cN_A$ at any point can 
be obtained by Lie dragging a tangent to $\cN_A$ at $S_0$ to that point along $k_A$. But this Lie dragging leaves the inner product 
with $k_A$ unchanged, since $\pounds_{k_A} (k_A \cdot t) = t \cdot \nabla_{k_A} k_{A} + k_A \cdot \nabla_{k_{A}} t = 0$ because 
$\nabla_{k_A} k_{A} = 0$ and $k_A \cdot \nabla_{k_{A}} t = k_A \cdot \nabla_{t} k_{A} = \tfrac{1}{2}\nabla_{t} k^2_{A} = 0$. $k_A$ is 
thus normal to $\cN_A$ everywhere.} 
implying that the line element on $\cN_A$ takes the form
\be   \label{line_element}
ds^2 = h_{ab} d\theta^a d\theta^b,
\ee
with no $dx_A$ terms. In other words, the induced metric is effectively a two  
dimensional Riemannian metric on cross sections of $\cN$ transverse to the 
generators. Using the complex coordinate $z = \theta^1 + i\theta^2$ one may rewrite 
the line element as
\be             \label{mu_parametrization}
       ds^2 = \rho (1 - \mu\bar{\mu})^{-1}[dz + \mu d\bar{z}][d\bar{z} + \bar{\mu}dz],
   \ee
with $\mu$ a complex number valued field of modulus less than 1, $\bar{\mu}$ its 
conjugate, and $\rho = \sqrt{\det h}$ the area density transverse to the 
generators. $\mu$ is the Beltrami coefficient. It encodes the two 
real degrees of freedom of the unimodular matrix $e_{ab} = h_{ab}/\rho$. $e$ will be 
called the {\em conformal 2-metric} because it captures precisely the degrees of freedom 
of $h$ that are invariant under local rescalings. (The parametrization of $e_{ab}$ by 
$\mu$ and $\bar{\mu}$ also works when $e_{ab}$ is complex, but then $\bar{\mu}$ is no 
longer the complex conjugate of $\mu$.)

The free data used in \cite{MR, PRL} consists of $\mu$ given on all of $\cN$ and some 
additional data specified only on the intersection 2-surface $S_0$, including among others 
$\rho_0$, the area density on $S_0$. The data on $S_0$ are specified as a function of the coordinates 
$\theta$, while $\mu$ is specified on each branch $\cN_A$ as a function of the $\theta$ (as before, constant 
along the generators) and the {\em area parameter}, $v_A$, which is set to $1$ on $S_0$ and is 
proportional to $\sqrt{\rho}$ on each generator so that    
\begin{equation}\label{rho_rho0_v2}
\rho(v_A, \theta^1, \theta^2) = \rho_0(\theta^1, \theta^2) v_A^2.  
\end{equation}
Note that it is assumed in \cite{MR, PRL}, and here, that $\rho$ varies monotonically along each generator in $\cN$.
As is explained in \cite{PRL} and further on in the present section, this is not a severe restriction on the 
applicability of the formalism.   

In \cite{PRL} it was found that each of the fields $\mu$ and $\bar{\mu}$ Poisson commutes with itself, that is
\be \label{trivial_mu_brackets}
\{\mu(\mathbf{1}),\mu(\mathbf{2})\} = \{\bar{\mu}(\mathbf{1}),\bar{\mu}(\mathbf{2})\} = 0,
\ee
where $\mathbf{1}, \mathbf{2}$ are points on $\cN$, and also that data living on distinct branches of $\cN$ Poisson commute. 
Furthermore, it was found that the field $\rho_0$ Poisson commutes with itself and with $\mu$ and $\bar{\mu}$, from which it
follows that $\rho = \rho_0 v_A^2$ also Poisson commutes with itself and with $\mu$ and $\bar{\mu}$:
\be \label{trivial_rho_brackets}
\{\rho(\mathbf{1}),\rho(\mathbf{2})\} = \{\rho(\mathbf{1}),\mu(\mathbf{2})\} = \{\rho(\mathbf{1}),\bar{\mu}(\mathbf{2})\} = 0.
\ee
The only non-zero bracket between the fields $\mu$, $\bar{\mu}$ and $\rho$ is the one between 
$\mu$ and $\bar{\mu}$ at points $\mathbf{1}, \mathbf{2}$ on the same branch $\cN_A$ of $\cN$. It is
\be
\{\mu(\mathbf{1}),\bar{\mu}(\mathbf{2})\}
 = 4\pi G\, \dg^2(\theta_\mathbf{2} - \theta_\mathbf{1})\, H(\mathbf{1},\mathbf{2})
\left[\frac{1 - \mu\bar{\mu}}{\sqrt{\rho}}\right]_\mathbf{1}\left[\frac{1 - \mu\bar{\mu}}{\sqrt{\rho}}\right]_\mathbf{2}\:
e^{\int_\mathbf{1}^\mathbf{2} \frac{1}{1 - \mu\bar{\mu}}[\bar{\mu} d\mu - \mu d\bar{\mu}]}.
\label{mu_mubar_bracket}
\ee
The delta distribution in the bracket vanishes unless the points $\mathbf{1}$ and 
$\mathbf{2}$ lie on the same generator.  When they do lie on the same generator the integral 
in the exponential is evaluated along the segment of generator from $\mathbf{1}$ to $\mathbf{2}$, and  
$H(\mathbf{1},\mathbf{2})$ is a step function which equals $1$ when the point $\mathbf{1}$ 
lies on $S_0$ or between $S_0$ and the point $\mathbf{2}$, and vanishes otherwise. 
(To define the product of these factors with the delta as a distribution $H$ and the integral 
may be extended continuously to pairs of points $\mathbf{1}, \mathbf{2}$ lying on distinct generators. 
The product does not depend on the continuous extensions chosen.) 

The fields $\mu$, $\bar{\mu}$ and $\rho$ on $\cN$ generate a closed Poisson algebra in 
which $\rho$ commutes with everything. This algebra does not include the full set of initial data
- there are data which do not commute with $\rho$ - but in the present work we will concern ourselves only
with the problem of finding a quantization of the algebra generated by $\mu$, $\bar{\mu}$ and $\rho$.
In this context only the quantization of $\mu$ and $\bar{\mu}$ is non-trivial. The quantum commutators of 
$\rho$ with $\mu$, $\bar{\mu}$ and $\rho$ itself will be set to zero, as the Poisson brackets suggest.
$\rho$ is thus unchanged by the action, via Poisson bracket or commutator, of any functional of $\mu$, 
$\bar{\mu}$ and $\rho$. It can therefore be treated both in the classical and the quantum theory of this subalgebra 
of the data as a fixed, state independent function on $\cN$. 

Note that only data on the same generator have non-zero brackets. This is a
reflection of causality. Only points lying on the same generator are connected by a 
causal curve.\footnote{
This is always true in a spacetime neighborhood of any point of $\cN$, and we will 
{\em require} it to be true globally for the double null sheets that we consider. 
It is possible to immerse, or even embed, a double null sheet such that points on 
different generators are connected by a causal curve in the ambient spacetime. But 
then there is always an isometric covering spacetime in which they are not causally 
connected: It is always possible to embed the double null sheet into an isometric 
covering of part of the original spacetime, with the covering map mapping the image 
of $\cN$ in the covering spacetime into the image of $\cN$ in the original 
spacetime, such that distinct generators are not connected by any causal curve. See 
\cite{MR}. It follows that the hypothesis that the generators are causally 
disconnected does not restrict the initial data in any way.}

The bracket (\ref{mu_mubar_bracket}) has the curious feature that it does not 
strictly preserve the reality of the induced metric on $\cN$. There exist functions 
of $\mu$ and $\bar{\mu}$ which are real on real metrics, but nevertheless generate 
Hamiltonian flows from real metrics to metrics with a non-zero imaginary component.
However, this is more a nuisance than a real problem because the imaginary component  
generated always takes the form of a shock wave that travels along $\cN$, and does not affect 
the spacetime metric in the interior, $\cD$, of the domain of dependence of $\cN$, which 
remains real. The bracket therefore provides a Poisson structure on the space of real 
solution metrics on $\cD$. See \cite{PRL}. This awkward aspect of 
the formalism disappears when the deformed conformal metric $\E$ is used as data 
in place of $\mu$: $\E$ encodes the degrees of freedom of $\mu$ modulo, precisely, 
the shock wave modes mentioned. 

Because data on distinct generators Poisson commute the Poisson algebra decomposes, 
roughly speaking, into commuting subalgebras, formed by the data on each generator. 
Of course this is not quite correct because the Poisson bracket \eqref{mu_mubar_bracket} 
is a distribution which is singular precisely when $\mathbf{1}$ and $\mathbf{2}$ lie 
on the same generator, but ``morally`` it is true: if one replaces the Dirac delta in the bracket 
by a Kronecker delta times a normalization factor, as one might in a lattice model, then the 
algebra certainly decomposes as claimed. This suggests that we might learn a great 
deal about the quantization of the Poisson algebra (\ref{trivial_mu_brackets} - 
\ref{mu_mubar_bracket}) of $\mu$, $\bar{\mu}$ and $\rho$ by studying the quantization 
of the ``one generator algebra''
\begin{align}
\{\mu(\mathbf{1}),\mu(\mathbf{2})\} & = \{\bar{\mu}(\mathbf{1}),\bar{\mu}(\mathbf{2})\} 
= \{\rho(\mathbf{1}),\mu(\mathbf{2})\} = \{\rho(\mathbf{1}),\bar{\mu}(\mathbf{2})\} = \{\rho(\mathbf{1}),\rho(\mathbf{2})\} = 0 
\label{one_gen_trivial}\\
\{\mu(\mathbf{1}),\bar{\mu}(\mathbf{2})\} & = 4\pi G_2\, H(\mathbf{1},\mathbf{2})
\left[\frac{1 - \mu\bar{\mu}}{\sqrt{\rho}}\right]_\mathbf{1}\left[\frac{1 - \mu\bar{\mu}}{\sqrt{\rho}}\right]_\mathbf{2}\:
e^{\int_\mathbf{1}^\mathbf{2} \frac{1}{1 - \mu\bar{\mu}}[\bar{\mu} d\mu - \mu d\bar{\mu}]}
\label{one_gen_mu_mubar}
\end{align}
of fields $\mu$, $\bar{\mu}$ and $\rho$ on a line. This is just the algebra 
(\ref{trivial_mu_brackets} - \ref{mu_mubar_bracket}) with the delta distribution in 
$\theta_\mathbf{2}-\theta_\mathbf{1}$ removed, the points $\mathbf{1}$ and $\mathbf{2}$ 
restricted to the same generator, and a rescaled Newton's constant $G_2$ in place of $G$. 

It is also the Poisson algebra of $\mu$, $\bar{\mu}$ and $\rho$ on the double null sheet of 
figure \ref{Symmetry_adapted_N} in cylindrically symmetric gravity, provided $G_2$ is equal 
to $G$ divided by the coordinate area of $S_0$ in symmetry adapted $\theta$ coordinates.\footnote{
The $\theta$ coordinates are symmetry adapted if the derivatives $\di_{\theta^a}$ are Killing vectors generating 
the cylindrical symmetry. With such coordinates the area density $\rho$ is constant on each symmetry orbit, and 
$G_2 = G/\int_{S_0} d^2\theta$ satisfies $G_2/\sqrt{\rho(\mathbf{1})\rho(\mathbf{2})} 
= G\sqrt{A(\mathbf{1})A(\mathbf{2})}$, where $A(p)$ is the area of the intersection of $\cN$ with the symmetry orbit 
through $p$. The Poisson algebra (\ref{one_gen_trivial}, \ref{one_gen_mu_mubar}) is therefore independent of the 
choice of symmetry adapted $\theta$ coordinates, but, somewhat surprisingly, it does depend on the symmetry adapted 
double null sheet chosen. This does not imply any ambiguity in the classical theory, because a rescaling of the 
brackets by a common factor corresponds to a rescaling of the action, which does not affect the classical solutions.
It does, however, seem to mean that the cylindrically symmetric quantum theory is not unambiguously defined by the full four 
dimensional quantum theory. For instance, consider the intersection $S$ of $\cN$ with the cylindrical symmetry orbit 
of circumference $10^6$ Planck lengths. The Poisson bracket \eqref{one_gen_mu_mubar} suggest that in a coherent state
the quantum uncertainty in the components of the conformal metric on $S$ is of the order of one over the root of the 
area of $S$ in Planck units. That is, the cylindrically symmetric quantum theory depends on the choice of $\cN$. This 
ambiguity is not unreasonable: The space of classical solutions has a well defined subspace of cylindrically symmetric solutions, but 
the states of the quantized cylindrically symmetric theory, in which the non symmetric modes of the initial data 
are strictly zero, is presumably not contained in the space of states of the full theory, in which all modes are 
expected to realize at least vacuum fluctuations.}
We can therefore use the known results on the quantization of cylindrically symmetric gravity, in particular those 
of Korotkin and Samtleben \cite{KorotkinSamtleben}, to quantize the one generator algebra
(\ref{one_gen_trivial}, \ref{one_gen_mu_mubar}).

\begin{figure}
\centering
\includegraphics{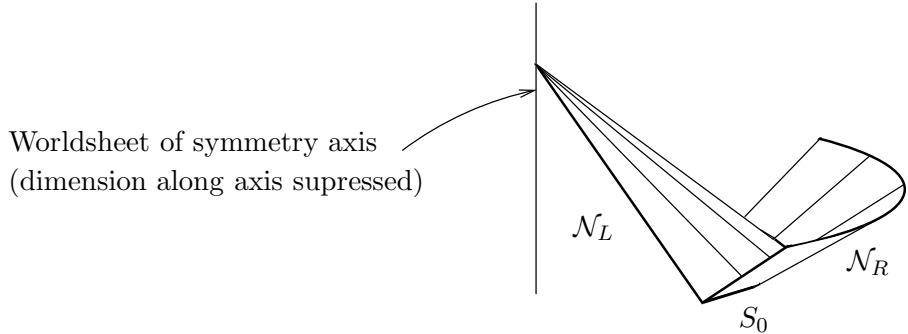}
\caption{A double null sheet adapted to a cylindrically symmetric spacetime, that 
is, a spacetime admitting a family of isometries the orbits of which are spacelike 
cylinders. The direction of the translation symmetry has been suppressed in the 
figure, reducing the spacetime to $2+1$ dimensions, and the symmetry orbits to 
circles. The central 2-surface $S_0$ of the adapted double null sheet is a portion 
of a symmetry orbit. If the charts $x_A, \theta^1, \theta^2$ are adapted to the 
cylindrical symmetry, in the sense that $\di/\di\theta^1$ and $\di/\di\theta^2$ are 
Killing vectors generating this symmetry, then $\mu$ depends only on $x_A$. Actually,
this $\cN$ does not quite fit our definition of a double null sheet, because $\cN_L$ 
touches the symmetry axis, where the generators meet. A double null sheet in the strict 
sense can be obtained by removing a neighborhood of the axis from $\cN$.
This subtlety has no consequences here and we shall call $\cN$ a double null sheet.}
\label{Symmetry_adapted_N}
\end{figure}

By cylindrically symmetric gravity we mean here vacuum general relativity with two commuting spacelike 
Killing fields that generate cylindrical symmetry orbits. And, following \cite{KorotkinSamtleben} and tradition, 
we add the requirement that the Killing orbits are orthogonal to a family of 2-surfaces. This apparently stringent 
additional condition is actually enforced by the vacuum field equations provided only two numbers, the so called 
{\em twist constants}, vanish. See \cite{Wald} Theorem 7.1.1. and \cite{Chrusciel}.  

Korotkin and Samtleben do not quantize the datum $\mu$, but they do quantize, 
among other things, the monodromy matrix $\M$ we have already mentioned, which is essentially the same as the deformed 
conformal metric $\E$. These encode the same physical degrees of freedom as $\mu$. In the following sections we will 
express $\E$, and $\M$, in terms of $\mu$, and verify that (\ref{one_gen_trivial}, \ref{one_gen_mu_mubar}) indeed implies 
the Poisson algebra of $\M$ that Korotkin and Samtleben quantize. 

The Poisson algebra of $\rho$ and $\mu$ in cylindrically symmetric gravity can be shown to coincide with the one generator algebra 
(\ref{one_gen_trivial}, \ref{one_gen_mu_mubar}) either by making a Poisson reduction of the Poisson algebra of null 
initial data in full four dimensional general relativity given in \cite{PRL}, or by calculating the Poisson brackets 
from the Einstein-Hilbert action restricted to cylindrically symmetric metrics, using a method analogous to that of 
\cite{MR, PRL}. Here we will do neither, because it is not necessary. The coincidence of the one generator Poisson 
algebra with that of cylindrically symmetric gravity certainly motivates the definition of the transformation 
$\mu \mapsto \E$ but the ultimate justification of this definition is that it transforms the one generator algebra 
into the algebra quantized in \cite{KorotkinSamtleben}, and this we verify directly. 

The model quantized in \cite{KorotkinSamtleben} is restricted by some further conditions, beyond cylindrical symmetry. It 
is assumed that spacetime becomes flat (locally) as one travels away from the symmetry axis, and that certain regularity
conditions hold at the symmetry axis. We will of course not put any conditions on the field at infinity, we cannot because 
$\cN_L$ doesn't reach infinity. But we will impose regularity conditions at the axis, namely that the area density on the 
symmetry orbits, $\rho$, vanishes at the axis, and that the limit of $e$ as the axis is approached along $\cN_L$ is well defined. 
Indeed, our basic definition of the transformation $\mu \rightarrow \E$ supposes that $e$ has a limiting value at the axis.
Nevertheless, in our calculation of the Poisson brackets we will need to treat {\em variations} $\dg$ about regular solutions 
for which $\dg e$ is singular at the axis. For this reason we extend the definition of the map $\mu \rightarrow \E$ to some 
fields that are singular at the axis.

It will also be assumed that $\rho$ and $\V$ are smooth on $\cN_L$, and that $\rho$ increases monotonically along the generators 
of $\cN_L$ as one moves away from the axis. (Note that in our figures and descriptions it will be assumed, for definiteness, 
that $d\rho$ is spacelike, and thus that the worldsheet of the symmetry axis is Lorentzian. This assumption is not required 
for our results.) 

These regularity conditions do not limit the scope of applicability of our results nearly as much as one might think. In solutions the 
monotonicity of $\rho$ on the generators of $\cN_L$ is largely a consequence of the field equations: In cylindrically symmetric 
vacuum solutions that are regular off the axis the Raychaudhuri equation implies that $\rho$ has at most one maximum, it either 
increases monotonically from zero at the axis forever or it reaches a maximum value and then decreases to zero in a finite affine 
distance. $\rho$ thus increases monotonically at least in a neighborhood of the axis. In the regular and asymptotically flat 
solutions that are the main focus of \cite{KorotkinSamtleben} it must be monotonic on all $\cN_L$ because $d\rho$ is non-zero 
and spacelike everywhere in spacetime. 

The regularity of $e$ at the axis is a stronger condition. Generically the conformal metric is not well defined at the axis. 
For instance, in flat spacetime the conformal metric of the cylindrically symmetric double null sheet of figure \ref{Symmetry_adapted_N} 
is singular at the axis. But also this condition restricts the applicability of the results less than it would seem to. Recall that we are 
studying cylindrically symmetric data not as an end in itself, but as a means to understand the one generator Poisson algebra, and 
ultimately the full algebra 
(\ref{trivial_mu_brackets} - \ref{mu_mubar_bracket}) in an arbitrary spacetime. Our results apply to the algebra 
(\ref{trivial_mu_brackets} - \ref{mu_mubar_bracket}) on any double null sheet for which the conditions on $\rho$ and $e_{ab}$ are 
satisfied on each generator.  Such double null sheets are certainly not generic, but there seem to be enough of them to describe 
any smooth solution to the vacuum field equations completely in terms of initial data on them. 

A double null sheet satisfying our conditions can be constructed from past light cones of regular points in any vacuum 
solution. If suitable coordinates $\theta^a$ are used to label the generators of a light cone, then the conformal metric with 
respect to these coordinates will be finite at the vertex: For instance, if $u^\ag$ are Riemann normal coordinates about the 
vertex, with $u^0$ timelike, and $\theta^a = u^a/(u^0 - u^3)$ ($a = 1,2$) then $e_{ab} = \dg_{ab}$.  Furthermore,  $\rho$ 
vanishes at the vertex and, if the cone is truncated close enough to the vertex, varies monotonically along the generators. 
The double null sheet can be constructed from the truncated past light cones of two regular points, provided these truncated 
cones intersect. Simply take $S_0$ to be a disk in the intersection, then the generators of the light cones that connect $S_0$ 
to the vertices sweep out the double null sheet.\footnote{
Note that the branch $\cN_L$ of the symmetry adapted double null sheet of figure \ref{Symmetry_adapted_N} is {\em not} a 
portion of a lightcone. The caustic at the axis is a line, not a point. This is the reason $e$ cannot have a well defined 
limit there in flat spacetime.}
Data on such double null sheets suffice to describe a solution if every spacetime point lies in the interior of the domain 
of dependence of some double null sheets of this type. This is clearly true in flat spacetime so, since it is essentially a 
local statement, it ought to be true also in curved spacetime.   

Of course it may nevertheless be interesting to generalize the cylindrically symmetric formalism to the case in which $e_{ab}$ is 
singular at the axis. This seems to be possible. As already mentioned the transformation $\mu \rightarrow \E$ can be extended 
easily to some fields for which $e$ is singular at the axis. More singular fields can perhaps be treated using the so called monodromy 
data of Alekseev \cite{Alekseev_summary} which is well defined when the axis is singular.

Note that although we have defined both branches, $\cN_L$ and $\cN_R$, of the double null sheet adapted to cylindrical 
symmetry, in the remainder of the paper we will only concern ourselves with the data on the branch $\cN_L$ swept out by 
generators going into the symmetry axis.

\section{The transformation to new variables.}\label{transformation}

In the present section we will define the map from the Beltrami coefficient $\mu$ to the deformed conformal metric $\E$ on 
$\cN_L$. $\E$ is a real, symmetric $2 \times 2$ matrix of determinant $1$, like the conformal 2-metric $e$. In fact it turns 
out that in cylindrically symmetric vacuum solutions satisfying our regularity conditions $\E$ at a point $r \in \cN_L$ 
equals $e$ on the axis at a certain instant of time $t_r$ determined by $r$ \cite{KorotkinSamtleben}. (See section \ref{KS_variables}.)

Figure \ref{Symmetry_reduced2} shows the symmetry reduced spacetime. In suitable coordinates the metric components of 
cylindrically symmetric solutions are constant on the symmetry orbits, and can therefore be thought of as functions on the
quotient of spacetime by these symmetry orbits. This quotient, a two dimensional manifold with boundary, is the reduced 
spacetime. The boundary is the worldline of the image of the symmetry axis in this spacetime. The branch $\cN_L$ of the 
adapted double null sheet is mapped to a line segment, which we will also call $\cN_L$. The symmetry axis at the instant 
$t_r$, a line in the full spacetime, corresponds to a point in the reduced spacetime which lies at the intersection of the 
past lightcone of $r \in \cN_L$ and the axis worldline. 

\begin{figure}
\begin{center}
\epsfig{file=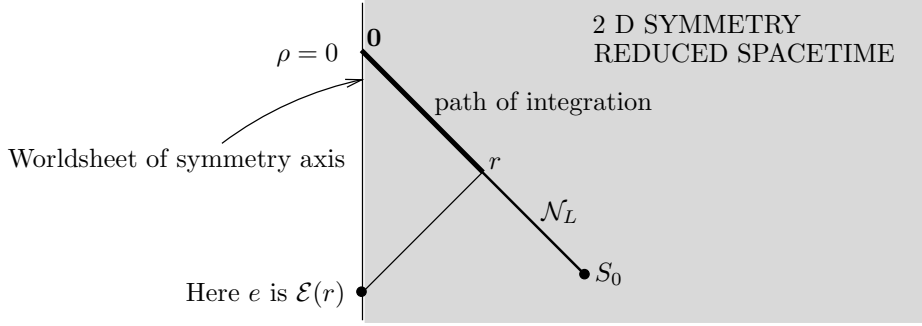}
\caption{The figure shows the two dimensional symmetry reduced spacetime, in which each point corresponds to a 
cylindrical symmetry orbit in the original spacetime. The vertical line is the worldline of the symmetry axis, and the 
boundary of the reduced spacetime. $\cN_L$ is a diagonal (null) line segment. A point $r$ on $\cN_L$ and its past 
lightcone are indicated. $\E(r)$ equals the conformal metric at the point (instant) $t_r$ on the axis worldline where this 
worldline meets the past lightcone of $r$.}
\label{Symmetry_reduced2}
\end{center}
\end{figure}

We shall define the transformation  $\mu \mapsto \E$ via a chain of transformations 
$\mu \mapsto \V \mapsto \Vh \mapsto \E$ involving the intermediate fields $\V$ and $\Vh$.

The field $\V$ is a density weight $-\tfrac{1}{2}$, positively oriented, real zweibein for the conformal 2-metric:
\begin{equation}
\V_a{}^i \V_b{}^j \dg_{ij} = e_{ab}, \qquad \qquad \det
\V \equiv \frac{1}{2} \eg^{ab}\eg_{ij}\V_a{}^i \V_b{}^j = 1.
\end{equation}
Letters $i, j, ...$ from the middle of the alphabet denote {\em internal indices}, which label 
the elements of the zweibein viewed as a basis of the space $S$ of density weight $-\tfrac{1}{2}$ 
1-forms; $\eg$ is the antisymmetric symbol, with $\eg_{12} = 1$; and $\dg_{ij}$ is the Kronecker 
delta. $\V$ may also be viewed as a linear map from an internal vector space to $S$. Then $\dg$ 
is a Euclidean metric on the internal space, and the internal indices $i, j, ...$ refer to a 
basis in this space which is orthonormal with respect to $\dg$.   

If a reference unit determinant real zweibein $Z$ is chosen, then any other such 
zweibein can be expressed as $\V_a{}^j = Z_a{}^i \V_i{}^j$. The matrices $\V_i{}^j$ 
form the group $SL(2,\R)$. The choice of a reference zweibein is not necessary for 
any of our constructions, but it allows us to describe them in the language of Lie 
groups.

One zweibein corresponding to the conformal metric defined by $\mu$ via \eqref{mu_parametrization} is
\begin{equation}\label{V_mu_def}
 \V = \frac{1}{\sqrt{1 - \mu\bar{\mu}}}\frac{1}{\sqrt{(1-\mu)(1-\bar{\mu})}}
\left[\begin{array}{cc} 1 - \mu\bar{\mu} & -i(\mu - \bar{\mu})\\ 0 & (1-\mu)(1-\bar{\mu}) \end{array}\right].
\end{equation}
But this is not the only possibility. The conformal metric determines the zweibein 
only up to local rotations, that is, up to right multiplication by an arbitrary 
position dependent element $h_i{}^j$ of the group $SO(2)$. One way to fix this 
gauge freedom is to require $\V$ to be upper diagonal and of positive trace, as it 
is in \eqref{V_mu_def}.

To define $\Vh$ we define the connection 
\begin{equation}
J_{i}{}^j = (\V^{-1})_i{}^b d \V_b{}^j
\end{equation} 
on $\cN_L$, deform it, and then integrate the deformed connection. $\V$ at any point $p \in \cN_L$ can 
be recovered from $J$ and the initial value of $\V$ at the reduced spacetime point $\mathbf{0}$ where 
$\cN_L$ meets the axis by integrating $J$ along the segment of $\cN_L$ from the axis to $p$:
\begin{equation}\label{integral_for_V}
\V(p) = \V(\mathbf{0})\: \cP e^{\int_{\mathbf{0}}^{p} J},
\end{equation}
where $\cP$ indicates that the exponential is path ordered. $\Vh$ is obtained from the same integral 
by substituting the deformed connection for $J$. (Note that we are using an exponential ordered 
from left to right, with the lower limit of integration corresponding to the left, and the upper to the right. See
appendix \ref{path_ordered_exp}.)

The connection $J$ is a 1-form on $\cN_L$ valued in the Lie algebra $\sltR$, that is, in the trace free, real, $2 \times 2$ 
matrices. Let $P$ be the symmetric component $\tfrac{1}{2}(J + J^t)$ of $J$, and $Q$ the antisymmetric component 
$\tfrac{1}{2}(J - J^t)$. (These are readily defined using the internal Euclidean metric to raise and lower indices.) Then the 
deformed connection is defined to be 
\begin{equation}\label{Jh_def1}
\Jh(q;r) = Q(q) + \frac{1}{\sqrt{1 - \rho(q)/\rho(r)}} P(q),
\end{equation}
where the radix denotes the principal square root, with $\sqrt{1} = 1$ and branch cut along the negative real axis. 
$\Jh$ is an $\sltC$ valued 1-form on $\cN_L$ that depends on two arguments. The first argument, the {\em field point} 
$q \in \cN_L$, corresponds to the argument of the undeformed connection $J$. $\Jh$ is a 1-form field with respect to $q$. 
The second argument, the {\em deformation point} $r \in \cN_L$, parametrizes the deformation. $\Jh$ is real when $q$ 
lies between $r$ and the axis. 

The field $\Vh(p;r)$ is obtained by integrating $\Jh(\cdot;r)$ along the segment of 
$\cN_L$ from the axis to $p$, holding the deformation point $r$ fixed:
\begin{equation}\label{integral_for_Vh}
\Vh(p;r) = \V(\mathbf{0})\: \cP e^{\int_{\mathbf{0}}^{p} \Jh(\cdot; r)}.
\end{equation}
That is, one replaces $J$ by $\Jh(\cdot;r)$ in the integral \eqref{integral_for_V}, maintaining the same prefactor 
$\V(\mathbf{0})$.  Equivalently $\Vh$ is the solution on $\cN_L$ to the differential equation
\begin{equation}\label{aux_lin_sys}
d\Vh = \Vh\Jh  
\end{equation}
which equals $\V$ on the axis. (See proposition \ref{path_ordered_exponential_prop_derivative} of the appendix). 

The final step is to define the deformed conformal metric $\E$. This is simply the conformal metric corresponding to the 
zweibein field $\U(q) = \Vh(q;q)$ obtained by setting the deformation point equal to the field point in $\Vh$.\footnote{
$U_a^i$ is essentially the field $W_a^i$ studied in \cite{Niedermaier_Samtleben}.
} 
Thus
\begin{equation}
\E_{ab} = \U_a{}^i \dg_{ij} \U_b{}^j. 
\end{equation}

This completes the definition of the transformation $\mu \mapsto \E$. Let us now examine it in detail. First let us verify 
that $\U$ is well defined, real and of determinant $1$, like $\V$. Since $\rho$ increases monotonically along $\cN_L$ as 
one moves away from the axis, the function 
\begin{equation}\label{u_def}
u(q;r) = \frac{1}{\sqrt{1 - \rho(q)/\rho(r)}}
\end{equation}
is real for $q$ on the segment of $\cN_L$ between the axis and $r$, and it is finite everywhere on this segment except at 
$q = r$. It follows that $\Jh = Q + uP$ is finite real and trace free on the segment excluding the endpoint $q = r$, and 
thus that $\Vh(\cdot;r)$ is well defined, real and of determinant $1$ there. $u$ is singular at $q = r$, but because the 
singularity is integrable, $\Vh$ is well defined, real and of determinant $1$ also there: Since $\rho$ is monotonic and 
smooth it may be used as a chart on $\cN_L$. In terms of this chart 
\begin{equation}
\Jh = \left[Q_\rho + \frac{1}{\sqrt{1 - \rho/\rho(r)}}P_\rho \right] d\rho,
\end{equation}
where $P_\rho$ and $Q_\rho$ are the $\rho$ components of the 1-forms $P$ and $Q$. Since $\V$ is also smooth (if a 
smooth $SO(2)$ gauge is adopted) these components are continuous. $\Jh$ thus diverges as an inverse square root of 
$\rho$, which is of course integrable. Proposition \ref{path_ordered_exponential_prop} of the appendix then indicates that 
$\U(r) = \Vh(r;r) = \V(\mathbf{0})\:\cP e^{\int_{0}^{\rho(r)} [Q_\rho + \frac{1}{\sqrt{1 - \rho/\rho(r)}}P_\rho]d\rho}$ 
is well defined, and equal to the limit of $\Vh(q;r)$ as $q \rightarrow r$. This establishes that $\U$ is real and of 
determinant $1$ as claimed. As corollaries $\U_i{}^j = Z^{-1}{}_i{}^a \U_a{}^j$, like $\V_i{}^j$, lies in $SL(2,\R)$ and 
$\E$ is well defined, real, and of determinant $1$. 

For points $p \in \cN_L$ that lie beyond $r$, so that $\rho(p) > \rho(r)$, $u(p;r)$ is the root of a negative real number.
$u$ is therefore pure imaginary, and a branch must be chosen to define its sign. Once a branch is chosen $\Vh$ is well 
defined but lies in $SL(2,\Complex)$ rather than $SL(2,\Real)$. See section \ref{KS_variables}.

Under $SO(2)$ gauge transformations $\Vh$ transforms like $\V$: Recall that under 
such a transformation $\V$ is multiplied on the right by a position dependent 
$SO(2)$ matrix $h$. That is, $\V \mapsto \V h$. Thus
\begin{equation}
 J \mapsto h^{-1}\V^{-1} d(\V h) = h^{-1} J h + h^{-1} d h.
\end{equation}
Taking symmetric and antisymmetric parts one obtains
\begin{align}
 P & \mapsto h^{-1} P h, \label{P_transform}\\
 Q & \mapsto  h^{-1} Q h + h^{-1} d h. \label{Q_transform}
\end{align}
$P$ transforms as an $SO(2)$ tensor, while $Q$ transforms as an $SO(2)$ connection. 
It follows that $\Jh = Q + uP$ transforms exactly like $J$, that is, 
$\Jh \mapsto h^{-1} \Jh h + h^{-1} d h$. This in turn implies that
\begin{equation}\label{Vh_transformation}
 \Vh \mapsto \Vh h,
\end{equation}
as can be demonstrated either by substituting the transform of the connection $\Jh$ 
and zweibein $\V(\mathbf{0})$ into the integral \eqref{integral_for_Vh}, or by noting that 
$\Vh h$ satisfies the differential equation \eqref{aux_lin_sys} with the
transformed $\Jh$ and the transformed initial datum $\V(\mathbf{0}) h$. As a corollary 
\eqref{Vh_transformation} implies that $\E$ is $SO(2)$ gauge invariant. It therefore 
depends only on the conformal metric $e$ (and $\rho$), and not on the zweibein
$\V$ chosen to represent $e$.

We have assumed that $\V$ is regular at the axis, but in fact the action of the Poisson bracket will in general not preserve this condition.
To define the Poisson bracket on $\Vh$ we must therefore define $\Vh$ on a somewhat more general class of $\V$ fields including some that are singular 
at the axis. Instead of defining $\Vh(p)$ as $\V(\mathbf{0})$ parallel transported to $p$ with the deformed connection $\Jh$, as in \eqref{integral_for_Vh}, 
one may define it as $\V(p)$ parallel transported to $\mathbf{0}$ with the undeformed connection $J$, and then parallel transported back to $p$ with the deformed
connection $\Jh$:
\begin{equation}\label{integral_for_Vh2}
\Vh(p;r) = \V(p)T_0(p,\mathbf{0})T(\mathbf{0},p)
\end{equation}
where $T_0(p,q) = \cP e^{\int_p^q J}$ and $T(q,p) = \cP e^{\int_q^p \Jh(\cdot; r)}$. This, by itself, does not extend the definition of $\Vh$ at all, but using 
proposition \ref{gauge_transformation} the expression \eqref{integral_for_Vh2} can be put into a form that is easily extended to the singular $\V$ fields in question 
provided both the field point $p$ and the deformation point $r$ lie off the axis. If one puts $A = \Jh(\cdot; r)$, $\lam = J$,
$a = \mathbf{0}$, $b = p$, and $\Lambda(a) = T_0(p,\mathbf{0})$ in the proposition, so that $\Lambda(q) = T_0(p,q) = \V^{-1}(p)\V(q)$, then the proposition shows that  
\begin{align}
\Vh(p;r) & = \V(p)\: \cP e^{\int_{\mathbf{0}}^{p} T_0(p,z)(\Jh(z; r) - J(z))_z T_0(z,p) dz}\\
& = \V(p)\: \cP e^{\int_{\mathbf{0}}^{p} \V^{-1}(p)\V(z)(\Jh(z; r) - J(z))_z \V^{-1}(z)\V(p) dz}\\
& = \cP e^{\int_{\mathbf{0}}^{p} \V(z)(\Jh(z; r) - J(z))_z \V^{-1}(z)dz}\:\V(p)\\
& = \cP e^{\int_{\mathbf{0}}^{p} (u - 1)\V P_z \V^{-1}dz}\:\V(p).\label{integral_for_Vh3} 
\end{align}
This last expression is our extended definition of $\Vh$. By proposition \ref{path_ordered_exponential_prop} it is well defined whenever $\V(p)$ is defined and $(u - 1)\V P_z \V^{-1}$ 
is integrable on the interval from $\mathbf{0}$ to $p$. If $p$ and $r$ lie off the axis this includes some cases in which $\V$ diverges at the axis, since $u - 1$ vanishes there. 
In particular it defines $\Vh$ on a large enough family of fields to determine the Poisson brackets of $\Vh$ at regular $\V$ fields. Presumably it also suffices to define the 
brackets at some singular $\V$ fields but that will not be explored in the present work. We will always assume that $\V$ is regular at the axis. Only the {\em variations} of 
$\V$ will be allowed to be singular there. 

It might seem that a phase space including only $\V$ fields that are regular at the axis would not be closed under the action of the Poisson bracket, but actually it is, in a 
roundabout way. The variations of $\V$ generated via the Poisson bracket differ from regular variations at most by what we call {\em zero modes}. These are the shock waves mentioned 
in section \ref{data0} that travel along $\cN$ but do not propagate into the interior of the domain of dependence of $\cN$. It is natural to take as the phase space the initial data 
on $\cN$ modulo zero modes. Then the Poisson bracket does not really take us out of the phase space corresponding to regular $\V$ fields. See \cite{PRL} for some related discussion.

\subsection{Coset space non-linear sigma models}\label{CSNLSM}

At each point the Beltrami coefficient $\mu$, or equivalently the conformal metric 
$e$, defines the matrix $\V_i{}^j$ in $SL(2,\R)$ up to right multiplication by an 
$SO(2)$ element. It can thus be identified with an element of the coset space 
$SL(2,\R)/SO(2)$. This suggests that cylindrically symmetric vacuum GR can be 
formulated as a coset space non-linear sigma model. Indeed this is the case. It is  
an $SL(2,\R)/SO(2)$ sigma model coupled to a dilaton and two dimensional gravity 
\cite{Geroch1}\cite{Breitenlohner_Maison_Gibbons}.

This form of the theory of cylindrically symmetric GR generalizes fairly directly 
to cylindrically symmetric reductions of a wide class of field theories, including 
electromagnetism coupled to gravity and various supergravity theories 
\cite{Breitenlohner_Maison_Gibbons}. In these models the field $\V$ takes values in 
some non-compact, connected, real, semi-simple matrix\footnote{A matrix group is
one that has a faithful finite dimensional matrix representation.} Lie group $G$ 
instead of $SL(2,\R)$, and the $SO(2)$ symmetry is replaced by a gauge symmetry 
under right multiplication by a field valued in the maximal compact subgroup $H$ of 
$G$. Although we will only study the vacuum gravity case we will often use aspects 
the formalism of the this wider class of models to clarify the logic.

This formalism is based on a few facts about semi-simple Lie groups: (See \cite{Helgason} 
for a systematic exposition of these ideas.)
The maximal compact subgroup $H$ can always be viewed as the subgroup of elements 
of $G$ invariant under an involutive automorphism $\eta$.\footnote{This follows 
from theorem 1.1 Ch. VI of \cite{Helgason}, and the fact that semi-simple matrix 
Lie groups have finite center, by proposition 4.1 Ch. XVIII of \cite{Hochschild}.}
For example, $SO(2)$ is the subgroup of $SL(2,\R)$ invariant under $g \mapsto 
g^\eta = (g^{-1})^t$. The automorphism $\eta$ on $G$ defines an automorphism of the 
algebra $\fg$, which will also be called $\eta$. In the case of $\slt$ it is just 
minus the transpose: $a^\eta = -a^t \ \forall a \in \slt$. The algebra $\fh$ of 
the maximal compact subgroup $H$ is of course invariant under $\eta$. In 
particular, $\sot$ consists of the antisymmetric $2 \times 2$ matrices.

Since any matrix can be decomposed into a sum of its antisymmetric and symmetric 
parts, the space of trace free matrices, $\slt$, is a sum of the space of 
antisymmetric matrices $\sot$ and the space of trace free symmetric matrices. This 
generalizes to the so called {\em Cartan decomposition} of $\fg$:
\begin{equation}
 \fg = \fh + \fk,
\end{equation}
with $\fh = \tfrac{1}{2}(\fg + \fg^\eta)$ and $\fk = \tfrac{1}{2}(\fg - \fg^\eta)$ 
being the eigensubspaces of $\eta$ corresponding to eigenvalues $1$ and $-1$ 
respectively. Occasionally $G$, $H$, $\fg$, $\fh$ and $\fk$ will denote the corresponding 
complexified objects, which are characterized by $\eta$ in the same way as the real ones.

In the general $G/H$ models mentioned $P = J_\fk = \tfrac{1}{2}(J - J^\eta)$, $Q = J_\fh = \tfrac{1}{2}(J + J^\eta)$, 
and $\Jh = Q + u P$, as in the vacuum gravity case; $\Vh$ and $\U$ are defined in the same way, in terms of 
$\V(\mathbf{0})$ and $\Jh$, as in the vacuum gravity case; and the $H$ gauge invariant field 
$\E$, analogous to the deformed conformal metric, is  $\E = \U ({\U}^\eta)^{-1}$. 
(Subscripts $\fh$ and $\fk$ indicate components in the subspaces of $\fg$ of the same name.)

It is worth noting that the Lie brackets of $\fh$ and $\fk$ always satisfy the 
following conditions:
\begin{equation}\label{hk_brackets}
[\fh,\fh] \subset \fh \qquad [\fk,\fh] \subset \fk \qquad [\fk,\fk] \subset \fh.
\end{equation}
The first relation simply confirms that $\fh$ is a subalgebra. All three relations 
are easily obtained by applying the involutive automorphism $\eta$ to the left side 
of each. For example
$[\fk,\fh]^\eta = [\fk^\eta,\fh^\eta] = - [\fk,\fh]$, so $[\fk,\fh]$ is contained 
in the $\eta$ eigensubspace of eigenvalue $-1$, namely $\fk$.

\section{Relation to the variables of Korotkin and Samtleben}\label{KS_variables}

The variables used by Korotkin and Samtleben in \cite{KorotkinSamtleben} differ slightly from the ones we use. Instead of 
the deformation point $r$ they use the {\em spectral parameter} $w = 2\rho(r) - \rho^+$ to parametrize the deformation, 
where $\rho^+$ is a real constant on $\cN_L$ which may be set to any desired value. Since $\rho$ is monotonic along $\cN_L$
the value of $w$ determines $r$ uniquely. In \cite{KorotkinSamtleben} the deformed zweibein $\Vh$ is a function of the 
field point and $w$, while the deformed conformal metric is replaced by the monodromy matrix $\M(w) = \E(r(w))$. $\rho$ is 
not a dynamical variable in their model, that is, the function $\rho$ on $\cN_L$ does not depend on the state of the system,
so the replacement of the deformation point $r$ by the spectral parameter is quite trivial. In particular, the Poisson 
brackets, and quantum commutators, of $\E$ can be read off from those of $\M$ by simply replacing $w$ by 
$2\rho(r) - \rho^+$ in the expressions for the latter, and vice versa.

As we have seen $\rho$ is effectively non-dynamical also in the $\mu$, $\bar{\mu}$, $\rho$ algebra of our model of 
cylindrically symmetric gravity. There is in fact a datum ($\lam$ in \cite{PRL}) which has non-zero bracket with $\rho$
but it is not included in the subalgebra of data that we study. In a similar way $\rho$ is non-dynamical in \cite{KorotkinSamtleben} 
because the the action is truncated, eliminating terms involving a degree of freedom ($\Gamma$ in \cite{KorotkinSamtleben}) which does not
Poisson commute with $\rho$. The models may be thought of as partial descriptions of cylindrically symmetric gravity, describing 
most of the degrees of freedom.
Alternatively, they may be thought of as complete descriptions of cylindrically symmetric gravity with regularity conditions
at the symmetry axis which eliminate the degree of freedom which fails to Poisson commute with $\rho$.

A further difference between our formalism and that of Korotkin and Samtleben is that in theirs the value of $\Vh$ at 
spatial infinity plays a key role. To define this limiting value we extend the definitions of $\V$, $J$, $\Jh$, and $\Vh$ 
from $\cN_L$ to the whole reduced spacetime: $\V$ is now a zweibein for the conformal metric on the cylindrical 
symmetry orbits in all spacetime, $J$ is $\V^{-1}d\V$, and $\Jh$ is a deformation of $J$ with components 
$\Jh_\pm(\cdot;w) = Q_\pm + u^{\mp 1} P_\pm$ in null coordinates $x^\pm$. Here the definition of $u$ has been generalized 
to
\begin{equation}\label{u_defw}
u = \sqrt{\frac{w + \rho^+}{w - \rho^-}}, 
\end{equation}
where $\rho^+$ and $\rho^-$ are the inward moving and outward moving components of $\rho$ respectively. In cylindrically 
symmetric solutions to the vacuum field equations (with vanishing twist constants) $\square \rho = 0$ on the reduced 
spacetime (\cite{Wald} eq. 7.1.21), so $\rho$ takes the form $\frac{1}{2}(\rho^+ + \rho^-)$ where $\rho^+$ is constant on 
ingoing null curves (moving toward the axis as time advances) while $\rho^-$ is constant on outgoing null curves. This
of course means that $\rho^+$ is a real constant on $\cN_L$.

On $C^2$ solutions $\Jh(\cdot;w)$ defined in this way turns out to be a flat connection for any value of $w$.\footnote{
Conversely, if the connection is flat for all $w$ then $\V$ satisfies the field equations. Thus the flatness of $\Jh$ is 
equivalent to the field equations on $\V$. The existence of such a zero curvature formulation of the field equations is 
characteristic of integrable field theories.}
$\Vh(q;w)$ may therefore be defined by an integral like \eqref{integral_for_Vh} taken along any curve from $\mathbf{0}$ 
to the field point $q$. Which curve is used does not matter since the connection is flat. Note that with the definition 
\eqref{u_defw} $\Jh$ and $\Vh$ are defined also for complex spectral parameter $w$.

\begin{figure}
\centering
\includegraphics[height=3cm]{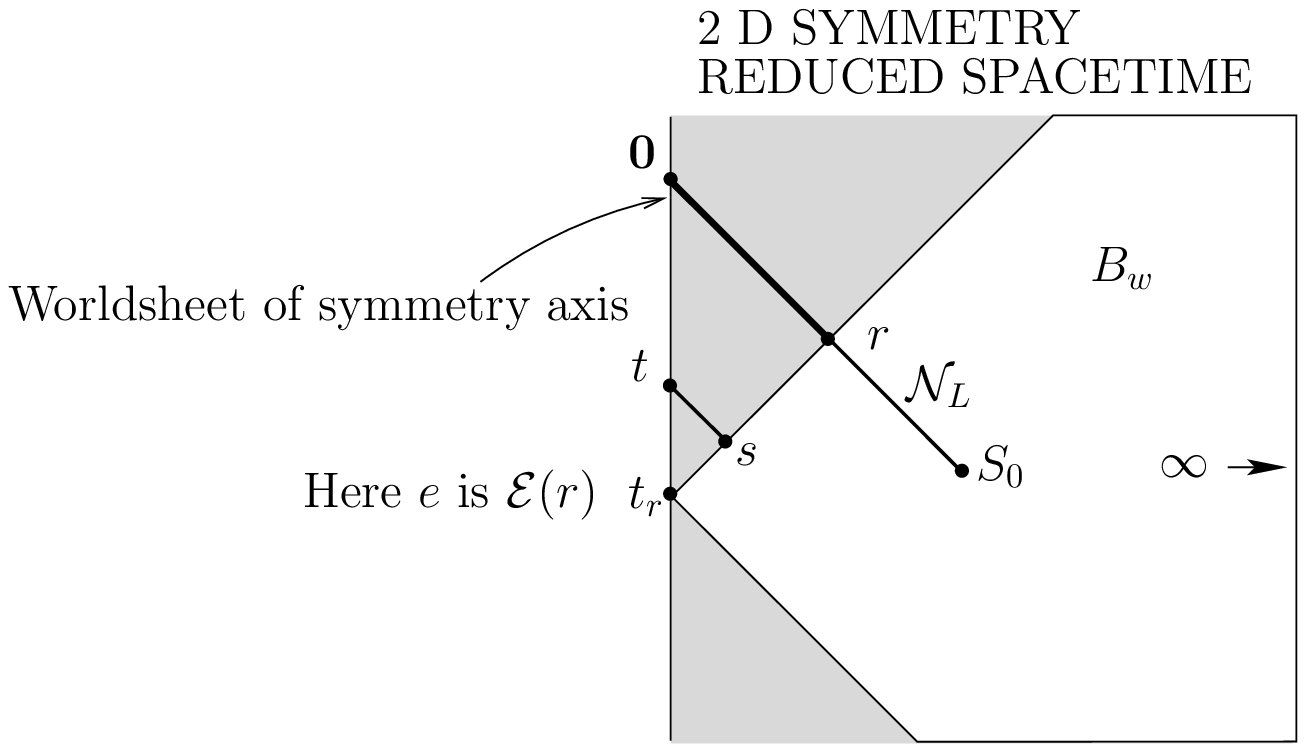}
\caption{This figure, like figure \ref{Symmetry_reduced2}, shows the two dimensional symmetry reduced spacetime. $u$ is real
in the shaded region, whereas it is purely imaginary in its unshaded complement. The unshaded region is therefore
the spacetime domain $B_w$, in which $w = \rho^-(r)$ lies on the branch cut of the function $u$ defined by \eqref{u_defw}.}
\label{branchcutdomain}
\end{figure}

Spatial infinity in \cite{KorotkinSamtleben} is characterized by $\rho^+ \rightarrow \infty$, $\rho^-/\rho^+ \rightarrow 1$. 
($d\rho$ is assumed to be spacelike throughout spacetime.)
The limit $\Vh(\infty;w)$ of $\Vh$ at spatial infinity is defined if there exists a sequence of reduced spacetime points 
such that $\rho^+ \rightarrow \infty$, $\rho^-/\rho^+ \rightarrow 1$ along the sequence, and $\Vh(\cdot;w)$ tends to the 
same limit along all such sequences. For $w$ real $\Vh(\infty;w)$ is actually double valued since $u$ is the principal root 
of a negative real number when $\rho^- > w$ and $\rho^+ > -w$. Korotkin and Samtleben therefore define
\begin{equation}\label{Tpmdef}
 T_{\pm}(w) = \Vh(\infty;w \pm i0) \V^{-1}(\infty),
\end{equation}
which are the central objects in their analysis. Here $\V(\infty)$ is the limit of the zweibein $\V$ at spatial infinity, 
assumed to exist, $\Vh(\cdot;w \pm i0)$ represents the limit $\lim_{\eg\rightarrow 0, \eg >0} \Vh(\cdot;w \pm i\eg)$, and 
$\Vh(\infty;w \pm i0)$ is the limit at spatial infinity of $\Vh(\cdot;w \pm i0)$. 

Note that in Minkowski space $e = \left[\begin{array}{rr} r & 0 \\ 0 & 1/r\end{array} \right]$ in standard cylindrical 
coordinates. $e$ thus has no finite limit either on the axis or at infinity, and of course a zweibein $\V$ of $e$ cannot 
then have finite limits either. Korotkin and Samtleben therefore do not work with asymptotically flat solutions directly, 
but rather with their Kramer-Neugebauer duals. (The Kramer-Neugebauer transformation is a symmetry 
transformation of the cylindrically symmetric vacuum gravity action. See \cite{Breitenlohner_Maison}.) In the 
Kramer-Neugebauer dual of Minkowski space $e = \left[\begin{array}{rr} 1 & 0 \\ 0 & 1\end{array} \right]$ in a suitable 
chart, so under any reasonable definition of asymptotically flat spacetimes that are regular at the axis $e$ has finite 
limits at both infinity and the axis in the Kramer-Neugebauer duals. And indeed Korotkin and Samtleben require that these 
limits exist in their model.

Korotkin and Samtleben quantize the Poisson algebra of $T_+$ and $T_-$, obtaining an algebra they term a ``twisted 
Yangian double'', closely related to Drinfel'ds Yangian algebra \cite{Drinfeld86}. Their quantization of the monodromy 
matrix is obtained by expressing $\M$ as a function of $T_\pm$:
\begin{equation}
\M(w)_{ab} = T_+(w)_{a}{}^c e(\infty)_{cd} T_-(w)_b{}^d 
= \Vh(\infty;w + i0)_a{}^i\dg_{ij}\Vh(\infty;w - i0)_b{}^j.
\end{equation}
(In \cite{KorotkinSamtleben} a basis in which $e(\infty)_{cd} = \dg_{cb}$ is used.) 

In the Kramer Neugebauer duals of asymptotically flat solution spacetimes that they consider this expression for the monodromy 
matrix agrees with our expression $\M(w)=\E(r(w))= \Vh(r(w);w)\Vh^t(r(w);w)$, a fact pointed out in \cite{Niedermaier_Samtleben}. 
In outline the proof runs as follows: The field
\begin{equation}
M(q;w) \equiv \Vh(q;w + i0)\Vh^t(q;w - i0). 
\end{equation}
is independent of the field point $q$ in the region $B_w$ of spacetime in which $w \in \Real$ lies on the branch cut of 
$u$ because
\begin{equation}\label{dM_is_0}
dM(\cdot;w) = \Vh(\cdot;w + i0)[\Jh(\cdot;w + i0) + \Jh^t(\cdot;w - i0)]\Vh^t(\cdot;w - i0) = 0, 
\end{equation}
The last equality holds when $w$ lies on the branch cut because then $u(\cdot;w + i0) = -u(\cdot;w - i0)$, and therefore 
$\Jh(\cdot;w + i0) = -\Jh^t(\cdot;w - i0)$. The spectral parameter $w$ lies on the branch cut of $u$ if and only if $w$ 
is real and $\rho^-(q) > w$, $\rho^+(q) > -w$, for then the radicand in the expression \eqref{u_defw} for $u$ is real and 
negative. If $w = \rho^-(r)$ for some deformation point $r \in \cN_L$ then $w$ lies on the branch cut at all points that 
are spacelike to the point $t_r$ on the axis that lies on the past light cone of $r$. See figure \ref{branchcutdomain}. 
$r$ of course lies on the boundary, $\di B_w$, of this region, as does in a sense, spacelike infinity. The constancy of 
$M(\cdot;w)$ in this region establishes the equality of Korotkin and Samtleben's expression for the monodromy matrix, 
equal to $M(\infty;w)$, with $M(r;w)$, which is equal to $\M(w)=\E(r)$ by continuity of $\Vh(\cdot;w \pm i0)$ along 
$\cN_L$ at $r$.   

To complete the proof two assumptions have to be justified: that the limit deformed zweibein $\Vh(\cdot;w \pm i0)$ really 
is continuous along $\cN_L$ at $r$, and that the limit connection $\Jh(\cdot;w \pm i0)$ is the connection corresponding to 
$\Vh(\cdot;w \pm i0)$ (and thus that $\Jh(\cdot;w \pm i0) = \Vh^{-1}(\cdot;w \pm i0)d\Vh(\cdot;w \pm i0)$). 

The component, $\Jh_-(\cdot;w \pm i0)$, of the limit connection along constant $\rho^+$ lines is well defined except at the
boundary $\rho^- = w$ of $B_w$. Outside $B_w$ it is just $\Jh_-(\cdot;w) = Q_- + \sqrt{\frac{w + \rho^+}{w - \rho^-}}P_-$, 
and inside $B_w$ it is $Q_- \pm i\sqrt{\left|\frac{w + \rho^+}{w - \rho^-}\right|}P_-$. Furthermore the norm 
$\norm{\Jh_-} = \sqrt{\Jh_{-\,i}{}^j\Jh_{-\,k}{}^l\dg^{ik}\dg_{jl}}$ of $\Jh_-(\cdot;w \pm i\eg)$ is bounded, for any 
$\eg>0$, by the function  
\begin{equation}
\norm{Q_-} + \left(\sqrt{\left|\frac{\operatorname{Re}{w} + \rho^+}{\operatorname{Re}{w} -\rho^-}\right|} + 1\right)
\norm{P_-}, 
\end{equation}
which is integrable along any finite segment of a constant $\rho^+$ line. Proposition 
\ref{path_ordered_exponential_prop2} then implies that if $p$ and $q$ are the end points of such a segment then the limit 
$T(p,q;w \pm i0)$ of the $\Jh(\cdot;w \pm i\eg)$ holonomy from $p$ to $q$ is the holonomy defined by 
the limit connection $\Jh(\cdot;w \pm i0)$. A similar argument applies to constant $\rho^-$ segments.

It follows immediately from proposition \ref{path_ordered_exponential_prop} and the integrability of 
$\Jh_-(\cdot;w \pm i0)$ that $\Vh(\cdot;w \pm i0)$ is continuous along $\cN_L$ at $r$. Propositions
\ref{path_ordered_exponential_prop} and \ref{path_ordered_exponential_prop_derivative} also imply that at any point $q$ off 
the line $\rho^- = w$ the limit connection satisfies $\Jh_-(q;w \pm i0) = T^{-1}(p,q;w \pm i0)\di_-T(p,\cdot;w \pm i0)|_q 
= \Vh^{-1}(q;w \pm i0)\di_-\Vh(\cdot;w \pm i0)|_q$, with $p$ any point at the same $\rho^+$ as $q$, and an analogous result 
for $\Jh_+(\cdot;w \pm i0)$. This establishes that $M(q;w) = M(r;w) = \M(w)$ for all finite points $q$ in $B_w$, from which  
it follows that the limiting value $M(\infty;w)$ also equals $\M(w)$.  

Incidentally it is now quite easy to demonstrate that on $C^2$ solutions $\E(r) = \M(w)$ equals the conformal metric $e$ at $t_r$ on 
the axis, provided the limiting value of $e$ on the axis exists and is differentiable along the axis worldline. 
Notice first that $\Jh = J$ along the axis, because $u = 1$ there, except at $t_r$. Thus $\Vh(t;w) = \V(t)$ at any point 
$t$ on the axis to the future of $t_r$. Now suppose that $s$ is the point of intersection of the past light cone of $t$ with 
$\di B_w$ then $\M(w) = M(s;w) = \V(t) T(t, s)T^t(t, s)\V^t(t)$. (See figure \ref{branchcutdomain}.) To show that 
$\M(w) = \V(t_r)\V^t(t_r) = e(t_r)$ it is therefore sufficient to show that $T(t, s)$ tends to $\One$ as $t \rightarrow t_r$. 
But by proposition \ref{path_ordered_exponential_prop} of the appendix 
\begin{equation}
\norm{T(t,s) - \One} \leq \exp\left(\int_{\rho^-(t)}^w \norm{\Jh_-}d\rho^-\right) -1 
\leq \exp\left(\int_{\rho^-(t)}^w \norm{Q_-} + \sqrt{\frac{w - \rho^-(t)}{w - \rho^-}} \norm{P_-}d\rho^-\right) -1 \rightarrow 0
\end{equation}
as $t \rightarrow t_r$.

\section{The Poisson brackets of the new variables}\label{brackets}

\newcommand{\Vhat}{\hat{\cal V}}
\newcommand{\Jhat}{\hat{J}}

To obtain the Poisson brackets of the deformed conformal metric $\E$, we proceed in 
steps, corresponding to those in the definition of the transformation 
$\mu \mapsto \E$. We begin in the following subsection by deriving the necessary 
components of the brackets of the zweibein $\V$ from those of $\mu$ and 
$\bar{\mu}$. Then, in subsection \ref{Vhat_brack} the brackets of the deformed 
zweibein $\Vhat$ are obtained from those of $\V$. The brackets of $\U$, the deformed zweibein 
evaluated at the deformation point, are calculated in subsection \ref{Vhat_brack}. Finally, in 
subsection \ref{section:brackofE}, the brackets of $\U$ are used to calculate the brackets of $\E$.

\subsection{The Poisson brackets of $\V$}\label{V_brack}

We shall calculate a certain block of components of the logarithmic bracket
\be	\label{Vlog_bracket}
\oo{\V}{}^{-1}(p_1)\ot{\V}{}^{-1}(p_2)\{\oo\V(p_1),\ot\V(p_2)\}.
\ee
Here a compact notation for tensor products has been used that will be employed extensively in the remainder of the paper: The 
tensor product $A\otimes B$ of a linear operator $A$ acting on a vector space $V_1$ and a linear operator $B$ acting on another 
vector space $V_2$ is denoted by $\oo A\ot B$, with the index over each factor indicating the space it acts on. In this way the 
bracket $\{\V_a{}^i(p_1), \V_b{}^j(p_2)\}$ may be written in index free notation as $\{\oo\V(p_1),\ot\V(p_2)\}$, or even 
$\{\oo\V,\ot\V\}$ when the arguments of $\oo\V$ and $\ot\V$ are clear from context.

The logarithmic bracket \eqref{Vlog_bracket} is an element of the tensor product of algebras $\oo\fg \ot\fg$. However, only the 
projection of the bracket on $\oo\fk \ot\fk$ will be needed to calculate the brackets of $\E$, which is our ultimate aim. $\E$ 
is a differentiable, $H$ gauge invariant functional of $\V$ and the brackets of such functionals depend only on the $\oo\fk \ot\fk$ 
component of \eqref{Vlog_bracket}: Let $F$ be such a functional, and let $\dg_{gauge}$ be any $H$ gauge variation, that is 
$\dg_{gauge}\V = \V a$ with $a$ any $\fh$ valued function on $\cN_L$, then
\be 
0 = \dg_{gauge} F = \int_{\cN_L} \frac{\dg F}{\dg \V_a{}^j}\V_a{}^i a_i{}^j dx.
\ee
(Here the variable of integration $x$ is a coordinate parameterizing $\cN_L$ and the functional derivative is taken with respect to 
$\V_a{}^j$ as a function of $x$.)
This implies that $\left[\frac{\dg F}{\dg \V}\right]^t \V$ traced together with any element of $\fh$ gives zero. As a consequence, 
when $\dg$ is an arbitrary variation
\be \label{delta_F}
\dg F = \int_{\cN_L} \tr (\left[\frac{\dg F}{\dg \V}\right]^t \V \V^{-1}\dg\V) dx =
\int_{\cN_L} \tr (\left[\frac{\dg F}{\dg \V}\right]^t \V [\V^{-1}\dg\V]_{\fk}) dx,
\ee
so only the $\fk$ component of $\V^{-1}\dg\V$ contributes to the variation of $F$.

To define the Poisson brackets of $\V$ it is necessary to express $\V$ as a function of $\mu$, that is, to fix the gauge. We will
calculate the bracket using symmetric gauge, in which $\V_a{}^i$ is a symmetric matrix of positive trace. But in the end, because
the gauge dependence of the $\oo\fk \ot\fk$ component of the logarithmic bracket is very simple, we can give an expression for it
valid in all gauges.

It will be convenient to work with a similarity transform of $\V$,
\be	\label{Vt_def}
\Vt = \frac{1}{2} \left[\begin{array}{rr} 1 & -i \\ 1 & i\end{array} \right] \V \left[\begin{array}{rr} 1 & 1 \\ i & -i\end{array}
\right].
\ee  
The elements of the columns of $\Vt$ are the components on the 1-forms $dz$ and 
$d\bar{z}$ (with $z = \theta^1 + i\theta^2$ as in \eqref{mu_parametrization}) of 
the complex null basis $\V^\pm = \V^1 \pm i \V^2$ formed from the orthogonal basis 
1-forms $\V^i = \V_a{}^i d\theta^a$. In terms of these components the line element 
on the cylindrical symmetry orbits may be expressed as 
\begin{align}
ds^2 = \rho \V_a{}^i \dg_{ij}\V_b{}^j d\theta^a d\theta^b  = & \rho \V_a{}^+ \V_b{}^- d\theta^a d\theta^b \\
  = & \rho (\V_z{}^+ dz + \V_{\bar{z}}{}^+ d\bar{z})(\V_z{}^- dz + \V_{\bar{z}}{}^- d\bar{z}).
\end{align}
This reproduces the expression \eqref{mu_parametrization} for the line element if 
\be	\label{Vt_mu}
\Vt = \frac{1}{\sqrt{1-\mu\bar{\mu}}} \left[\begin{array}{cc} 1 & \bar{\mu} \\ \mu & 1\end{array} \right].
\ee    
This is not the only possibility, but it is the one corresponding to $\V$ symmetric  with positive trace. Indeed, inverting the 
transformation \eqref{Vt_def} one obtains
\be	\label{Vsym_mu}
\V = \frac{1}{\sqrt{1-\mu\bar{\mu}}} \left[\begin{array}{cc} 1 + \frac{1}{2}(\mu + \bar{\mu}) & -\frac{i}{2}(\mu - \bar{\mu}) 
\\ -\frac{i}{2}(\mu - \bar{\mu}) & 1 - \frac{1}{2}(\mu + \bar{\mu}) \end{array} \right].
\ee
(Further transforming to upper triangular gauge one obtains \eqref{V_mu_def}.)  

Applying the similarity transformation \eqref{Vt_def} to the Pauli matrices one obtains $\widetilde{\sg_x}= \sg_y$,
$\widetilde{\sg_y}= \sg_z$, $\widetilde{\sg_z}= \sg_x$. The transform $\tilde{\fh}$ of the subalgebra 
$\fh = \mathfrak{so}(2)$ is thus generated by $i\sg_z$, and the $SO(2)$ gauge transformation becomes 
$\Vt \mapsto \Vt e^{i\phi \sg_z}$, while $\tilde{\fk}$, the transform of $\fk$, is spanned by $\sg_x$ and $\sg_y$, 
consisting therefore of Hermitian matrices that vanish on the diagonal.

We are now ready to compute the bracket. For any variation $\dg$ that preserves the symmetric gauge
\be	\label{Vt_variation}
\Vt^{-1}\dg\Vt = \frac{1}{2}\frac{\mu\dg\bar{\mu} - \bar{\mu}\dg\mu}{1-\mu\bar{\mu}} \sg_z + \frac{1}{1-\mu\bar{\mu}}
\left[\begin{array}{cc}0 & \dg\bar{\mu} \\ \dg\mu & 0  \end{array}\right].
\ee
Thus 
$[\Vt^{-1}\dg\Vt]_{\tilde{\fk}} = \frac{\dg\mu}{1-\mu\bar{\mu}} s_- +  \frac{\dg\bar{\mu}}{1-\mu\bar{\mu}} s_+$, where 
$s_\pm = \frac{1}{2}(\sg_x \pm i \sg_y)$, and it follows that 
\begin{align}     
[\oo{\Vt}{}^{-1}\ot{\Vt}{}^{-1}\{\oo\Vt,\ot\Vt\}]_{\tilde{\fk}\tilde{\fk}} = & \frac{1}{(1-\mu\bar{\mu})_{\bf 1} (1-\mu\bar{\mu})_{\bf 2}}
(\{\mu({\bf 1}),\bar{\mu}({\bf 2})\} \oo s_- \ot s_+ + \{\bar{\mu}({\bf 1}),\mu({\bf 2})\} \oo s_+ \ot s_-)\\
= & 4\pi G_2 \frac{1}{\sqrt{\rho_{\bf 1}\rho_{\bf 2}}}\left\{ H({\bf 1},{\bf 2}) e^{\int_{\bf 1}^{\bf 2}\frac{\bar{\mu}d\mu - \mu d\bar{\mu}}{1 
- \mu\bar{\mu}}}
\oo s_- \ot s_+ - H({\bf 2},{\bf 1}) e^{-\int_{\bf 1}^{\bf 2}\frac{\bar{\mu}d\mu - \mu d\bar{\mu}}{1 - \mu\bar{\mu}}}\oo s_+ \ot s_-\right\}.
\end{align}

Equation \eqref{Vt_variation} also shows that the $\tilde{\fh}$ component of the connection is
\be
\tilde{Q} = [\Vt^{-1}d\Vt]_{\tilde{\fh}} = \frac{1}{2}\frac{\mu d\bar{\mu} - \bar{\mu} d\mu}{1-\mu\bar{\mu}} \sg_z.
\ee      
The exponential 
$e^{\ag({\bf 1},{\bf 2})} = e^{\int_{\bf 1}^{\bf 2}\frac{\bar{\mu}d\mu - \mu d\bar{\mu}}{1 - \mu\bar{\mu}}}$ can 
therefore be reexpressed in terms of
$\cP e^{\int_{\bf 1}^{\bf 2}\tilde{Q}} = \left[\begin{array}{cc} e^{-\ag/2} & 0 \\ 0 & e^{\ag/2} \end{array}\right]$.
Indeed
\be
e^\ag s_- = e^{\ag/2} s_- e^{\ag/2} = \cP e^{\int_{\bf 1}^{\bf 2}\tilde{Q}} s_-\left[\cP e^{\int_{\bf 1}^{\bf 2}
\tilde{Q}}\right]^{-1} = \cP e^{\int_{\bf 1}^{\bf 2}\tilde{Q}} s_-\cP e^{\int_{\bf 2}^{\bf 1}\tilde{Q}},
\ee
and similarly $e^{-\ag} s_+ = \cP e^{\int_{\bf 1}^{\bf 2}\tilde{Q}} s_+\cP e^{\int_{\bf 2}^{\bf 1}\tilde{Q}}$.
(Of course, $\tilde{\fh} \simeq \sot$ is abelian, so the path ordering in $\cP e^{\int_{\bf 1}^{\bf 2}\tilde{Q}}$ 
is not really necessary.) Thus
\be\label{Vt_bracket1}
[\oo{\Vt}{}^{-1}\ot{\Vt}{}^{-1}\{\oo\Vt,\ot\Vt\}]_{\tilde{\fk}\tilde{\fk}} 
= 4\pi G_2 \frac{1}{\sqrt{\rho_1\rho_2}}\cP e^{\int_{\bf 1}^{\bf 2}\oo{\tilde{Q}}}\left\{ H({\bf 1},{\bf 2}) \oo s_- \ot s_+ 
- H({\bf 2},{\bf 1}) \oo s_+ \ot s_-\right\}\cP e^{\int_{\bf 2}^{\bf 1}\oo{\tilde{Q}}}.
\ee
(Here the superscript $1$ on the integrand $\oo{\tilde{Q}}$ does not imply that it is evaluated at the point $\bf 1$.)

This expression for the bracket can be given a more illuminating, gauge covariant, form.
\be
\oo s_- \ot s_+ = \frac{1}{4}(\oo\sg_x - i\oo\sg_y)(\ot\sg_x + i\ot\sg_y) = 2(\Omega_{\tilde{\fk}} + i\vareg_{\tilde{\fk}})
\ee
where $\Omega_{\tilde{\fk}} = \frac{1}{8}(\oo\sg_x\ot\sg_x + \oo\sg_y\ot\sg_y)$ and
$\vareg_{\tilde{\fk}} = \frac{1}{8}(\oo\sg_x\ot\sg_y - \oo\sg_y\ot\sg_x)$, and
$\oo s_+ \ot s_- = 2(\Omega_{\tilde{\fk}} - i\vareg_{\tilde{\fk}})$.
Furthermore, the step function $H$ is a sum of an odd step function and a constant:
\be	\label{step_expansion}
H({\bf 1},{\bf 2}) = \frac{1}{2}  s({\bf 1},{\bf 2}) + \frac{1}{2},
\ee 
where $s({\bf 1},{\bf 2})$ takes the value $1$ if the point $\bf 1$ lies on $S_0$ or between $S_0$ and the point $\bf 2$, 
and $-1$ otherwise. 
As a result
\be
H({\bf 1},{\bf 2}) \oo s_- \ot s_+ - H({\bf 2},{\bf 1}) \oo s_+ \ot s_- = 2 (s({\bf 1},{\bf 2}) \Omega_{\tilde{\fk}}
+ i \vareg_{\tilde{\fk}}).
\ee
The $\fk\fk$ component of the logarithmic bracket of the original zweibein $\V$ is obtained by acting on both sides of 
\eqref{Vt_bracket1} with the inverse of the similarity transformation \eqref{Vt_def}:
\be\label{V_bracket1}
[\oo{\V}{}^{-1}\ot{\V}{}^{-1}\{\oo\V,\ot\V\}]_{\fk\fk}
= 8\pi G_2 \frac{1}{\sqrt{\rho_1\rho_2}}\cP e^{\int_{\bf 1}^{\bf 2}\oo Q}\left\{ s({\bf 1},{\bf 2}) \Omega_\fk
+ i \vareg_\fk\right\}\cP e^{\int_{\bf 2}^{\bf 1}\oo Q},
\ee
where 
\begin{equation} \label{Omega_def_sl2R}
\Omega_\fk = \frac{1}{8}(\oo\sg_z\ot\sg_z + \oo\sg_x\ot\sg_x)  
\end{equation}
and
\begin{equation} \label{Vareg_def_sl2R}
\vareg_\fk = \frac{1}{8}(\oo\sg_z\ot\sg_x - \oo\sg_x\ot\sg_z).
\end{equation}

Although this expression for the $\fk\fk$ component of the bracket was calculated 
in symmetric gauge, it is actually valid in any gauge, because both sides transform 
in the same way under gauge transformations - by a similarity transformation:
Under a gauge transformation $\V \mapsto \V h$, with $h$ an $H = SO(2)$ valued 
field, and $[\V^{-1}\dg\V]_\fk \mapsto h^{-1}[\V^{-1}\dg\V]_\fk h$ for any 
variation $\dg$, even if $\dg$ acts non trivially on $h$. The left side of 
\eqref{V_bracket1} therefore transforms by a similarity transformation by 
$h({\bf 1})$ in the $1$ space and by $h({\bf 2})$ in the $2$ space. As to the right 
side, the gauge holonomy $\cP e^{\int_{\bf 1}^{\bf 2}\oo Q}$ transforms 
to $\oo{h}{}^{-1}({\bf 1})\cP e^{\int_{\bf 1}^{\bf 2}\oo Q}\oo h({\bf 2})$ while 
$\Omega_\fk$ and $\vareg_\fk$ are invariant under simultaneous similarity 
transformations of both space $1$ and space $2$ by the same $h \in SO(2)$. That is, 
$\Omega_\fk = \oo{h}{}^{-1}({\bf 2})\ot{h}{}^{-1}({\bf 2})\Omega_\fk\oo{h}({\bf 2})\ot{h}({\bf 2})$ 
and similarly for $\vareg_\fk$. It follows that the right side transforms like the left side.

The invariance of $\Omega_\fk$ follows from the fact that it commutes with the generator $\oo\eg + \ot\eg$ 
of $SO(2)$ similarity transformations, $\eg$ being the antisymmetric matrix 
$\left[\begin{smallmatrix} 0 & 1 \\ -1 & 0 \end{smallmatrix} \right]$ which generates $SO(2)$. 
This also proves the invariance of $\vareg_\fk = \frac{1}{2}[\oo\eg, \Omega_\fk]$. The deeper 
reason that  $\Omega_\fk$ is invariant is that  $\Omega_\fk$ is the inverse of the restriction 
of the Killing form of  $\fg$ to $\fk$. $\Omega_\fk$ is invariant under the 
adjoint action of $H$ because both the Killing form, and the subspace $\fk$ of $\fg$ are.  

For later use we define the notations $\Omega_\fg$ for the inverse of the Killing form of $\fg$ 
and $\Omega_\fh$ for the inverse of the restriction of the Killing form to $\fh$. Note that 
because $\fk$ and $\fh$ are Killing orthogonal 
$\Omega_\fg = \Omega_\fk + \Omega_\fh$. Note also that for $\fg = \slt$ and $\fh = \sot$ 
\begin{align}
 \Omega_{\fg\,i}{}^j{}_k{}^l & = \frac{1}{4}[\dg_i^l \dg^j_k - \frac{1}{2}\dg_i^j \dg_k^l]\\
 \Omega_{\fk\,i}{}^j{}_k{}^l & = \frac{1}{8}[\dg_i^l \dg^j_k + \dg^{jl} \dg_{ik} - \dg_i^j \dg_k^l]\\
 \Omega_{\fh\,i}{}^j{}_k{}^l & = \frac{1}{8}[\dg_i^l \dg^j_k - \dg^{jl} \dg_{ik}].
\end{align}
 
The definitions of $\Omega_\fg$, $\Omega_\fk$, and $\Omega_\fh$ as the inverses of 
restrictions of the Killing form can be applied to the Cartan decomposition of any 
semi-simple Lie algebra $\fg$. The real, $\Omega_\fk$, term in the bracket 
\eqref{V_bracket1} can thus be extended straightforwardly to the wider class of 
coset space sigma models mentioned in subsection \ref{CSNLSM}. The  
generalization of $\vareg_\fk$ is less obvious.  However, this object does not 
appear in the brackets of $\E$ or $\U$, calculated in the following subsections, 
so these brackets can be generalized without difficulty. 

The imaginary $\vareg_\fk$ term in the bracket is in fact its strangest aspect.
It arises because the step function $H({\bf 1},{\bf 2})$ in the bracket \eqref{one_gen_mu_mubar} is 
not antisymmetric. As pointed out in \cite{PRL}, this is the price one has to pay 
to obtain brackets of $\mu$ and $\bar{\mu}$ that satisfy the Jacobi relations.

Because of this imaginary term the bracket does not preserve the reality of the conformal metric $e$. That is, a real functional $F$ 
of $e$ can generate a Hamiltonian flow that takes real $e$ to complex $e$: The variation of $e$ on $\cN$ generated by such an $F$ is 
\be
\{F, e\} = \{F, \V\V^t\} = \{F,\V\} \V^t + \V \{F,\V\}^t = 2\V[\V^{-1}\{F,\V\}]_\fk \V^t,
\ee
so if $\V$ is real (which can always be assumed if $e$ is real) the imaginary part of this variation is $2\V\operatorname{Im}[(\V^{-1}\{F,\V\})_\fk] \V^t$.
Since $F$ is a functional only of $e$ it is $H$ gauge invariant. Thus, by the reasoning that led to \eqref{delta_F}, 
\begin{align} 
\operatorname{Im}(\ot\V{}^{-1}\{F, \ot \V\})_\fk & =  \operatorname{Im} \int_{\cN} \oo\tr (\oo{\left[\frac{\dg F}{\dg\V}\right]}^t\oo\V 
[\oo{\V}{}^{-1}\ot{\V}{}^{-1}\{\oo\V,\ot\V\}]_{\fk\fk})\, dx_1\\
	      & = \sum_{A \in {x,z}} c_A \frac{1}{\sqrt{\rho_2}}
\ot{\left[\cP e^{\int_{\bf 2}^{\bf 0} Q} \sg_A \cP e^{\int_{\bf 0}^{\bf 2} Q} \right]},\label{V_im_modes}
\end{align}
where $c_A = \tfrac{\pi}{2} G_2  \int_{\cN_L} \frac{1}{\sqrt{\rho}}\tr (\left[\frac{\dg F}{\dg\V}\right]^t\V\cP e^{\int_{\bf 1}^{\bf 0} Q}
[\eg,\sg_A] \cP e^{\int_{\bf 0}^{\bf 1} Q})\,dx_1$ if $\mathbf{2} \in \cN_L - S_0$, and is given by an analogous integral over $\cN_R$ if $\mathbf{2} \in \cN_R - S_0$. 
If $\bf 2$ lies on the intersection $S_0$ then $c_A$ is the sum of the integrals over $\cN_L$ and $\cN_R$. (Recall that only data living on the same generator
have non-zero brackets.) Since these coefficients are not zero in general a real functional $F$ of $e$ can, and in general does, excite two modes of the conformal 
metric with imaginary coefficients on each branch of $\cN$.

The imaginary component of $(\ot\V{}^{-1}\{F, \ot \V\})_\fk$ is the sum of initial data for inward moving shock waves, 
\begin{equation}\label{zero_mode_L}
\begin{array}{cl} \frac{1}{\sqrt{\rho_2}}
\ot{\left[\cP e^{\int_{\bf 2}^{\bf 0} Q} \sg_A \cP e^{\int_{\bf 0}^{\bf 2} Q} \right]} & \text{for $\mathbf{2} \in \cN_L$}\\
0 & \text{for $\mathbf{2} \in \cN_R - S_0$},\end{array}
\end{equation}
and for analogous outward moving shock waves which are non-zero on $\cN_R$. These are the zero modes mentioned earlier. They do not propagate into the interior of 
the domain of dependence of $\cN$, as can be verified directly from the field equations for cylindrically symmetric GR, as given, for instance, in 
\cite{NKS}.    
A more conceptual argument which establishes the same conclusion also in the symmetryless case is given in \cite{PRL} in terms of the corresponding modes of $\mu$ and $\bar{\mu}$. 

Since the zero modes do not propagate into the interior of the domain of dependence of $\cN$ they are not really 
part of the the initial data that determine the metric there, and it would be desirable to have data from which this mode has
been projected out. As we shall see, the deformed conformal metric is such data.

\subsection{The Poisson brackets of $\Vhat$}\label{Vhat_brack}

We turn now to the calculation of the Poisson bracket of $\oo\Vhat(p_1;r_1)$ with $\ot\Vhat(p_2;r_2)$ for field points $p_1$, $p_2$ 
and deformation points $r_1$, $r_2$ off the axis. As a first step a simple expression will be found for the variation of
$\Vhat$ due to a variation of the field $\cV$. (The variation of $\Vhat$ due to a variation of $\rho$ is not needed 
because $\Vhat$ Poisson commutes with $\rho$.)

Recall that when $\V$ is regular at the point $\mathbf{0}$ where $\cN_L$ meets the symmetry axis then $\Vhat$ at a field point $q$ on $\cN_L$ is
\be\label{Vhat_def_repetition}
\Vhat(q) = \cV(\mathbf{0})T(\mathbf{0},q), 
\ee
where 
\be
T(p,q) = \cP e^{\int_p^q \Jhat}
\ee
is the holonomy from $p$ to $q$ defined by $\Jhat$. 

However, as \eqref{V_bracket1} shows, the Poisson bracket $\oo{\V}{}^{-1}(p_1)\ot{\V}{}^{-1}(p_2)\{\oo\V(p_1),\ot\V(p_2)\}$ diverges as 
$p_1 \rightarrow \mathbf{0}$. The variation of $\V$ generated via the Poisson bracket by $\V(p_2)$ is singular at $\mathbf{0}$. For this 
reason we have provided a somewhat more sophisticated definition of $\Vhat$ in section \ref{transformation} which extends the definition 
\eqref{Vhat_def_repetition} to fields $\V$ having certain types of singularity at $\mathbf{0}$. According to this extended definition
\be\label{Vhat_def2_repeat}
\Vhat(q) = \cP e^{\int_{\mathbf{0}}^q (u-1)\V P_z \V^{-1} dz}\:\V(q).
\ee
We will see that this definition of $\Vhat$ suffices to define the brackets of $\Vhat$ with $\V$, and ultimately with $\Vhat$
at regular $\V$ fields. \eqref{Vhat_def_repetition} defines the variation of $\Vhat$ corresponding to the singular variations of $\V$ 
which the Poisson bracket produces even at solutions in which $\V$ is regular. 

In the following the deformation point will be represented by the {\em deformation parameter} $l = \rho(r)$. $l$ is closely 
related to the spectral parameter $w$, in fact $l = \tfrac{1}{2}(w + \rho^+)$, but it is a bit more convenient for
our purposes. Field points will also be represented by the corresponding values of a coordinate $\chi$ on $\cN_L$. 
$\chi$ starts at $0$ on the axis and increases smoothly and monotonically from there, but is otherwise arbitrary. 
In order to keep the field $\rho$ explicitly visible in our formalism $\chi$ will in general {\em not} be identified with $\rho$. 
The coordinates of the various field points involved in the calculation will be denoted by different letters 
$x, z, ...$, but these are all values of the same function $\chi$ evaluated at the corresponding points. 

By proposition \ref{path_ordered_exponential_functional_derivative} the variation of $\Vhat$ is
\be\label{Vhat_variation_00}
\dg\Vhat(x) = \int_0^x S(0,z) \dg[\V(u-1)P_z\V^{-1}]_z S(z,x) dz\ \V(x) + S(0,x) \dg\V(x),
\ee
where $S(a,b) \equiv \cP e^{\int_{a}^b (u-1)\V P_z \V^{-1} dz}$, provided $(u-1)\V P_z \V^{-1}$ and its variation are integrable on the interval $[0,x]$.
The integrability of $(u-1)\V P_z \V^{-1}$ follows from the smoothness of $\V$ and the local integrability of $u$. Whether or not the variation is integrable
depends on the variation under consideration.

Using the relation $S(0,z) = \Vhat(z)\V^{-1}(z)$, which follows from \eqref{Vhat_def2_repeat}, and $S(z,x) = S(0,z)^{-1}S(0,x)$, \eqref{Vhat_variation_00} 
may be expressed as
\be\label{Vhat_variation_0}
\dg\Vhat(x) = \int_0^x \Vh\V^{-1}\dg[\V(u-1)P_z\V^{-1}]\V\Vh^{-1} dz\ \Vh(x)\V^{-1}(x)\,\V(x) + \Vh(x)\V^{-1}(x)\,\dg\V(x).
\ee
It follows that 
\be\label{Vhat_variation_1}
\Vhat^{-1}(x)\dg\Vhat(x) = \int_0^x T(x,z)(u-1)\left(\dg P_z + [\V^{-1}\dg \V, P_z]\right)_z T(z,x)\, dz + \V^{-1}(x)\dg\V(x).
\ee

The integrand can be rewritten in a useful way. Note first that 
\be
\dg J = \dg[\cV^{-1} d\cV] = d (\cV^{-1}\dg\cV) + [J, \cV^{-1}\dg\cV],
\ee
so, since $[\fk,\fk] \subset \fh$ and $[\fk,\fh] \subset \fk$ 
\be
\dg P = d (\cV^{-1}\dg\cV)_\fk + [Q, (\cV^{-1}\dg\cV)_\fk] + [P, (\cV^{-1}\dg\cV)_\fh] = D (\cV^{-1}\dg\cV)_\fk + [P, (\cV^{-1}\dg\cV)_\fh],
\ee
where $D = d + ad_Q$ is an $H$ covariant derivative on $\fg$ valued fields (with $ad_X Y \equiv [X,Y]$). Thus
\be
\dg P + [\V^{-1}\dg \V, P] = D(\V^{-1}\dg\V)_\fk - [P, (\V^{-1}\dg \V)_\fk],\label{delta_simplification}
\ee
and therefore
\begin{align}
(u - 1)(\dg P & + [\V^{-1}\dg \V, P])\\
& = D((\frac{1}{u} - 1) (\V^{-1}\dg\V)_\fk) + (\frac{1}{u} - 1) [\hat{P}, (\V^{-1}\dg \V)_\fk] - d(\frac{1}{u}) (\V^{-1}\dg\V)_\fk
+ (u - \frac{1}{u})D(\V^{-1}\dg\V)_\fk,\label{delta_simplification1}
\end{align}
where $\hat{P} = uP$ is the $\fk$ component of $\hat{J}$. On $\cN_L$ $u = \frac{1}{\sqrt{1 - \rho/l}}$ so $u - \frac{1}{u} = \frac{u}{l}\rho$ and
$d \frac{1}{u} = - \frac{1}{2} \frac{u}{l}d\rho$ there. It follows that 
\begin{align}
(u - 1)(\dg P + [\V^{-1}\dg \V, P])
& = [D + ad_{\hat{P}}][(\frac{1}{u} - 1) (\V^{-1}\dg\V)_\fk] + \frac{1}{2}\frac{u}{l} \di_z\rho (\V^{-1}\dg\V)_\fk + \frac{u}{l} \rho D_z(\V^{-1}\dg\V)_\fk\\
& = [D + ad_{\hat{P}}][(\frac{1}{u} - 1) (\V^{-1}\dg\V)_\fk] + \frac{u}{l} \sqrt{\rho} D_z[\sqrt{\rho}(\V^{-1}\dg\V)_\fk].\label{delta_simplification2}
\end{align}

The first term in \eqref{delta_simplification2} gives rise to an easily integrable contribution to the integrand of \eqref{Vhat_variation_1}
because for any $\fg$ valued field $X$
\begin{align}
\di_z[T(0,z)X T(z,x)] & = T(0,z)\{\hat{J}_z(z) X - X \hat{J}_z(z) + \di_z X\} T(z,x)\\
		      & = T(0,z)\{[D_z + ad_{\hat{P}_z}]X\} T(z,x).
\end{align}
Equation \eqref{Vhat_variation_1} thus reduces to
\begin{align}
\Vhat^{-1}\dg\Vhat (x) = &\, (\cV^{-1}\dg\cV)_{\shat{\fg}}(x) - \lim_{z \rightarrow 0} T(0,z)(\frac{1}{u} - 1) (\V^{-1}\dg\V)_\fk T(z,x) \nonumber\\
& + \int_0^x T(x,z)\left(\frac{u}{l}\sqrt{\rho} D_z [\sqrt{\rho}(\cV^{-1}\dg\cV)_\fk]\right)_z T(z,x) dz,  \label{Vhat_variation2}
\end{align}
where the inverted caret indicates that the $\fk$ component is divided by $u$: $X_{\shat{\fg}} \equiv X_\fh + \frac{1}{u} X_\fk$.
The limit term vanishes if $(\V^{-1}\dg\V)_\fk$ diverges more slowly than $1/\rho$ as the axis $z = 0$ is approached, 
because  $\frac{1}{u} - 1 = \sqrt{1 - \tfrac{\rho}{l}} - 1$ goes to zero linearly in $\rho$ there and $T(0,z)$ and $T(z,x)$ have 
finite limiting values when $\V$ is regular as we have assumed. Therefore, for such variations 
\be\label{Vhat_variation}
\Vhat^{-1}\dg\Vhat (x) = (\cV^{-1}\dg\cV)_{\shat{\fg}}(x) 
 + \int_0^x T(x,z)\left(\frac{u}{l}\sqrt{\rho} D_z [\sqrt{\rho}(\cV^{-1}\dg\cV)_\fk]\right)_z T(z,x) dz.
\ee
This applies in particular to variations of $\V$ generated by the Poisson bracket \eqref{V_bracket1}, which diverge only as $1/\sqrt{\rho}$ as the
axis is approached.

We are thus ready to calculate the Poisson bracket between the $\Vhat$. 
The formula \eqref{Vhat_variation} expresses logarithmic variations $\Vhat^{-1} \dg\Vhat$ of $\Vhat$ in terms of
logarithmic variations of $\cV$. If we denote the logarithmic bracket
$\oo{\cV}{}^{-1}(x_1)\ot{\cV}{}^{-1}(x_2)\{\oo\cV(x_1),\ot\cV(x_2)\}$ of the field $\cV$ by $\oot A(x_1,x_2) \equiv
\oot A_{\fg\fg}(x_1,x_2)$, then \eqref{Vhat_variation} shows that the logarithmic bracket of the deformed zweibein $\oo\Vhat(x_1)$ with 
the undeformed zweibein $\ot\V(x_2)$ is
\begin{equation}\label{Vhat_V_0}
\oo{\Vhat}{}^{-1}\ot{\V}{}^{-1}\{\oo\Vhat,\ot\V\}= \sA[12]{g}{g}(x_1,x_2) + \int_0^{x_1} 
\oo T(x_1,z) \left(\frac{\ue}{l_1}\sqrt{\rho}\oo D_z [\sqrt{\rho}\A[12]{k}{g}(z,x_2)]\right)_z\oo T(z,x_1) dz.
\end{equation}
The fields inside the round brackets in the integrand are evaluated at $z$, as the subscript indicates, and $l_1$ is of course the deformation parameter of 
$\overset{1}{\Vhat}$, $\overset{1}{u}$, and $\overset{1}{T}$.

Differentiating the expression \eqref{V_bracket1} for the $\fk\fk$ component $\A[12]{k}{k}(z_1,z_2)$ of the logarithmic bracket of the 
undeformed zweibein yields the relations
\bearr
\oo D_{z_1}\left[\sqrt{\rho_1}\A[12]{k}{k}(z_1,z_2)\right]& = & 16\pi G_2\frac{\de(z_1-z_2)}{\sqrt{\rho_2}}\Omk, \label{equAkk1}\\
\ot D_{z_2}\left[\sqrt{\rho_2}\A[12]{k}{k}(z_1,z_2)\right]& = & -16\pi G_2\frac{\de(z_1-z_2)}{\sqrt{\rho_1}}\Omk.\label{equAkk2}
\eearr
The first of these allows us to evaluate the $\A{k}{k}$ contribution to the integral in \eqref{Vhat_V_0}, yielding:
\begin{align}
\oo{\Vhat}{}^{-1}\ot{\V}{}^{-1}\{\oo\Vhat,\ot\V\}= \sA[12]{g}{g}(x_1,x_2)
&+16\pi G_2 \frac{\ue(x_2)}{l_1}\oo T(x_1,x_2)\Omk\oo T(x_2,x_1)\Theta(x_2,x_1)\notag\\
&+\int_0^{x_1} \oo T(x_1,z) \left(\frac{\ue}{l_1}\sqrt{\rho}\oo D_z [\sqrt{\rho}\A[12]{k}{h}(z,x_2)]\right)_z\oo T(z,x_1)dz.\label{Vhat_V_1}
\end{align}
Here $\Theta(x,y) = \int_0^{y} \dg(z-x) dz$ is the distribution corresponding to the step function which is $1$ if $x < y$ and $0$ otherwise.
Note that the Dirac delta and the step $\Theta$ are order $0$ distributions, that is Radon measures, so their products with continuous functions are
well defined. 

In this derivation we have committed the following small sin: The variation of $(\oo u-1) \oo\V \oo P_z \oo\V {}^{-1}$ generated by $\ot\cV(x_2)$ via the Poisson 
bracket is not an integrable function, so 
\eqref{Vhat_variation_00} is not justified. This can be remedied by smearing with a continuous test function of $x_2$ supported on a compact subset of $x_2 > 0$.   
From the expression \eqref{V_bracket1} it is clear that smearing $\A{k}{k}$ in $x_2$ produces a $C^1$ function of $x_1$ which diverges as $1/\sqrt{\rho}$ at the axis. 
Equation \eqref{delta_simplification2} then shows that the corresponding $\dg[(u-1)\V P_z \V^{-1}]$ is integrable, so \eqref{Vhat_variation_00} holds. The calculation 
can then proceed, yielding the $\fg\fk$ component of the bracket $\oo{\Vhat}{}^{-1}\ot{\V}{}^{-1}\{\oo\Vhat,\ot\V\}$ as a distribution in $x_2$.
This is the part of the bracket we will actually use, but we note that the calculation can be justified in the same way for all components of this bracket if a suitable 
gauge, such as upper triangular gauge, is chosen for $\V$ so that all components of $A$ are determined by $\A{k}{k}$. (This does not mean that \eqref{Vhat_V_1} is only valid 
in certain gauges, rather it is in certain gauges that it is evident that the calculation is valid. The result can then be expressed in any gauge, taking always the form 
\eqref{Vhat_V_1}.)

Applying \eqref{Vhat_variation} again, this time to \eqref{Vhat_V_1}, an expression for the logarithmic bracket of $\oo\Vhat(x_1)$ with $\ot\Vhat(x_2)$ is obtained:
\begin{align}
\oo{\Vhat}{}^{-1}\ot{\Vhat}{}^{-1}\{\oo\Vhat,\ot\Vhat\} = &\,\sAs[12]{g}{g}(x_1,x_2)
+16\pi G_2 \frac{1}{l_1}\left(\frac{\ue}{\uz}\right)_{x_2}\oo T(x_1,x_2)\Omk\oo T(x_2,x_1)\Theta(x_2,x_1)\notag\\
&-16\pi G_2 \frac{1}{l_2}\left(\frac{\uz}{\ue}\right)_{x_1}\ot T(x_2,x_1)\Omk\ot T(x_1,x_2)\Theta(x_1,x_2)\notag\\
&+16\pi G_2\int_0^{x_2} \ot T(x_2,z) \left(\frac{\uz}{l_2}\sqrt{\rho}\ot D_z [\frac{\ue}{l_1}\sqrt{\rho}\oo T(x_1,z)\Omk\oo T(z,x_1)\Theta(z,x_1)]\right)_z\ot T(z,x_2)dz\notag\\
&+\int_0^{x_1} \oo T(x_1,z) \left(\frac{\ue}{l_1}\sqrt{\rho}\oo D_z [\sqrt{\rho}\A[12]{k}{h}(z,x_2)]\right)_z\oo T(z,x_1)dz\notag\\
&+\int_0^{x_2} \ot T(x_2,z) \left(\frac{\uz}{l_2}\sqrt{\rho}\ot D_z [\sqrt{\rho}\A[12]{h}{k}(x_1,z)]\right)_z\ot T(z,x_2)dz, \label{braVzero1}
\end{align}
where $l_i$ is the deformation parameter of $\overset{i}{\Vhat}$, $\overset{i}{u}$, and $\overset{i}{T}$.
The relation \eqref{equAkk2} has been used to obtain the second line.

Here we are faced once more with the problem that the variation of $(u-1)\V P_z \V {}^{-1}$, this time generated by $\oo\Vhat$, is not in general 
an integrable function. This can be avoided by smearing $\oo\Vhat$ over $x_1$ or $l_1$, or both, with a test function. Since we are ultimately
interested in the Poisson brackets of the fields $\U(l)$, that is of $\Vh(x;l)$ with $l = \rho(x)$, we will smear with test functions supported 
on such $(x,l)$. However, before undertaking that calculation, in subsection \ref{U_brack}, we will calculate the bracket 
$\oo{\Vhat}{}^{-1}\ot{\Vhat}{}^{-1}\{\oo\Vhat,\ot\Vhat\}$ in the case $x_1 \neq x_2$, $\rho(x_i) \leq l_i$, which can be done without smearing.
This in fact suffices to determine $\{\oo\U(l_1), \ot\U(l_2)\}$ in all cases except $l_1 = l_2$. The task of subsection \ref{U_brack} thus reduces 
essentially to determining the singular distributional component of $\{\oo\U(l_1), \ot\U(l_2)\}$ at $l_1 = l_2$.

Let us therefore evaluate the expression (\ref{braVzero1}) for the bracket $\oo{\Vhat}{}^{-1}\ot{\Vhat}{}^{-1}\{\oo\Vhat,\ot\Vhat\}$ with the 
restriction $x_2 < x_1$,  $\rho(x_i) \leq l_i$. (The bracket can be evaluated for $x_2 > x_1$,  $\rho(x_i) \leq l_i$ in an entirely analogous manner, 
calculating first $\oo{\V}{}^{-1}\ot{\Vhat}{}^{-1}\{\oo\V,\ot\Vhat\}$.) The restrictions on $x_1$ and $x_2$ ensure that 
$\oo{\Vhat}{}^{-1}(x_1)\ot{\V}{}^{-1}(z)\{\oo\Vhat(x_1),\ot\V(z)\}_{\fk\fk}$ is $C^1$ in $z \in (0, x_2]$, or indeed as smooth as $\V$ is, since they 
guarantee that $z$ never gets to the step at $x_1$, and that $\rho(z) \leq \rho(x_2) < \rho(x_1) \leq l_1$. Its only singularity is a $1/\sqrt{\rho(z)}$ 
as $z$ approaches $0$. Thus our calculation of the $\fk\fg$ component of $\oo{\Vhat}{}^{-1}\ot{\Vhat}{}^{-1}\{\oo\Vhat,\ot\Vhat\}$ is valid, by the same 
argument that justified the calculation of $\oo{\Vhat}{}^{-1}\ot{\V}{}^{-1}\{\oo\Vhat,\ot\V\}_{\fg\fk}$ above. Moreover, as in that case  
this justification can be extended to all components of $\oo{\Vhat}{}^{-1}\ot{\Vhat}{}^{-1}\{\oo\Vhat,\ot\Vhat\}$ by adopting a suitable gauge for $\V$ in 
intermediate stages of the calculation.

Our strategy will be to express the integrand of the third term in \eqref{braVzero1} as a total derivative, allowing us to integrate this term in closed form. 
Because $z \leq x_2 < x_1$ the factor in brackets in this term can be written as
\begin{align}\label{third_term}
\left(\frac{\uz}{l_2}\sqrt{\rho}\right.&\left.\ot D_z [\frac{\ue}{l_1}\sqrt{\rho}\oo T(x_1,z)\Omk\oo T(z,x_1)]\vphantom{\frac{\uz}{l_2}\sqrt{\rho}} \right)_z \notag\\
& = \oo T(x_1,z) \left(\frac{\uz}{l_2}\sqrt{\rho}\di_z\left[\frac{\ue}{l_1}\sqrt{\rho}\right]\Omk + \frac{\ue\uz}{l_1l_2}\rho [\oo\Jhat + \ot Q,\Omk] \right)_z \oo T(z,x_1).
\end{align}
To simplify this expression we make use of two identities, 
\be \label{u_identity_1}
\frac{\ue\uz}{l_1 l_2}\rho = -\frac{1}{l_1 - l_2}\left(\frac{\ue}{\uz} - \frac{\uz}{\ue}\right)
\ee
and 
\be\label{u_identity_2}
\frac{\uz}{l_2}\sqrt{\rho} \di_z \left[\frac{\ue}{l_1}\sqrt{\rho}\right] 
= - \frac{1}{l_1 - l_2}\di_z\left(\frac{\ue}{\uz}\right).
\ee
The first identity follows immediately from
\be
\left(\frac{\ue}{\uz} - \frac{\uz}{\ue}\right)\bigg/\frac{\ue\uz}{l_1 l_2}\rho 
= \frac{1}{\rho}l_1l_2\left(\frac{1}{\uz^2} - \frac{1}{\ue^2}\right) = l_1l_2\left(\frac{1}{l_1} - \frac{1}{l_2}\right) = l_2 - l_1.
\ee
To demonstrate the second identity note that
\be
\frac{\uz}{l_2}\sqrt{\rho} \di_z \left(\frac{\ue}{l_1}\sqrt{\rho}\right) + \frac{\ue}{l_1}\sqrt{\rho} \di_z \left(\frac{\uz}{l_2}\sqrt{\rho}\right) =
\di_z\left[\frac{\ue\uz}{l_1 l_2}\rho\right] = - \frac{1}{l_1 - l_2}\di_z\left(\frac{\ue}{\uz} - \frac{\uz}{\ue}\right),
\ee
and that 
\begin{align}
\frac{\uz}{l_2}\sqrt{\rho} \di_z \left(\frac{\ue}{l_1}\sqrt{\rho}\right) - \frac{\ue}{l_1}\sqrt{\rho} \di_z \left(\frac{\uz}{l_2}\sqrt{\rho}\right) 
& = \left[\di_z \ln \left(\frac{\ue}{l_1}\sqrt{\rho}\right) - \di_z \ln \left(\frac{\uz}{l_2}\sqrt{\rho}\right)\right]\frac{\ue\uz}{l_1 l_2}\rho \\
& = - \frac{1}{l_1 - l_2}\left(\frac{\ue}{\uz} - \frac{\uz}{\ue}\right)\left(\frac{\di_z \ue}{\ue} - \frac{\di_z \uz}{\uz}\right) \\
& = - \frac{1}{l_1 - l_2}\di_z\left(\frac{\ue}{\uz} + \frac{\uz}{\ue}\right).
\end{align}
Adding these equations yields \eqref{u_identity_2}.

Substitution of the two identities \eqref{u_identity_1} and \eqref{u_identity_2} reduces \eqref{third_term} to  
\be\label{third_term_2}
- \frac{1}{l_1 - l_2} \oo T(x_1,z) \left(\di_z\left(\frac{\ue}{\uz}\right)\Omk + \left(\frac{\ue}{\uz} - \frac{\uz}{\ue}\right)[\oo\Jhat + \ot Q,\Omk]\right)_z\oo T(z,x_1).
\ee
This expression looks like it is singular at $l_1 = l_2$, but this is not so if $z < x_2 < x_1$. The coefficient of $1/(l_1 - l_2)$ has a regular zero at $l_1 = l_2$ 
which cancels the divergence. The only real divergence is an integrable one at $z = x_2$ when $\rho(x_2) = l_2$.

Now let us analyze the commutator term. $[\oo Q + \ot Q,\Omk] = 0$, because $\Omk$ is invariant under the simultaneous adjoint action of $H$ on both $\oo{\fg}$ and $\ot{\fg}$, so
\be
[\oo\Jh + \ot Q,\Omk] = [\oo\Jh - \oo Q,\Omk] = [\oo{\hat{P}},\Omk].
\ee
Similarly, because $\Omega_\fg$ is invariant under the simultaneous adjoint action of all $G$ on $\oo{\fg}$ and $\ot{\fg}$,
\be	\label{Omega_g_invariant}
0 = [\Pe + \Pz, \Omega_\fg] = [\Pe + \Pz, \Omega_\fh + \Omega_\fk].
\ee
Since
\be
[\Pe, \Omega_\fh]\in \oo{\fk}\otimes\ot{\fh} \qquad [\Pz, \Omega_\fh]\in \oo{\fh}\otimes\ot{\fk}
\qquad [\Pe, \Omega_\fk]\in \oo{\fh}\otimes\ot{\fk} \qquad [\Pz, \Omega_\fk]\in \oo{\fk}\otimes\ot{\fh},
\ee
\eqref{Omega_g_invariant} implies that $[\Pe, \Omega_\fk] = - [\Pz, \Omega_\fh]$ and $[\Pz, \Omega_\fk] = - [\Pe, \Omega_\fh]$. It follows that
\be
\frac{\uz}{\ue}[\oo{\hat{P}}, \Omega_\fk] = -[\ot{\hat{P}}, \Omega_\fh] \qquad \text{and} \qquad \frac{\ue}{\uz}[\ot{\hat{P}}, \Omega_\fk] = -[\oo{\hat{P}}, \Omega_\fh],
\ee
and therefore that
\be\label{uPOmega_id}
\left(\frac{\ue}{\uz} - \frac{\uz}{\ue}\right)[\oo{\hat{P}},\Omk] = \frac{\ue}{\uz}[\oo{\hat{P}},\Omk] + [\ot{\hat{P}},\Omh] 
= \frac{\ue}{\uz}[\oo{\hat{P}} + \ot{\hat{P}},\Omk] + [\oo{\hat{P}} + \ot{\hat{P}},\Omh]
= [\oo\Jh + \ot\Jh, \frac{\ue}{\uz}\Omk + \Omh].
\ee

These results allow us to write the third term of \eqref{braVzero1} as
\begin{align}
- & \frac{16\pi G_2}{l_1 - l_2}\int_0^{x_2}\ot T(x_2,z)\oo T(x_1,z)\left(\di_z\left(\frac{\ue}{\uz}\right)\Omk + [\oo\Jh + \ot\Jh, \frac{\ue}{\uz}\Omk + \Omh]\right)_z 
\oo T(z,x_1)\ot T(z,x_2)\,dz\notag\\
& = - \frac{16\pi G_2}{l_1 - l_2}\int_0^{x_2} \di_z\left[\ot T(x_2,z)\oo T(x_1,z)\left(\frac{\ue}{\uz}\Omk + \Omh\right)_z \oo T(z,x_1)\ot T(z,x_2)\right]\,dz\notag\\
& = \frac{16\pi G_2}{l_1 - l_2}\left\{\ot T(x_2,0)\oo T(x_1,0)\Omg \oo T(0,x_1)\ot T(0,x_2) -  \oo T(x_1,x_2)\left(\frac{\ue}{\uz}\Omk + \Omh\right)_{x_2} \oo T(x_2,x_1)\right\}.
\label{third_term3}
\end{align}

As mentioned earlier, the logarithmic bracket $\oo{\Vhat}{}^{-1}\ot{\Vhat}{}^{-1}\{\oo\Vhat,\ot\Vhat\}$ is also well defined, without smearing, when $x_2 > x_1$ (and $\rho(x_i) \leq l_i$).
To calculate it in this case one first evaluates $\oo{\V}{}^{-1}\ot{\Vhat}{}^{-1}\{\oo\V,\ot\Vhat\}$ in analogy with \eqref{Vhat_V_0} - \eqref{Vhat_V_1}, and then one applies 
\eqref{Vhat_variation} to the variation of $\oo\V$ generated by $\ot\Vhat$ to obtain $\oo{\Vhat}{}^{-1}\ot{\Vhat}{}^{-1}\{\oo\Vhat,\ot\Vhat\}$. 
It follows immediately from the antisymmetry of the bracket of the undeformed zweibein $\V$ that 
\be 
\oo{\V}{}^{-1}\ot{\Vhat}{}^{-1}\{\oo\V,\ot\Vhat\} = - \ot{\Vhat}{}^{-1}\oo{\V}{}^{-1}\{\ot\Vhat,\oo\V\}.
\ee
But applying \eqref{Vhat_variation} to the bracket on the right side of this equation yields precisely the expression \eqref{braVzero1} with the roles of $1$ and $2$ 
reversed. Because $x_1 < x_2$ our argument for the validity of this expression for the bracket applies, as does our calculation of the third term in \eqref{braVzero1}. 

In conclusion, the bracket of the deformed zweibeine $\oo\Vhat$ and $\ot\Vhat$ is well defined without smearing when $x_2 \neq x_1$ and $\rho(x_i) \leq l_i$, and it 
is antisymmetric under interchange of $1$ and $2$. Explicitly the logarithmic bracket is 
\begin{align}
\oo{\Vhat}{}^{-1}(x_1)\ot{\Vhat}{}^{-1}(x_2)&\{\oo\Vhat(x_1),\ot\Vhat(x_2)\}= \sAs[12]{g}{g}(x_1,x_2)\notag\\
&+16\pi G_2 \frac{1}{l_1}\left(\frac{\ue}{\uz}\right)_{x_2}\oo T(x_1,x_2)\Omk\oo T(x_2,x_1)\Theta(x_2,x_1)\notag\\
&-16\pi G_2 \frac{1}{l_2}\left(\frac{\uz}{\ue}\right)_{x_1}\ot T(x_2,x_1)\Omk\ot T(x_1,x_2)\Theta(x_1,x_2)\notag\\
&-\frac{16\pi G_2}{l_1-l_2}\Tt(x_2,x_1)\left(\left(\frac{\uz}{\ue}\right)_{x_1}\Omk+\Omh\right)\Tt(x_1,x_2)\Theta(x_1,x_2)\notag\\
&-\frac{16\pi G_2}{l_1-l_2}\To(x_1,x_2)\left(\left(\frac{\ue}{\uz}\right)_{x_2}\Omk+\Omh\right)\To(x_2,x_1)\Theta(x_2,x_1)\notag\\
&+\frac{16\pi G_2}{l_1-l_2}\To(x_1,0)\Tt(x_2,0)\Omg\To(0,x_1)\Tt(0,x_2)\notag\\
&+\int_0^{x_1} \oo T(x_1,z) \left(\frac{\ue}{l_1}\sqrt{\rho}\oo D_z [\sqrt{\rho}\A[12]{k}{h}(z,x_2)]\right)_z\oo T(z,x_1)dz\notag\\
&+\int_0^{x_2} \ot T(x_2,z) \left(\frac{\uz}{l_2}\sqrt{\rho}\ot D_z [\sqrt{\rho}\A[12]{h}{k}(x_1,z)]\right)_z\ot T(z,x_2)dz. \label{bravhrhozero}
\end{align}
This expression is almost entirely algebraic. The integrals that remain in \eqref{bravhrhozero} involve only the gauge components $\Ahk$ and $\Akh$ of $A$.
They cannot be evaluated without fixing a particular gauge, but on the other hand they do not influence the brackets of the deformed conformal metric $\E$, 
which is gauge invariant. 

The restriction $x_1 \neq x_2$ is in fact not essential for the validity of \eqref{bravhrhozero}. If $\rho(x_i)$ is strictly smaller than $l_i$ for one of the arguments 
$\ove{i}{\Vhat}$ then \eqref{bravhrhozero} holds as a distributional equality, without the restriction $x_1 \neq x_2$. This is easily established by calculating the bracket 
with $\ove{i}{\Vhat}$ smeared with a test function supported on $x_i < \rho^{-1}(l_i)$. The only cases that are really excluded are $l_1 = \rho(x_1) = \rho(x_2) = l_2$, 
and of course the case in which one or both of the $\rho(x_i)$ strictly exceeds $l_i$ - a case we have not attempted to address here. 

Equation \eqref{bravhrhozero} suffices for the calculation of the brackets $\{\oo\U(l_1),\ot\U(l_2)\}$ and $\{\oo\E(l_1),\ot\E(l_2)\}$ save in the case $l_1 = l_2$, which 
will be treated in the next subsection.

\subsection{The Poisson brackets of $\U$}\label{U_brack}

Recall that the deformed conformal metric at $r \in \cN_L$ is $\E(r)=\U(r)\U^t(r) = \Vh(r;r)\Vh^t(r;r)$ 
(see section \ref{transformation}). The Poisson bracket between the deformed conformal metric
at deformation point $r_1$ and at deformation point $r_2$ is therefore
\begin{equation}\label{E_bracket0}
\{\oo \E,\ot \E \} = \{\oo \U,\ot \U\} \oo \U{}^t \ot \U{}^t + \oo \U\{\oo \U{}^t,\ot \U\}\ot \U{}^t 
+ \ot \U\{\oo \U,\ot \U{}^t\}\oo \U{}^t
 + \oo \U \ot \U\{\oo \U{}^t,\ot \U{}^t\}.  
\end{equation}

This sum is nothing but ($4$ times) the symmetrization of the first term on the indices of space $1$ and on the 
indices of space $2$. It is therefore easily evaluated in terms of the logarithmic bracket 
$\oo \U{}^{-1}\ot \U{}^{-1}\{\oo \U,\ot \U\}$: If $\oot C$ denotes this logarithmic bracket, then the first term  in \eqref{E_bracket0} is 
$\oo \U\ot \U \oot C \oo \U{}^t\ot \U{}^t$, and the bracket of the deformed conformal metrics is
\be\label{E_bracket1}
\{\oo \E,\ot \E \} = \oo \U\ot \U[\oot C + {}^t \oot C + \oot C{}^t + {}^t \oot C {}^t]\oo \U{}^t\ot \U{}^t,
\ee
where a pre-superscript $t$ indicates transposition in space $1$ while a post-superscript $t$ indicates 
transposition in space $2$. Now note that $C$ is an element of $\fg\otimes\fg$ and that $\sk$ is the symmetric 
subspace of $\fg = \mathfrak{sl}(2, \Real)$: $\tfrac{1}{2}(a + a^t)$ is the $\sk$ component of $a \in \fg$. Thus 
\be\label{E_bracket2}
\{\oo \E,\ot \E \}= 4\oo \U\ot \U \oot C_{\sk\sk} \,\oo \U{}^t\ot \U{}^t.
\ee
In this subsection we calculate $\oot C_{\sk\sk} = [\oo \U{}^{-1}\ot \U{}^{-1}\{\oo \U,\ot \U\}]_{\sk\sk}$. 

An expression for $\oot C_{\sk\sk}$ can be obtained from equation \eqref{bravhrhozero} for the logarithmic bracket 
of $\Vhat$ by setting the field points of $\Vhe$ and $\Vhz$ equal to their deformation points and projecting to $\fk\otimes\fk$.
Projecting to $\fk$ in both space $1$ and space $2$ immediately eliminates the gauge dependent integrals in lines 
7 and 8 of \eqref{bravhrhozero}. It also eliminates the terms containing factors of $\Omh$ in lines 4 and 5, and 
the $\fh$ contributions to $\Asgg$ in line 1. (It does not eliminate the $\Omh$ term in line 6 because there $\Omh$ 
is multiplied by parallel transport matrices in both space $1$ and space $2$.) Further terms vanish because 
$1/\ue(x_1)$ and $1/\uz(x_2)$ are zero when the field points coincide with the corresponding deformation points:
When this is the case $\rho(x_i) = l_i$, so $1/\overset{i}{u}(x_i)=\sqrt{1 - \frac{l_i}{l_i}}=0$.
Meanwhile the values $\ue(x_2)=\frac{1}{\sqrt{1 - \frac{l_2}{l_1}}}$ and $\uz(x_1)=\frac{1}{\sqrt{1 - \frac{l_1}{l_2}}}$
are finite, assuming $l_1 \neq l_2$. All that is left, ultimately, is the projection on $\fk\otimes\fk$ of line 6:
\be \label{U_brack_0}
\oot C_{\sk\sk}=\frac{16\pi G_2}{l_1-l_2}\left[\To(x_1,0)\Tt(x_2,0)\Om_\fg\To(0,x_1)\Tt(0,x_2)\right]_{\sk\sk}.
\ee

This expression is not defined when $l_1 = l_2$. In fact it cannot even be interpreted as a distribution on any domain that includes the 
line $l_1 = l_2$, because the function $\frac{1}{l_1-l_2}$ is not locally integrable and therefore does not define a distribution. 
However, if one takes care to smear the logarithmic variations of $\oo \U$ and $\ot \U$ with test functions throughout the calculation of 
the bracket one does obtain a distribution on the domain $l_1 > 0, l_2 > 0$, namely
\be \label{U_brack_1}
\oot C_{\sk\sk} = 16\pi G_2\: p.v.\left(\frac{1}{l_1-l_2}\right)\left[\To(x_1,0)\Tt(x_2,0)\Om_\fg\To(0,x_1)\Tt(0,x_2)\right]_{\sk\sk},
\ee
where $p.v.(\frac{1}{l_1 - l_2})$ is the Cauchy principal value of $\frac{1}{l_1 - l_2}$, a distribution defined by 
\be \label{Cauchy_pv_def}
\int_{\Real^2} f(l_1,l_2)\:p.v.\left(\frac{1}{l_1-l_2}\right)d^2l = \lim_{\eg\rightarrow 0^+} \int_{|l_1-l_2|>\eg} 
f(l_1,l_2)\frac{1}{l_1-l_2}\,d^2l.
\ee

Note that the right side of this definition may be rewritten as
\begin{align}
 \lim_{\eg\rightarrow 0^+} \int_{l_1 > l_2 + \eg} \frac{ f(l_1,l_2) -  f(l_2,l_1)}{l_1-l_2}\,d^2l = & \int_{l_1 > l_2} \frac{ f(l_1,l_2) -  f(l_2,l_1)}{l_1-l_2}\,d^2l \\
	     = & \int_{\Real^2} \tfrac{1}{2}[f(l_1,l_2) - f(l_2,l_1)]\:\frac{1}{l_1-l_2}d^2l. \label{Cauchy_pv_def2}
\end{align}
That is, it is just $\frac{1}{l_1-l_2}$ integrated against the antisymmetric component of the test function $f$. 

Note also that $p.v.\left(\frac{1}{l_1-l_2}\right)$ is an order 1 distribution: it is a continuous linear functional on compactly supported $C^1$ 
test functions, not just on smooth ones. Its products with $C^1$ functions are therefore defined. This means that the product \eqref{U_brack_1} is 
defined as a distribution provided $T(\mathbf{0},r;r)$ is $C^1$ in $l = \rho(r)$. And it is, as we shall now see, because $\V$ and $\rho$ are 
assumed to be smooth on $\cN_L$. 

In fact, the smoothness of $\V$ as a function of $\rho$ implies that the holonomy $T(p,r;r)$ is smooth in the variables 
$s = \sqrt{1 - \frac{\rho(p)}{l}}$ and $l$: Since $\rho$ is smooth and monotonic it can be used as a chart on $\cN_L$
\footnote{We have avoided using this chart to keep the field $\rho$ visible in our formalism but we will make an exception in this subsection.}
and $T(\mathbf{0},r;r)$ may be expressed as a function of $\rho(r) = l$ and $\rho(p)$ which will be called just $\rho$ here. 
This function $T(\rho,l;l)$ is determined by the initial value problem 
\begin{equation}
 T(l,l;l) = \One \qquad\qquad \left(\frac{\di T}{\di \rho}\right)_l = - [Q_\rho(\rho) + u P_\rho(\rho)]\:T, 
\end{equation}
where $Q_\rho$ and $P_\rho$ are the $\fh$ and $\fk$ components of $\V^{-1}\di_\rho \V$, which are smooth in $\rho$ because $\V$ is. 
Changing to the variables $s, l$ the initial value problem becomes
\begin{equation}
 T(0,l;l) = \One \qquad\qquad \left(\frac{\di T}{\di s}\right)_l = 2l[ s Q_\rho(l[1 - s^2]) + P_\rho(l[1 - s^2])]\:T. 
\end{equation}
Since the right side is $C^\infty$ in $s$, $l$ and $T$, it follows that the solution is $C^\infty$ in $s$ and $l$. See for instance Ch. II, Sec. 4 of 
\cite{Lefschetz}.
 
Let us demonstrate \eqref{U_brack_1}. As a first step consider the logarithmic variation of $\U$ corresponding to a given smooth variation of $\V$, smeared with a smooth
$\fk$ valued test function $\varphi$ of $l$. By \eqref{Vhat_variation} it takes the value
\be\label{U_smeared_variation}
\int_0^\infty {\rm tr}[\varphi\, \U^{-1}\dg \U]\, dl = \int_0^\infty dl \int_0^l d\rho\, 
{\rm tr}[T(\rho,l)\varphi T(l,\rho)\frac{u}{l}\sqrt{\rho} D_\rho (\sqrt{\rho}(\cV^{-1}\dg\cV)_\fk)].
\ee
We shall suppose that the support of $\varphi$ is compact and excludes $l = 0$. 

The integrand is integrable because it is the product of a continuous function of compact support in the the domain of integration with a locally integrable function, $u$, 
so the order of integration may be reversed giving
\be\label{U_smeared_variation1}
\int_0^\infty {\rm tr}[\varphi\, \U^{-1}\dg \U]\, dl = \int_0^\infty d\rho\,
{\rm tr}[F \sqrt{\rho} D_\rho (\sqrt{\rho}(\cV^{-1}\dg\cV)_\fk) ],
\ee
with 
\begin{equation}\label{F_def}
  F(\rho) = \int_\rho^\infty  T(\rho,l)\varphi(l) T(l,\rho)\frac{u}{l}\, dl.
\end{equation}

Suppose now that the variation of $\V$ is that generated via the Poisson bracket by $\ot\V$ at a point $x_2$. Then, by \eqref{equAkk1}
\begin{equation}\label{U_smeared_variation2}
 \int_0^\infty \oo{\rm tr}[\oo \varphi \,\oo \U{}^{-1} \ot \V{}^{-1}\{\oo \U, \ot \V\}]\, dl_1 = 16 \pi G_2\oo{\rm tr}[\oo F(\rho(x_2))\oot\Omega_\fk].
\end{equation}
(To better distinguish it from $\ot \V$, $\V$ has been labeled with a $1$, as have the remaining variables appearing in the left side of 
\eqref{U_smeared_variation1} and the representation space in which $\V$ acts.) Of course the variation of $\V$ generated by $\ot\V(x_2)$ is not
smooth, it must be smeared in $x_2$ before \eqref{U_smeared_variation1} can be applied, so \eqref{U_smeared_variation2} is a distributional equation.

Equation \eqref{U_smeared_variation2} provides an expression for the logarithmic variation of $\ot \V$ generated via the Poisson bracket by 
the logarithmic gradient $\oo\U{}^{-1} \dI \,\oo \U$ of $\oo \U$ on phase space smeared with $\oo \varphi$. Substituting this logarithmic variation into 
\eqref{U_smeared_variation1} yields the expression
\begin{equation}\label{smeared_U_bracket0}
  \int_0^\infty dl_2 \int_0^\infty dl_1 \ \oot{\rm tr}[\oo\varphi \ot\varphi\,\oo \U{}^{-1} \ot \U{}^{-1}\{\oo \U, \ot \U\}]  
  = 16 \pi G_2\int_0^\infty d\rho\  \oot {\rm tr}[\ot F \sqrt{\rho} \ot D_\rho [\sqrt{\rho}\oo F\oot\Omega_\fk] ]
\end{equation}
for the logarithmic bracket of $\oo\U$ and $\ot\U$ smeared with test functions $\oo\varphi$ and $\ot\varphi$. Note that at any given phase space 
point $p$ the smeared logarithmic bracket on the left side of the equation may be viewed as the bracket between the smeared fields 
$\int_0^{\infty} \overset{i}{{\rm tr}}[\overset{i}{\theta}\, \overset{i}{\U}] dl_i$ with $\overset{i}{\theta} = \overset{i}{\varphi}\,\overset{i}{\U}(p){}^{-1}$, 
which is perhaps a more conventional way of defining a smeared Poisson bracket.

In the preceding calculation \eqref{Vhat_variation} was assumed to be valid be for the logarithmic variation of $\ot \V$ defined by \eqref{U_smeared_variation2},
which requires that $F_\fk$ is sufficiently regular. In fact $F$ is a $C^1$ function of $\rho$ on its entire domain $[0,\infty)$, which is more than sufficient. 
To see this consider first the case $\rho < l_{min}/2$, where $l_{min} > 0$ is the minimum of $l$ in the support of $\varphi$. For such $\rho$ the continuity of $dF/d\rho$
is easily established using the fact that both the integrand of \eqref{F_def} and its derivative in $\rho$ are continuous, and that the support of $\varphi$ is compact.
If on the other hand $\rho > 0$ $F$ may be reexpressed as an integral over $a = l/\rho$:
\begin{equation}\label{F_2}
  F(\rho) = \int_1^\infty \phi \frac{u}{a}\, da
\end{equation}
with $\phi(\rho,a) =  T(\rho,\rho a;\rho a)\varphi(\rho a) T(\rho a,\rho; \rho a)$. Since $T$ is a $C^1$ function of $s$ and $l$, it is also a $C^1$ function of $\rho$ and 
$t = \sqrt{a - 1}$, because $s = t/\sqrt{t^2 + 1}$ and $l = \rho (t^2 + 1)$ are smooth functions of $\rho$ and $t$. It follows that $\di_\rho \phi$ is continuous in $\rho$ and $a$.
Since it is also of compact support in $a$, and $u = (1 -a^{-1})^{-1/2}$ is locally integrable, the $\rho$ derivative of the integrand $\phi \frac{u}{a}$ of \eqref{F_2} is integrable 
over the two dimensional domain $a \in [1, \infty)$, $\rho \in [\rho_1, \rho_2]$, where $\rho_2 \geq \rho_1 > 0$. Thus 
\begin{equation}\label{F_3}
  F(\rho_2) = F(\rho_1) + \int_1^\infty da\,\int_{\rho_1}^{\rho_2} d\rho \di_\rho\phi \frac{u}{a},
\end{equation}
and by Fubini's theorem the order of integration may be reversed. It follows that the $\rho$ derivative of $F$ is the integral of $\di_\rho\phi \frac{u}{a}$
and, by dominated convergence, that $dF/d\rho$ is continuous.

To demonstrate \eqref{U_brack_1} we must evaluate the right side of \eqref{smeared_U_bracket0}, $16\pi G_2 \int_0^\infty H\,d\rho$ with  
\begin{equation}
 H = \oot {\rm tr}[\ot F \sqrt{\rho} \ot D_\rho [\sqrt{\rho}\oo F\Omega_\fk]] 
 = \oot{\rm tr}[\oot \Omega_\fk\left(\tfrac{1}{2}\oo F \ot F + \rho \ot F \frac{d}{d\rho} \oo F + \rho \ot F [\oo Q, \oo F]\right)]
\end{equation}

Note that $H$ is integrable because the $\overset{i}{F}$ are $C^1$ and of compact support, so $\int_0^\infty H d\rho = \lim_{\rho_0 \rightarrow 0^+} \int_{\rho_0}^\infty H d\rho$.  
When $\rho_0 > 0$ the integral $I(\rho_0) \equiv \int_{\rho_0}^\infty H d\rho$ may be expressed as an integral over $\rho$ and $a_1$ and $a_2$ with $a_i = l_i/\rho$:
\begin{equation}\label{I_rho_0}
 I(\rho_0) = \int_{\rho_0}^\infty d\rho \int_1^\infty da_2 \int_1^\infty da_1 
 \ \oot{\rm tr}[\Omega_\fk\left(\tfrac{1}{2}\oo\phi \ot\phi + a_1 \di_{a_1} \oo\phi \ot\phi - \rho [\oo \Phat, \oo\phi \ot\phi]\right)]
 \frac{\oo u \ot u}{a_1 a_2}.
\end{equation}
Here we have used once more the fact that the $\rho$ derivative may be taken inside the integral \eqref{F_2} for $F$, so that $dF/d\rho$ is the integral of
\begin{equation}
 [\di_\rho\phi]_a \frac{u}{a} = \left([\di_\rho \phi]_l + \left(\frac{\di l}{\di \rho}\right)_a [\di_l \phi]_\rho \right)\frac{u}{a} 
                          = \left(-[\Jhat_\rho, \phi] + \frac{a}{\rho} [\di_a \phi]_\rho \right)\frac{u}{a}. 
\end{equation}

The integrand of \eqref{I_rho_0} is also integrable because it is the product of a continuous function of compact support and a locally integrable factor 
$\oo u \ot u$: The first term in the bracket is continuous, and the remaining two terms sum to $\rho \left([\di_\rho \oo\phi]_a + [\oo Q{}_\rho,\oo\phi]\right) \ot\phi$ 
which is also continuous; the integrand has compact support in the domain $\rho\geq \rho_0 > 0$ and $a_i \geq 1$ because $l_i = \rho a_i$ are bounded in the supports 
of $\overset{i}{\varphi}$. It follows that $I(\rho_0)$ can be expressed as the limit of the integral 
\begin{equation}\label{I_rho_0_epsilon}
 I_\eg(\rho_0) = \int_{\rho_0}^\infty d\rho \int_{1 + \eg}^\infty da_2 \int_{1 + \eg}^\infty da_1 
 \ \oot{\rm tr}[\Omega_\fk\left(\tfrac{1}{2}\oo\phi \ot\phi + a_1 \di_{a_1} \oo\phi \ot\phi - \rho [\oo \Phat, \oo\phi \ot\phi]\right)]
 \frac{\oo u \ot u}{a_1 a_2},
\end{equation}
as $\eg > 0$ tends to zero.

When $a_1 \neq a_2$ the integrand can be shown to be a divergence using two identities:  
\begin{align}
 \left(\tfrac{1}{2}\oo\phi \ot\phi + a_1 \di_{a_1} \oo\phi \ot\phi \right)\frac{\oo u \ot u}{a_1 a_2}
      = & -  \frac{1}{a_1 - a_2} \frac{\oo u}{\ot u} [ 1 +  a_1 \di_{a_1} +  a_2 \di_{a_2}] \oo\phi \ot\phi \nonumber \\
        & + \di_{a_1}\left(\frac{a_1}{a_1 - a_2}\frac{\ot u}{\oo u}\oo\phi\, \ot\phi \right) 
        + \di_{a_2}\left(\frac{a_2}{a_1 - a_2}\frac{\oo u}{\ot u}\oo\phi\, \ot\phi \right)      
\end{align}
and 
\begin{equation}
 [\oo \Phat, \Omega_\fk]\frac{\oo u \ot u}{a_1 a_2} = - \frac{1}{a_1 - a_2}[\oo\Jhat + \ot\Jhat, \frac{\oo u}{\ot u} \Omega_\fk + \Omega_\fh].
\end{equation}
The second follows immediately from \eqref{u_identity_1} and \eqref{uPOmega_id}. The first is related to \eqref{u_identity_1} and \eqref{u_identity_2} but is most easily 
verified by expanding the derivatives on the right side.

Applying these identities to the integrand of \eqref{I_rho_0_epsilon} one obtains 
\begin{align}
\lefteqn{ \oot{\rm tr}[\left(\tfrac{1}{2}\oo\phi \ot\phi + a_1 \di_{a_1} \oo\phi \ot\phi - \rho [\oo \Phat, \oo\phi \ot\phi]\right)\Omega_\fk]
 \frac{\oo u \ot u}{a_1 a_2}}\qquad\qquad\qquad\qquad & \nonumber\\ 
 = & - \frac{1}{a_1 - a_2} \frac{\oo u}{\ot u} [ 1 +  a_1 \di_{a_1} +  a_2 \di_{a_2}] \oot{\rm tr}[\oo\phi \ot\phi \Omega_\fk]  
     + \frac{1}{a_1 - a_2} \rho \:\oot{\rm tr}[[\oo\Jhat + \ot\Jhat, \oo\phi \ot\phi] \left(\frac{\oo u}{\ot u} \Omega_\fk + \Omega_\fh\right)]\nonumber \\
   & + \di_{a_1}\left(\frac{a_1}{a_1 - a_2}\frac{\ot u}{\oo u}\:\oot{\rm tr}[\oo\phi\, \ot\phi \Omega_\fk] \right) 
     + \di_{a_2}\left(\frac{a_2}{a_1 - a_2}\frac{\oo u}{\ot u}\:\oot{\rm tr}[\oo\phi\, \ot\phi \Omega_\fk] \right)\\
 = & - \di_\rho \left(\frac{1}{a_1 - a_2} \rho \:\oot{\rm tr}[\oo\phi \ot\phi \left(\frac{\oo u}{\ot u}  \Omega_\fk +  \Omega_\fh \right)]\right)
     + \frac{1}{a_1 - a_2} [ 1 +  a_1 \di_{a_1} +  a_2 \di_{a_2}] \oot{\rm tr}[\oo\phi \ot\phi \Omega_\fh] \nonumber \\
   & + \di_{a_1}\left(\frac{a_1}{a_1 - a_2}\frac{\ot u}{\oo u}\:\oot{\rm tr}[\oo\phi\, \ot\phi  \Omega_\fk] \right) 
     + \di_{a_2}\left(\frac{a_2}{a_1 - a_2}\frac{\oo u}{\ot u}\:\oot{\rm tr}[\oo\phi\, \ot\phi  \Omega_\fk] \right) \\
 = & - \di_\rho \left(\frac{1}{a_1 - a_2} \rho \:\oot{\rm tr}[\oo\phi \ot\phi \left(\frac{\oo u}{\ot u}  \Omega_\fk +  \Omega_\fh \right)]\right) \nonumber\\
   & + \di_{a_1}\left(\frac{a_1}{a_1 - a_2} \oot{\rm tr}[\oo\phi\, \ot\phi \left(\frac{\ot u}{\oo u}  \Omega_\fk +  \Omega_\fh\right)] \right) 
     + \di_{a_2}\left(\frac{a_2}{a_1 - a_2} \oot{\rm tr}[\oo\phi\, \ot\phi\left(\frac{\oo u}{\ot u}  \Omega_\fk +  \Omega_\fh\right)] \right)\label{simplified_integrand}
\end{align}

Note that the integrand of \eqref{I_rho_0_epsilon} is linear in the product $\varphi_1$ and $\varphi_2$. In the following we will resolve this product into an antisymmetric 
component $A$ and a symmetric component $S$ under the interchange of the two test functions $\varphi_1$ and $\varphi_2$, 
\begin{align}
 A(\rho;l_1,l_2) = & \oo T(\rho, l_1)\ot T(\rho, l_2) \frac{1}{2}\left(\oo\varphi_1(l_1) \ot\varphi_2(l_2) - \oo\varphi_2(l_1) \ot\varphi_1(l_2)\right) \ot T(l_2,\rho) \oo T(l_1,\rho)\\
 S(\rho;l_1,l_2) = & \oo T(\rho, l_1)\ot T(\rho, l_2) \frac{1}{2}\left(\oo\varphi_1(l_1) \ot\varphi_2(l_2) + \oo\varphi_2(l_1) \ot\varphi_1(l_2)\right) \ot T(l_2,\rho) \oo T(l_1,\rho),
\end{align}
and we will treat the two corresponding components of the integrand differently. The $A$ component will be expressed in the form \eqref{simplified_integrand}, that is as
\begin{align}
 \oot{\rm tr}[\left(\tfrac{1}{2}A + a_1 \di_{a_1} A - \rho [\oo \Phat, A]\right)\Omega_\fk]\frac{\oo u \ot u}{a_1 a_2} =
 & - \di_\rho \left( \rho \:\oot{\rm tr}[\frac{A}{a_1 - a_2} \left(\frac{\oo u}{\ot u}  \Omega_\fk +  \Omega_\fh \right)]\right) \nonumber\\
  & + \di_{a_1}\left(a_1 \oot{\rm tr}[\frac{A}{a_1 - a_2} \left(\frac{\ot u}{\oo u}  \Omega_\fk +  \Omega_\fh\right)] \right) 
    + \di_{a_2}\left(a_2 \oot{\rm tr}[\frac{A}{a_1 - a_2} \left(\frac{\oo u}{\ot u}  \Omega_\fk +  \Omega_\fh\right)] \right),\label{simplified_integrandA}
\end{align}
while the $S$ component will be simplified in a different way. 

Notice that $A/(a_1 - a_2)$ is smooth in $\rho$, $a_1$, and $a_2$ when $\rho> 0, a_i>1$ since 
$(\oo\varphi_1(\rho a_1) \ot\varphi_2(\rho a_2) - \oo\varphi_2(\rho a_1) \ot\varphi_1(\rho a_2))/(a_1 - a_2)$ is smooth. Both sides of \eqref{simplified_integrandA} are therefore 
continuous at $a_1 = a_2$, demonstrating that this equation holds in the entire domain of integration of \eqref{I_rho_0}, including this line.  
Since the $A$ component of the integrand of \eqref{I_rho_0} is a divergence, and the derivand in each term is smooth except at $a_1 = 1$ and $a_2 = 1$ which lie outside the domain of 
integration, the $A$ component of the integral $I_\eg(\rho_0)$ is a sum of boundary terms at $\rho = \rho_0$, $a_1 = 1 + \eg$, and $a_2 = 1 + \eg$. 

The latter two of these boundary terms vanish as $\eg$ tends to zero: Consider for instance the boundary term at $a_1 = 1 + \eg$,
\begin{equation}
 - \int_{\rho_0}^\infty d\rho \int_{1 + \eg}^\infty da_2  
 \ a_1 \oot{\rm tr}[\frac{A}{a_1 - a_2} \left(\frac{\ot u}{\oo u}  \Omega_\fk +  \Omega_\fh\right)]. 
\end{equation}
The $\Omega_\fk$ term in the integrand is the product of $\ot u = (1 - a_2^{-1})^{-\frac{1}{2}}$, which is locally integrable and independent of $a_1$, and a factor which is continuous in 
$a_1$, $a_2$ and $\rho$ and of compact support in the domain of integration. The $\Omega_\fh$ term is continuous and compactly supported. The whole integrand is thus bounded by an integrable 
$a_1$ independent function. It follows by dominated convergence that the limit of the integral is the integral of the limiting value of the integrand on the plane 
$a_1 = 1$, $a_2 > 1$, $\rho > \rho_0$, and this is zero. 

The $A$ contribution to $I(\rho_0)$ is therefore just the limit as $\eg \rightarrow 0^+$ of the $\rho = \rho_0$ boundary term,
\begin{equation}\label{A_rho_integral}
 \int_{1 + \eg}^\infty da_2 \int_{1 + \eg}^\infty da_1\  \rho_0 \:\oot{\rm tr}[\frac{A}{a_1 - a_2} \left(\frac{\oo u}{\ot u}  \Omega_\fk +  \Omega_\fh \right)]  
 = \int_{\rho_0(1 + \eg)}^\infty dl_2 \int_{\rho_0(1 + \eg)}^\infty dl_1\ \oot{\rm tr}[\frac{A}{l_1 - l_2} \left(\frac{\oo u}{\ot u}  \Omega_\fk +  \Omega_\fh \right)],
\end{equation} 
and the $A$ component of $\int_0^\infty H d\rho$ is obtained by taking the limit of this expression as both $\rho_0, \eg >0$ approach $0$. To evaluate this limit note that if 
$\rho_0(1 + \eg) < l_{min\,12}$, the minimum value attained by $l$ in the supports of the test functions $\varphi_1$ and $\varphi_2$, then the lower limits of integration in 
\eqref{A_rho_integral} may as well be set to $0$. Furthermore, if $\rho_0 < l_{min\,12}/2$ then $\frac{\oo u}{\ot u} < \sqrt{2}$. Thus for sufficiently small $\rho_0$, 
the integrand is bounded by a $\rho_0$ independent integrable function, and by dominated convergence the $\rho_0 \rightarrow 0^+$ limit of the integral is 
\begin{align}
\lefteqn{\int_0^\infty dl_2 \int_0^\infty dl_1\ \oot{\rm tr}[\frac{A(0;l_1,l_2)}{l_1 - l_2} (\Omega_\fk + \Omega_\fh)]} \qquad\qquad\qquad &\nonumber \\
= & \int_0^\infty dl_2 \int_0^\infty dl_1\ \frac{1}{l_1 - l_2}\: \oot{\rm tr}[\frac{1}{2}\left(\oo\varphi_1(l_1) \ot\varphi_2(l_2) - \oo\varphi_2(l_1) \ot\varphi_1(l_2)\right) 
\ot T(l_2,0) \oo T(l_1,0)  \Omega_\fg \oo T(0, l_1)\ot T(0, l_2)].\label{A_rho_integral3}
\end{align}

If we set $f(l_1, l_2) = \oot{\rm tr}[\oo\varphi_1(l_1) \ot\varphi_2(l_2) \ot T(l_2,0) \oo T(l_1,0)  \Omega_\fg \oo T(0, l_1)\ot T(0, l_2)]$ then, because $\Omega_\fg$
is symmetric with respect to interchange of the spaces $1$ and $2$, the left side of \eqref{A_rho_integral3} becomes
\begin{equation}
 \int_0^\infty dl_2 \int_0^\infty dl_1\ \tfrac{1}{2}[f(l_1,l_2) - f(l_2,l_1)]\frac{1}{l_1 - l_2} = \int_0^\infty dl_2 \int_0^\infty dl_1\ p. v. \left(\frac{1}{l_1 - l_2}\right)\:f(l_1,l_2)
\end{equation}
(by \eqref{Cauchy_pv_def2}).

This is of course the result we expect for the whole integral $\int_0^\infty H d\rho$, not just the $A$ component. But the $S$ component vanishes.
The $S$ component of $I_\eg(\rho_0)$ is 
\begin{equation}
 \int_{\rho_0}^\infty d\rho \int_{1 + \eg}^\infty da_2 \int_{1 + \eg}^\infty da_1
 \ \oot{\rm tr}[\Omega_\fk\left(\tfrac{1}{2}S + a_1 \di_{a_1} S - \rho [\oo \Phat, S]\right)] \frac{\oo u \ot u}{a_1 a_2}. 
\end{equation}
Since the domain of integration is symmetric under interchange of $a_1$ and $a_2$ we may symmetrize the integrand with respect to this interchange without changing the integral. 
Similarly the trace with $\Omega_\fk$ is symmetric under the interchange of the spaces $1$ and $2$, so we may also symmetrize 
$\tfrac{1}{2}S + a_1 \di_{a_1} S - \rho [\oo \Phat, S]$ with respect to this interchange without changing the result. The integrand may therefore be replaced by 
\begin{equation}
 \frac{1}{2} \oot{\rm tr}[\Omega_\fk\left([1 + a_1 \di_{a_1} + a_2 \di_{a_2}] S - \rho [\oo \Phat + \ot \Phat, S]\right)] \frac{\oo u \ot u}{a_1 a_2}.
\end{equation}
The commutator term may also be written as $- \rho [\oo \Jhat + \ot \Jhat, S]$, since $[\oo Q + \ot Q, \Omega_\fk] = 0$, so the symmetrized integrand equals
$\tfrac{1}{2}\di_\rho \left( \rho\: \oot{\rm tr}[\Omega_\fk S] \frac{\oo u \ot u}{a_1 a_2}\right)_{a_1,a_2}$. Applying once more the divergence theorem we see that the $S$ component of 
$I_\eg(\rho_0)$ consists of just a boundary term at $\rho = \rho_0$: 
\begin{equation}
 - \frac{1}{2}  \int_{1 + \eg}^\infty da_2 \int_{1 + \eg}^\infty da_1\ \rho_0\: \oot{\rm tr}[\Omega_\fk S] \frac{\oo u \ot u}{a_1 a_2}  
 = - \frac{1}{2} \rho_0  \int_{\rho_0(1 + \eg)}^\infty dl_2\: \int_{\rho_0(1 + \eg)}^\infty dl_1 \oot{\rm tr}[\Omega_\fk S] \frac{\oo u \ot u}{l_1 l_2}.
\end{equation}
This is $\rho_0$ times an integral that tends to a finite limit as $\eg, \rho_0 \rightarrow 0^+$. The $S$ component of $\int_0^\infty H d\rho$ therefore vanishes.

In sum, we have found that 
\begin{align}
  \lefteqn{\int_0^\infty dl_2 \int_0^\infty dl_1 \ \oot{\rm tr}[\oo\varphi \ot\varphi \,\oo \U{}^{-1} \ot \U{}^{-1}\{\oo \U, \ot \U\}]} \qquad\qquad\qquad & \nonumber\\  
  = & 16 \pi G_2 \int_0^\infty dl_2 \int_0^\infty dl_1 \:\oot{\rm tr}[\oo\varphi \ot\varphi \:
  p. v. \left(\frac{1}{l_1 - l_2}\right)\ot T(l_2,0) \oo T(l_1,0)  \Omega_\fg \oo T(0, l_1)\ot T(0, l_2)], \label{smeared_U_bracket1} 
\end{align}
which is precisely what we wished to demonstrate, namely \eqref{U_brack_1}.

Before going on to calculate the Poisson bracket between $\E$s, let us return to the problem of the generation of imaginary components of 
the conformal metric by the Poisson bracket. Recall that we found in subsection \ref{V_brack} that the Poisson bracket does not quite preserve 
the reality of the conformal metric $e$: Real functionals of the conformal metric can generate an imaginary contribution to $e$, a sum of 
two zero modes with imaginary coefficients. We pointed out that these imaginary modes are a nuisance rather than a catastrophe 
because they do not propagate off the initial data surface - they do not affect the solution defined by the initial data in the 
interior of the domain of dependence of $\cN$. We also claimed that this problem does not arise if the deformed conformal metric $\E$ is used as 
data in place of $e$, because $\E$ is insensitive to zero modes. Let us verify this last claim on $\cN_L - S_0$ (leaving aside the more subtle situation
at $S_0$). 

The variation of $\E$ generated by a real functional $F$ of $e$ on $\cN$ is
\begin{equation}
 \{F, \E\} = \{F, \U \U^t \} = \U[\U^{-1}\{F, \U\} + (\U^{-1}\{F, \U\})^t]\U^t = 2\U(\U^{-1}\{F, \U\})_\fk \U^t.
\end{equation}
Thus, if $\U$ is real the imaginary component of the variation of $\E$ generated by $F$ is proportional to the imaginary component of $(\U^{-1}\{F, \U\})_\fk$.
The logarithmic variation of $\U$ is in turn determined by the logarithmic variation of $\V$ via \eqref{Vhat_variation} with $\rho = l$. Specifically
\begin{equation}
(\U^{-1}\{F, \U\})_\fk(x) = \int_0^x T(x,z)\left(\frac{u}{l}\sqrt{\rho} D_z [\sqrt{\rho}(\V^{-1}\{F, \V\})_\fk]\right)_z T(z,x) dz.
\end{equation} 
At the end of subsection \ref{V_brack} we saw that $\operatorname{Im}(\V^{-1}\{F, \V\})_\fk$ is a sum of zero modes. On $\cN_L - S_0$ in particular
this reduces to a sum of the zero modes given by \eqref{zero_mode_L}, which we may call $\phi^L_A$ . But $D_z [\sqrt{\rho}\phi^L_A] = 0$, so these zero 
modes do not contribute to the $\fk$ component of the logarithmic variation of $\U$. $\E$ is thus insensitive to the zero modes in $\V$, which implies as a corollary 
that it is insensitive to the imaginary component of the variation of $\V$ generated by $F$. The Hamiltonian flow generated by a real functional of $e$
preserves the reality of $\E$. This is of course consistent with the absence of any imaginary terms in the bracket of $\E$ with $\E$ \eqref{poibraE}.

This result does not really contradict our claim that $\E$ determines $e$ uniquely. $\E$ determines a unique {\em real} $e$.
Furthermore, the zero modes of $e$ are not genuine initial data, since they do not affect the Cauchy development off $\cN$, so $\E$ together
with the data on $S_0$ is complete initial data.

\subsection{The Poisson algebra of the deformed conformal metric $\E$}\label{section:brackofE}

To obtain the brackets $\{\oo\E, \ot\E\}$ all that remains to do is to reexpress \eqref{U_brack_1}, 
\be
\oot C_{\sk\sk} = 16\pi G_2\: p.v.\left(\frac{1}{l_1-l_2}\right)\left[\To(x_1,0)\Tt(x_2,0)\Om_\fg\To(0,x_1)\Tt(0,x_2)\right]_{\sk\sk},
\ee
in a convenient form and substitute the result into \eqref{E_bracket2},
\be
\{\oo \E,\ot \E \}= 4\,\oo \U\ot \U \oot C_{\sk\sk} \,\oo \U{}^t\ot \U{}^t.
\ee

Let us define $\vOm = \oo Z\ot Z\Om_\fg \oo Z{}^{-1}\ot Z{}^{-1}$, where 
$Z_a{}^i$ is an arbitrarily chosen unit determinant zweibein. $\vOm_a{}^b{}_c{}^d$ is the tangent space tensor that corresponds, 
via the zweibein $Z_a{}^i$, to the internal space tensor $\Om_\fg{}_i{}^j{}_k{}^l$. Since $\Om_\fg$ is invariant under equal 
$G = SL(2,\Real)$ transformations in spaces $1$ and $2$ all choices of zweibein lead to the same $\vOm$. Indeed 
\begin{equation}
\vOm_a{}^b{}_c{}^d = \frac{1}{4}[\dg_a^d \dg^b_c - \frac{1}{2}\dg_a^b \dg_c^d] 
\end{equation}
regardless of the zweibein chosen. 

Since $Z$ is arbitrary we can choose it to be equal to $\V(0)$. Thus
$\Ve{}^{-1}(0)\Vz{}^{-1}(0)\vOm\Ve(0)\Vz(0) = \Om_\fg$ and 
\begin{align}
\oot C_{\sk\sk}
= p.v.\left(\frac{16\pi G_2}{l_1-l_2}\right) & \left[\To(x_1,0)\Tt(x_2,0)\Ve{}^{-1}(0)\Vz{}^{-1}(0)\vOm\Ve(0)\Vz(0)\To(0,x_1)\Tt(0,x_2)\right]_{\sk\sk}\\
= p.v.\left(\frac{16\pi G_2}{l_1-l_2}\right) & \left[\oo \U{}^{-1}(x_1)\ot \U{}^{-1}(x_2)\vOm\,\oo \U(x_1)\ot \U(x_2)\right]_{\sk\sk}\\
= p.v.\left(\frac{4\pi G_2}{l_1-l_2}\right) & \bigg(\oo \U{}^{-1}\ot \U{}^{-1}\vOm\,\oo \U\ot \U
+ \oo \U{}^{-1}\ot \U{}^t\vOm^t\,\oo \U(\ot \U{}^t)^{-1}\\
&\ + \oo \U{}^t\ot \U{}^{-1}\ {}^t\vOm (\oo\U{}^t)^{-1}\ot \U 
+ \oo \U{}^t\ot \U{}^t\ {}^t\vOm^t (\oo\U{}^t)^{-1} (\ot\U{}^t)^{-1} \bigg).
\end{align}
Inserting this expression for $C_{\sk\sk}$ in \eqref{E_bracket2}, produces the remarkably elegant result
\be\label{poibraE}
\{\oo \E,\ot \E\}= p.v.\left(\frac{16\pi G_2}{l_1-l_2}\right)\left(\vOm \oo\E\ot\E 
 + \oo\E\,{}^t\!\vOm \ot\E + \ot\E \vOm^t \oo\E + \oo\E\ot\E\,{}^t\!\vOm^t\right).
\ee
Expressed explicitly in terms of the components of $\E$ and the area density $\rho$ the bracket is
\be\label{poibraE_explicit}
\{\E_{ab}(\mathbf{1}),\E_{cd}(\mathbf{2})\}= p.v.\left(\frac{16\pi G_2}{\rho(\mathbf{1})-\rho(\mathbf{2})}\right)
\operatorname{Sym}_{(ab),(cd)} \left(\E_{ad}(\mathbf{1})\E_{cb}(\mathbf{2}) 
- \frac{1}{2}\E_{ab}(\mathbf{1})\E_{cd}(\mathbf{2}) \right),
\ee
where $\operatorname{Sym}_{(ab),(cd)}$ indicates that the expression must be symmetrized with respect to interchange of 
the indices in the pairs $a,b$ and $c,d$.

The Poisson bracket \eqref{poibraE} is equivalent to the Poisson bracket of the monodromy matrix $\M$ given in \cite{KorotkinSamtleben}. 
To obtain the bracket of \cite{KorotkinSamtleben} we first transform \eqref{poibraE} from the tangent space of the symmetry orbits to 
the internal space using an arbitrary, non-dynamical, unit determinant zweibein $Z$. We obtain a bracket between the internal components 
$\E_{ij}$ of $\E$ identical in form to \eqref{poibraE} but with $\vOm$ replaced by $\Om_\fg$. This bracket may be slightly simplified 
using the identities $\Om^t_\fg = -\Om^\eta_\fg = \Omk - \Omh = {}^t\!\Om_\fg$ and ${}^t\!\Om_\fg^t = \Omk + \Omh = \Omg$, 
yielding finally
\be\label{poibraE_int}
\{\oo \E,\ot \E\}= p.v.\left(\frac{16\pi G_2}{l_1-l_2}\right)\left(\Om_\fg \oo\E\ot\E + \oo\E\ot\E \Om_\fg
 - \oo\E\Om_\fg^\eta \ot\E - \ot\E \Om_\fg^\eta \oo\E \right).
\ee
This, in turn, is equivalent to a bracket on the internal components $\M_{ij}$ of the monodromy matrix:
Recall that $\E(r) = \M(\rho^-(r))$, with $\rho^- = 2\rho - \rho^+$, so $\M = \E \circ (\rho^-)^{-1}$. Since $\rho^-$ Poisson 
commutes with itself and $\E$, the bracket \eqref{poibraE} is equivalent to
\be\label{poibraM}
\{\oo \M,\ot \M\}= p.v.\left(\frac{32\pi G_2}{w_1-w_2}\right)\left(\Om_\fg \oo\M\ot \M+\oo\M\ot\M\Om_\fg
 -\oo\M\Om^\eta_\fg \ot\M-\ot\M \Om^\eta_\fg \oo\M\right).
\ee

The bracket \eqref{poibraM} agrees with that obtained by Korotkin and Samtleben \cite{KorotkinSamtleben} taking into account that 
they work with units such that $8\pi G_2 = 1$ and their $\Om_\fg$ is four times ours.
Nevertheless their result differs in two important respects from ours. First, their bracket is derived
in a completely different way, from a canonical formulation in terms of spacelike initial data. Second, their formalism 
assumes that space extends infinitely far from the symmetry axis, and is asymptotically flat in a suitable sense.
Our derivation involves only data on the finite null segment $\cN_L$, it makes no assumptions about the
existence or properties of more distant regions of spacetime. In this sense our result is more general than 
that of \cite{KorotkinSamtleben}.\footnote{
It is in fact mentioned in \cite{KorotkinSamtleben} that $\M$ can be defined without invoking spatial infinity, but 
this approach is not developed there. In their calculation of the brackets of $\M$ they define $\M$ in terms of fields
at spatial infinity.}
This is important because in the absence of symmetries four dimensional null canonical general relativity is difficult
to formulate except in a quasi-local form, in which the initial data hypersurface is truncated before the generators
form caustics.

\section{Definition and Poisson brackets of $\E$ in the absence of cylindrical symmetry}\label{symmetryless}

Our classical results apply quite directly to data on a double null sheet $\cN$ in full GR, without cylindrical symmetry, 
provided $\cN$ satisfies a stringent regularity condition. If the generators of a branch of $\cN$ meet at a caustic at which 
$\mu$ does not diverge, then the change of variables $\mu \mapsto \E$ defined in section \ref{transformation} may be 
applied, unchanged, to $\mu$ along each generator of the branch. The resulting deformed conformal metric, which now 
depends on the transverse coordinates $\theta^1, \theta^2$, satisfies the Poisson brackets
\be\label{poibraE_4}
\{\oo \E,\ot \E\}= p.v.\left(\frac{16\pi G}{\rho(\mathbf{1})-\rho(\mathbf{2})}\right)\dg^2(\theta_\mathbf{2} - \theta_\mathbf{1})\left(
\vOm \oo\E\ot\E  + \oo\E\,{}^t\!\vOm \ot\E + \ot\E \vOm^t \oo\E + \oo\E\ot\E\,{}^t\!\vOm^t\right),
\ee
obtained by replacing $G_2$ by $G\dg^2(\theta_\mathbf{2} - \theta_\mathbf{1})$ in \eqref{poibraE}. Or equivalently
\be\label{poibraE_explicit_4}
\{\E_{ab}(\mathbf{1}),\E_{cd}(\mathbf{2})\}= p.v.\left(\frac{16\pi G}{\rho(\mathbf{1})-\rho(\mathbf{2})}\right) 
\dg^2(\theta_\mathbf{2} - \theta_\mathbf{1})
\operatorname{Sym}_{(ab),(cd)} \left(\E_{ad}(\mathbf{1})\E_{cb}(\mathbf{2}) 
- \frac{1}{2}\E_{ab}(\mathbf{1})\E_{cd}(\mathbf{2}) \right).
\ee

Generically $\mu$ does diverge in a caustic. On the other hand the conformal metric does not diverge 
along the generators of a past light cone as one approaches the vertex. As noted in section \ref{data0} this provides a 
simple way to construct double null sheets satisfying the regularity condition everywhere in a smooth spacetime: Choose 
two points such that their past light cones intersect, and define $S_0$ to be a disk in this intersection. The generators 
of the light cones that connect $S_0$ to the vertices sweep out the double null sheet. A formalism restricted to double 
null sheets satisfying the regularity condition may therefore suffice to describe arbitrary spacetimes. On the other hand, 
it is also likely that the formalism can be extended to non-regular double null sheets.  

\section{Quantization}\label{quantization}

Korotkin and Samtleben \cite{KorotkinSamtleben} have proposed an associative $*$-algebra that quantizes 
\eqref{poibraM} in the vacuum gravity case: It is generated by the quantum monodromy matrix $\M$, a symmetric matrix of 
self adjoint operators depending on the spectral parameter $w$, which satisfy exchange relation
\be\label{quantum_exchange_relation}
R(v - w)\oo \M(v) R'(w - v + 2i a) \ot \M(w)
= \ot \M(w)R'(v - w + 2i a)\oo \M(v)R(w - v)\frac{v - w - 2ia}{v - w + 2ia}, 
\ee
where $a = 4\pi G_2\hbar$, and $R(u) = (u - ia/2)I - 4ia \Om_\fg$ and $R'(u) = (u - i a/2)I - 4ia \Om^\eta_\fg$, with 
$I = \oo\One \ot\One$ the identity map on the product of spaces 1 and 2. 

The definition of the $*$-algebra must be completed by a quantization of the classical condition $\det \M = 1$
which is compatible with the other relations defining the algebra. Korotkin and Samtleben 
impose such a condition, not directly on $\M$ itself but on $T_+$ and $T_-$, a pair of operators 
depending on data all the way out to spatial infinity. (Their classical counterparts are defined in 
section \ref{KS_variables} of the present work). $T_\pm$ are not available to us, since our initial 
data surface, $\cN$, is compact and so knows nothing of spatial infinity. Their condition might 
still be usable if it were translated into a condition directly on $\M$, but this is not straightforward 
and will be left to future investigations. 

The exchange relation \eqref{quantum_exchange_relation} refers to the internal components of $\M$, which depend on the 
arbitrarily chosen reference zweibein $Z_a^i$. The exchange relation can, however, be rewritten in a manifestly $Z$ independent form
in terms of tangent space tensors:
\be\label{quantum_exchange_relation2}
\vR (v - w)\oo \M(v) {\vR}'(w - v + 2i a) \ot \M(w)
= \ot \M(w)\,{}^t\!{\vR}'^t(v - w + 2i a)\oo \M(v)\,{}^t\!{\vR}^t(w - v)\frac{v - w - 2ia}{v - w + 2ia}, 
\ee
with $\vR (u)_a{}^b{}_c{}^d = (u - ia/2)\dg_a^b\dg_c^d - 4ia \vOm_a{}^b{}_c{}^d = u \dg_a^b\dg_c^d - i a\dg_a^d\dg_c^b$ and 
${\vR}'(u)^a{}_{b c}{}^d = (u - ia/2)\dg_b^a\dg_c^d + 4ia \vOm_b{}^a{}_c{}^d = (u - i a)\dg^a_b\dg_c^d + i a\dg^a_c\dg_b^d$.

The quantization is extended to $\rho$ by setting the commutator of $\rho$ with the quantum monodromy matrix $\M$ to zero,
as suggested by the fact that classically $\rho$ Poisson commutes with $\M$. A quantization of the Poisson bracket 
\eqref{poibraE} of the deformed conformal metric $\E(\cdot) = \M(\rho^-(\cdot))$ can then be read off immediately from the 
exchange relation \eqref{quantum_exchange_relation2} for $\M$: 
\begin{equation}\label{quantum_exchange_relationE}
\vR(\Delta)\oo \E(\mathbf{1}) {\vR}'(-\Delta + 2ia)\ot 
\E(\mathbf{2})= \ot \E(\mathbf{2})\,{}^t\!{\vR}'^t(\Delta + 2ia)\oo 
\E(\mathbf{1})\,{}^t\!{\vR}^t(-\Delta) \frac{\Delta - 2ia}{\Delta + 2ia},
\end{equation}
with $\Delta = \rho^-(\mathbf{1}) - \rho^-(\mathbf{2})= 
2[\rho(\mathbf{1}) - \rho(\mathbf{2})]$. The symmetry and reality conditions on $\M$ imply that 
\begin{align}
\E_{ab} &= \E_{ba}\\ 
\E_{ab}^* &= \E_{ab}.
\end{align}
The quantization of the classical condition $\det\E = 1$ of course remains to be determined. 

That the exchange relation \eqref{quantum_exchange_relation2} indeed quantizes the Poisson bracket
\eqref{poibraE} can be verified directly by expanding the exchange relation to first order in $a$.

Note that $a$, like $\rho$ and $\Delta$, is an area density: because $G_2$ is Newton's constant 
$G$ divided by the $\theta$ coordinate area of $S_0$, $a$ is $16\pi$ times the Planck area divided 
by this coordinate area. That is, $a$ is an area density that assigns $16\pi$ Planck areas to $S_0$. 

What has been obtained (modulo the quantization of the condition $\det\E = 1$) is an algebraic 
quantization. The exchange relations, equivalent to commutation 
relations, have been specified exactly, as have further relations 
defining the $*$-algebra of the quantum deformed conformal metric. 
However, no unitary representation of the 
algebra by operators in a Hilbert space has been given. The 
specification of such a representation is probably necessary in order 
to complete the quantization. $*$-algebras often admit 
several unitarily inequivalent unitary representations (this is 
the case for instance for the $*$-algebra of initial data in 
canonically quantized free field theory) and unitarily 
inequivalent representations define distinct theories because 
the possible assignments of expectation values to the 
set of observables differ between such representations. See \cite{Wald_CSQFT} 
for a detailed discussion of this issue.

Korotkin and Samtleben do propose a representation of their 
quantum algebra in \cite{KorotkinSamtleben_PRL} but they were not 
able to show that it is unitary, i.e. that the $*$ operation in 
the algebra is mapped to the adjoint operation in the representation 
\cite{KorotkinSamtleben}.   

When one tries to take over the quantization \eqref{quantum_exchange_relationE} to the symmetryless case a difficulty arises.
To pass from \eqref{poibraE} to \eqref{poibraE_4} one substitutes $G_2$, Newton's constant divided by the coordinate area of $S_0$,
by $G\dg^2(\theta_\mathbf{2} - \theta_\mathbf{1})$. To generalize \eqref{quantum_exchange_relationE} one should therefore replace 
$a$ by $4\pi G\hbar \dg^2(\theta_\mathbf{2} - \theta_\mathbf{1})$.
However the resulting expression is not well defined, because the exchange relation \eqref{quantum_exchange_relationE} is not linear 
in $a$.

\section{Acknowledgements}

The present work has developed from the Diplom thesis of one of the authors (A.F.) completed under the direction of 
the other author (M.R.). The main results, which are included in the thesis, are joint work of A.F. and M.R.. The present 
article was written by M.R., who added several refinements and extensions to the results in the process. 

The authors thank Prof. Herbert Balasin for bringing about their collaboration, and A. F. thanks the 
Physics Institute of the Science Faculty of the Universidad de la Republica of Uruguay for hospitality during an extended stay 
there, and the International Office of the Technische Universit\"{a}t Wien for financial support during this visit. 
M.R. thanks the Centro de Ciencias Matem\'aticas de la UNAM in Morelia, Mexico, the Perimeter Institute in Waterloo, Canada 
and the Centre de Physique Th\'eorique in Luminy, France, for hospitality during the course of this work, and 
Jose Antonio Zapata, Bianca Dittrich, and Alejandro Perez for stimulating discussions. 
In an early stage of this project the work of M.R. was partially supported by the Foundational Questions Institute 
under grant RFP2-08-24, later it was partially supported by PAPIIT-Universidad Nacional Aut\'onoma de M\'exico through 
grant IN109415.

\appendix

\section{Properties of the path ordered exponential}\label{path_ordered_exp}

The results presented here are well known, at least in outline, but we have found no reference presenting them in the precise form that we need.

The path ordered exponential is the holonomy defined by a connection $A$ along a curve. Here and in the following $A$ will be a Lebesgue 
integrable function on a (possibly infinite) segment $I$ of the real line taking values in the complex square matrices of some fixed finite dimensionality $n$.

\begin{definition}\label{path_order_def}
The path ordered exponential of $A$ from $a \in I$ to $b \in I$ is defined by the power series
\begin{equation}\label{path_order_def_eq}
T(a,b) \equiv \cP e^{\int_{a}^{b} A dz} = \One + \sum_{s = 1}^\infty \int_a^b dz_1 \int_{z_1}^b dz_2 ... \int_{z_{s-1}}^b dz_s A(z_1)...A(z_s).
\end{equation}
\end{definition}

If $a \leq b$ then
\begin{equation}\label{path_order_def_eq1}
 \cP e^{\int_{a}^{b} A dz} = \One + \sum_{s = 1}^\infty \int_{a < z_1 < ...< z_s < b} A(z_1)...A(z_s) dz_1 ... dz_s,
\end{equation}
while if $a \geq b$
\begin{equation}\label{reverse_path_order_def_eq1}
 \cP e^{\int_{a}^{b} A dz} = \One + (-1)^s\sum_{s = 1}^\infty \int_{a > z_1 > ...> z_s > b} A(z_1)...A(z_s) dz_1 ... dz_s.
\end{equation}

Note that in our definition the exponential is ordered from left to right in the sense that factors of $A$ with argument closer to the lower 
bound of integration, $a$, appear to the left of factors with arguments closer to $b$, the upper bound of integration. The opposite ordering is 
often used in the literature and leads to expressions that differ by transpositions of factors from the ones obtained here. 

Note also that the definition applies even when $a$ or $b$, or both, are infinite, provided the interval $I$ on which $A$ is integrable 
includes these points.

The most elementary properties of the path ordered exponential can be obtained directly from this definition.

\begin{proposition}\label{path_ordered_exponential_prop}
The path ordered exponential $T(a,b) = \cP e^{\int_{a}^{b} A dz}$ is well defined and continuous in $a$ and $b$ for all $a, b$ in the interval $I$. 
Moreover $T$ satisfies the bound $\norm{T(a,b) - \One} \leq e^{|\int_a^b \norm{A} dz|} - 1$ for any submultiplicative norm $\norm{\cdot}$, and the product relation 
$T(a,b)T(b,c) = T(a,c)$ for all $a, b , c \in I$.  
\end{proposition}

Note that for all $a,b \in I$, finite or not, continuity requires that $T(a',b') \rightarrow T(a,b)$ when $a', b' \rightarrow a, b$.

\begin{proof}
By \eqref{path_order_def_eq}
\begin{equation}\label{path_order_def_eq2}
\cP e^{\int_{a}^{b} A dz} - \One =  \sum_{s = 1}^\infty \int_a^b dz_1 \int_{z_1}^b dz_2 ... \int_{z_{s-1}}^b dz_s A(z_1)...A(z_s).
\end{equation}
Because $A$ is integrable this series is absolutely convergent: Let $\norm{\cdot}$ be a submultiplicative norm on $n \times n$ matrices
(such as for instance $\norm{A}^2 = \sum_{i,j} |A_i{}^j|^2$) then the sum of the norms of the terms in the series \eqref{path_order_def_eq2} is
\begin{align}
\sum_{s = 1}^\infty \norm{\int_a^b dz_1 \int_{z_1}^b dz_2 ... \int_{z_{s-1}}^b dz_s A(z_1)...A(z_s)}
& \leq \sum_{s = 1}^\infty \left|\int_a^b dz_1 \int_{z_1}^b dz_2 ... \int_{z_{s-1}}^b dz_s \norm{A(z_1)}...\norm{A(z_s)}\right|\\
& = \sum_{s = 1}^\infty \frac{1}{s!} \left|\int_{a}^{b} \norm{A(z)} dz \right|^s.\label{norm_sum}
\end{align}
Since $A$ is integrable $|\int_{a}^{b} \norm{A(z)} dz| < \infty$ and the series \eqref{norm_sum} converges to
the finite value $e^{|\int_{a}^{b} \norm{A(z)} dz|} - 1$. Since the space of $n \times n$ matrices is Cauchy
complete this implies that the original series \eqref{path_order_def_eq2} converges. It also implies that
\begin{equation}\label{path_ordered_exp_bound}
 \norm{\cP e^{\int_{a}^{b} A dz} - \One} \leq e^{|\int_{a}^{b} \norm{A} dz|} - 1.
\end{equation}

When $a\leq b \leq c$ the product relation follows from the expression \eqref{path_order_def_eq1}: The order by
order product series (that is, the Cauchy product) of the series for $T(a,b)$ and $T(b,c)$ is the series for $T(a,c)$, 
and since the series of the two factors converge absolutely, by Mertens' theorem the Cauchy product converges
to the product of the two factors. See \cite{Hardy}. 

To complete the proof of the product relation it is sufficient to demonstrate that $T(a,b)T(b,a) = \One$ when $a \leq b$.
This product can be expressed as $\tilde{T}(-1,0)\tilde{T}(0,1)$ where $\tilde{T}$ is the path ordered exponential of the connection
\begin{equation}
     \tilde{A}(t)= 
\begin{cases}
    (b - a) A(b + t(b-a)), & \text{if } t \in [-1,0]\\
    (b - a) A(b - t(b-a)), & \text{if } t \in [0,1],
\end{cases}
\end{equation}
which by the preceding result equals $\tilde{T}(-1,1) = \One + \sum_{s = 1}^\infty \int_{-1 < t_1 < ...< t_s < 1} \tilde{A}(t_1)...\tilde{A}(t_s) dt_1 ... dt_s$.
All terms in this series save the first, $\One$, vanish: Consider the order $s$ term. At each point of the domain of integration $\{-1 < t_1 < ...< t_s < 1\}$  
at least one of the variables $t_i$ will have the smallest absolute value. The domain is therefore the union, disjoint modulo intersections of measure zero, of
the sets $U_i = \{-1 < t_1 < ...< t_s < 1, |t_i| \leq |t_j| \forall j \}$. A sequence $[t_1, ..., t_n]$ belongs to $U_i$ iff $-1 < t_1 < ...<t_{i-1} < 0 < t_{i+1} < .... < t_s < 1$ 
and $|t_i| \leq m \equiv \min(|t_{i-1}|, |t_{i+1}|)$. The integral $\int_{U_i} \tilde{A}(t_1)...\tilde{A}(t_s) dt_1 ... dt_s$ therefore contains a factor
$\int_{-m}^m \tilde{A}(t_i) dt_i$, which vanishes because $\tilde{A}$ is an odd function of $t$.

The continuity of the path ordered exponential now follows from the product relation
and the bound \eqref{path_ordered_exp_bound} which implies that 
$\cP e^{\int_{a}^{b} A dz} \rightarrow \One$ when $b \rightarrow a$.
\end{proof}

A more powerful way to characterize the path ordered exponential is via an integral equation, or rather either 
of two equivalent integral equations.

\begin{proposition}\label{integral_eqn}
For any $a \in I$ the path ordered exponential $T(a,b) = \cP e^{\int_{a}^{b} A dz}$ is the only solution to the integral equation
\begin{equation}\label{int_eq_TA}
 T(a,b) = \One + \int_a^b T(a,z)A(z)dz\ \ \forall b\in I 
\end{equation}
that is bounded as a function of $b$. For any fixed $b \in I$ it is also the unique solution to the integral equation
\begin{equation}\label{int_eq_AT} 
 T(a,b)	= \One + \int_a^b A(z)T(z,b)dz\ \ \forall a\in I  
\end{equation}
that is bounded as a function of $a$ on $I$. 
\end{proposition}

\begin{proof}
 First let us demonstrate that the path ordered exponential satisfies the integral equations \eqref{int_eq_TA} and \eqref{int_eq_AT}. Let $T_r(a,b)$ be the $r$th partial
 sum of the series \eqref{path_order_def_eq} for $T(a,b)$. It is clear from the expressions \eqref{path_order_def_eq1} and \eqref{reverse_path_order_def_eq1} for this 
 series that
 \begin{equation}
  T_r(a,b) = \One + \int_a^b T_{r-1}(a,z)A(z)dz.
 \end{equation}
 But as $r \rightarrow \infty$ $T_{r-1}(a,z)A(z)$ converges pointwise to $T(a,z)A(z)$ and 
 \begin{equation}
  \norm{T_{r-1}(a,z)A(z)} \leq \norm{A(z)} + (e^{|\int_{a}^{z} \norm{A} dz'|} - 1)\norm{A(z)} 
 \end{equation}
 by prop. \ref{path_ordered_exponential_prop}, so it is bounded by an integrable function. It follows from the dominated convergence theorem that 
 $T(a,b) = \lim_{r \rightarrow \infty} T_r(a,b)$ satisfies \eqref{int_eq_TA}. The demonstration of \eqref{int_eq_AT} is analogous.
 
 Prop. \ref{path_ordered_exponential_prop} assures that $T(a,b)$ is bounded, since it is continuous and $I$ is either compact or $T(a,b)$ has finite limiting values 
 as $a$ or $b$ approach infinity.
 
 Now let us prove that \eqref{int_eq_TA} and boundedness implies that $T$ is the path ordered exponential. We therefore suspend for a moment the definition of $T$
 as the path ordered exponential, supposing only that it is a bounded solution of \eqref{int_eq_TA}, and show that this implies that $T(a,b) 
 = \lim_{r \rightarrow \infty} T_r(a,b) \equiv \cP e^{\int_{a}^{b} A dz}$ (where $T_r(a,b)$ is still defined as the $r$th partial sum of \eqref{path_order_def_eq}).  
 Substituting \eqref{int_eq_TA} into itself iteratively $r-1$ times one obtains
 \begin{equation}
  T(a,b) - T_r(a,b) = \int_a^b dz_1 \int_{z_1}^b dz_2 ... \int_{z_{r-1}}^b dz_r T(a,z_1)A(z_1) ...A(z_r).
 \end{equation}
 Since $T(a,z)$ is bounded on $I$: $\norm{T(a,z)}< M\ \ \forall z\in [a,b]$ for some finite $M$. Thus
 \begin{equation}
  \norm{T(a,b) - T_r(a,b)} \leq \frac{1}{r!}\left|\int_a^b \norm{A} dz\right|^r M,
 \end{equation}
 implying that $T(a,b) = \lim_{r \rightarrow \infty} T_r(a,b)$. A similar argument shows that \eqref{int_eq_AT} and the boundedness of $T$ also implies that
 $T$ is the path ordered exponential.
\end{proof}

Note that when $a$ and $b$ are finite \eqref{int_eq_TA} (or \eqref{int_eq_AT}) implies that $T$ is the path ordered exponential also under the weaker 
hypothesis that $T(a,b)$ is locally integrable in $b$ since then this equation implies that $T$ is the integral of a locally integrable function, and thus continuous. 

The indefinite integrals of Lebesgue integrable functions are more than just continuous, they are {\em absolutely continuous}. Absolutely continuous functions
are differentiable almost everywhere and, by the fundamental theorem of calculus for the Lebesgue integral (Theorem 7.20 of \cite{Rudin87}), the integral over 
an interval $[a,b]$ of the derivative of such a function $f$ is equal to $f(b) - f(a)$.    

The path ordered exponential may therefore be characterized as the solution to a differential equation. 

\begin{proposition}\label{path_ordered_exponential_prop_derivative}
The path ordered exponential $T(a,b) = \cP e^{\int_{a}^{b} A dz}$ satisfies the differential equation 
\begin{equation}\label{holonomy_diff_eq}
 dT/db = T A
\end{equation} 
at almost all $b \in I$. Moreover, it is the only absolutely continuous function that satisfies this equation almost everywhere and takes the value 
$\One$ when $b = a$. If $A$ is continuous then \eqref{holonomy_diff_eq} holds everywhere on $I$. 
\end{proposition}

\begin{proof}
Equation \eqref{holonomy_diff_eq} follows from differentiating \eqref{int_eq_TA} in $b$. By the fundamental theorem of calculus for the Lebesgue integral 
the derivative is well defined and equal to the integrand on the right side of \eqref{int_eq_TA} almost everywhere.
Conversely, suppose that $T(a,b)$ satisfies \eqref{holonomy_diff_eq} almost everywhere, and that it is absolutely continuous. Then $T(a,b) - T(a,a) = \int_a^b T(a,z)A(z)dz$, 
again by the fundamental theorem of calculus for the Lebesgue integral. The condition $T(a,a) = \One$ then implies that  $T(a,b) = \One + \int_a^b T(a,z)A(z)dz$. 

If $A$ is continuous the integrand of \eqref{int_eq_TA} is continuous so the ordinary fundamental theorem of calculus implies that \eqref{holonomy_diff_eq} holds everywhere.
\end{proof}

Path ordered exponentials satisfy a convergence theorem similar to the dominated convergence theorem for ordinary Lebesgue integrals.

\begin{proposition}\label{path_ordered_exponential_prop2}
Suppose $A_m$ is a sequence of integrable matrix valued functions on a possibly infinite interval $[a,b]$ that converges pointwise to 
the function $A_\infty$. Suppose, moreover, that there exists an integrable function $g$ such that $\norm{A_m} \leq g$ for all $m$. 
Then the path ordered exponential of $A_m$ from $a$ to $b$ converges to the path ordered exponential of $A_\infty$.
\end{proposition}
\begin{proof}
The idea is to apply Lebesgue's dominated convergence theorem \cite{Royden} twice: first to the integral in each term
of the series expansion of the path ordered exponential of $A_m$, to show that it converges to the corresponding term in the 
path ordered exponential of $A_\infty$, and then to the series to show that the limit of the sum converges to the sum for
$A_\infty$. We start with the integrals: We know that $A_m(z_1)...A_m(z_s) \rightarrow A_\infty(z_1)...A_\infty(z_s)$ as 
$m \rightarrow \infty$, and that $\norm{A_m(z_1)...A_m(z_s)} \leq g(z_1)...g(z_s)$, 
with $\int_a^b dz_1 ... \int_{z_{s-1}}^b dz_s\: g(z_1)...g(z_s) < \infty$. The dominated convergence theorem then shows that 
\begin{equation}
\int_a^b dz_1 ... \int_{z_{s-1}}^b dz_s\: A_m(z_1)...A_m(z_s)
\rightarrow \int_a^b dz_1 ... \int_{z_{s-1}}^b dz_s\: A_\infty(z_1)...A_\infty(z_s).
\end{equation}
These are of course the terms in the expansion of the path ordered exponentials.

Because $A_m$ is bounded by $g$
\begin{equation}
 \norm{\int_a^b dz_1 ... \int_{z_{s-1}}^b dz_s\: A_m(z_1)...A_m(z_s)} \leq \left|\int_a^b dz_1 ... \int_{z_{s-1}}^b dz_s\: g(z_1)...g(z_s)\right|
 = \frac{1}{s!}\left|\int_a^b g dz\right|^s.
\end{equation}
This bound can be summed from $s = 1$ to $\infty$, yielding $e^{|\int_a^b g dz|} - 1 < \infty$. Applying the dominated convergence theorem to the series for 
the path ordered integral we conclude that
\begin{align}
 \cP e^{\int_{a}^{b} A_m(z) dz} - \One
 & \equiv \sum_{s=1}^\infty \int_a^b dz_1 ... \int_{z_{s-1}}^b dz_s\: A_m(z_1)...A_m(z_s) \nonumber\\
 & \rightarrow \sum_{s=1}^\infty \int_a^b dz_1 ... \int_{z_{s-1}}^b dz_s\: A_\infty(z_1)...A_\infty(z_s) \equiv
 \cP e^{\int_{a}^{b} A_\infty(z) dz} - \One.
\end{align}
\end{proof}

Path ordered exponentials are holonomies and they transform in a simple way under ``gauge transformations'', that is, local changes of basis in the $n$ dimensional space in 
which the matrix connection $A$ acts. 

\begin{proposition}\label{gauge_transformation}
Suppose $\Lambda(x) = \Lambda(a)\:\cP e^{\int_{a}^{x} \lam dz}$ with $\lambda$ an integrable $n \times n$ matrix valued function on $[a,b]$ and $\Lambda(a)$ 
an invertible $n \times n$ matrix then
\begin{equation}
\Lambda(a) T(a,b) \Lambda^{-1}(b) = \cP e^{\int_{a}^{b} \Lambda(A - \lambda)\Lambda^{-1} dz}
\end{equation}
\end{proposition}
\begin{proof}
 By prop. \ref{path_ordered_exponential_prop} both $\Lambda(x)$ and $T(a,x)$ are bounded functions of $x$. It follows that $\Lambda(a) T(a,x) \Lambda^{-1}(x)$ is also. 
 By prop. \ref{integral_eqn} it therefore suffices to show that
 \begin{equation}\label{integral_eq_gauge_transformation}
  \Lambda(a) T(a,x) \Lambda^{-1}(x) = \One + \int_a^x \Lambda(a) T(a,z) \Lambda^{-1}(z) \Lambda(z) (A - \lambda)(z) \Lambda^{-1}(z) dz
 \end{equation}
 for all $x \in [a,b]$ to demonstrate the claim. 
 
 Now 
 \begin{equation}
  \Lambda^{-1}(z)\Lambda(x) = \cP e^{\int_{z}^{x} \lam dz'} = \One + \int_{z}^{x} \lam(z')\Lambda^{-1}(z')\Lambda(x)dz'
 \end{equation}
 by the integral equation \eqref{int_eq_AT} of prop. \ref{integral_eqn}. Using the other equation equation, \eqref{int_eq_TA}, to expand $T(a,z)$ one obtains the following expression for 
 the integral on the right of \eqref{integral_eq_gauge_transformation}: 
 \begin{align}
  \lefteqn{\int_a^x \Lambda(a) T(a,z)(A - \lambda)(z) \Lambda^{-1}(z) dz}\qquad\qquad\qquad \\
  = \Lambda(a) & \left\{ \int_a^x T(a,z)A(z) dz + \int_a^x dz \int_z^x dz'\: T(a,z)A(z)\lambda(z')\Lambda^{-1}(z')\Lambda(x)\right.\notag \\
		    & \left.  - \int_a^x \lam(z)\Lambda^{-1}(z)\Lambda(x)dz - \int_a^x dz' \int_a^{z'} dz\: T(a,z)A(z)\lambda(z')\Lambda^{-1}(z')\Lambda(x)\right\}\Lambda^{-1}(x).
 \end{align}
 Since the second and fourth terms cancel this equals
 \begin{align}
  \Lambda(a) & \left\{ \One + \int_a^x T(a,z)A(z) dz - \One - \int_a^x \lambda(z)\Lambda^{-1}(z)\Lambda(x)dz\right\} \Lambda^{-1}(x) \\
  & = \Lambda(a)\left\{T(a,x) - \Lambda^{-1}(a)\Lambda(x)\right\}\Lambda^{-1}(x)\\
  & = \Lambda(a)T(a,x)\Lambda^{-1}(x) - \One.
 \end{align}
 This proves \eqref{integral_eq_gauge_transformation}, and thus the proposition.
\end{proof}

With the previous proposition in hand we are ready to calculate the functional derivative of the path ordered exponential with respect to the connection.

\begin{proposition}\label{path_ordered_exponential_functional_derivative}
At each point in the Banach space $L^1(I)$ of integrable connections $A$ on $I$ the derivative $\dI T(a,b)$ of the path ordered exponential by the 
connection exists as a bounded linear operator from $L^1(I)$ to the $n \times n$ matrices. The contraction of the derivative with any variation of the 
connection $\dg A \in L^1(I)$ is  
\begin{equation}
 \dI T(a,b) \cont \dg A = \int_a^b T(a,z) \dg A(z) T(z,b) dz.
\end{equation}
\end{proposition}

\begin{proof}
Let us first prove the claim at the trivial connection $A_0 = 0$. The path ordered exponential $T$ of a (not necessarily zero) connection $A \in L^1(I)$ satisfies
\begin{equation}
 T(a,b) = \One + \int_a^b A(z)T(z,b)dz = \One + \int_a^b A(z)dz + \int_a^b A(z)[T(z,b) - \One]dz.
\end{equation}
Thus
\begin{align}
 \norm{T(a,b) - \One - \int_a^b A(z)dz} & \leq \left|\int_a^b \norm{A(z)}\norm{T(z,b) - \One}dz\right|\\
			                & \leq \left|\int_a^b \norm{A(z)}(e^{|\int_z^b \norm{A(z')}dz'|} - 1)dz\right|\\
					& \leq K(e^K - 1).
\end{align}
with $K = \int_I \norm{A(z)}dz$. It follows that for all $\eg>0$ there exists a $\dg > 0$ such that 
\begin{equation}
 \norm{T(a,b) - \One - \int_a^b A(z)dz} \leq \eg K 
\end{equation}
when $K < \dg$. Because $K$ is the $L^1(I)$ norm of $A$ this means that $T(a,b)$ has a derivative at $A_0 = 0$ 
(that is, a linear approximation) defined by $\dI T(a,b)\cont \dg A = \int_a^b \dg A(z) dz\ \forall \dg A \in L^1(I)$.
It is clear that this is a bounded operator in the $L^1(I)$ operator norm.

Now let us prove the claim at a non-zero connection $A_0 \in L^1(I)$. This will be done essentially by acting on the connections with the ``gauge transformation'' 
of prop. \ref{gauge_transformation}, so that $A_0$ is mapped to zero, and applying the preceding result. Specifically we set $\lam = A_0$ and $\Lambda(a) = \One$, 
so that $\Lambda(z) = T_0(a,z) \equiv \cP e^{\int_{a}^{z} A_0 dz'}$. Then proposition \ref{gauge_transformation} shows that the path ordered exponential $T$ of any 
connection $A \in L^1(I)$ satisfies 
\begin{equation}
T(a,b) T_0(b,a) = \cP e^{\int_{a}^{b} \Delta dz},
\end{equation}
with $\Delta(z) \equiv T_0(a,z)(A - A_0)T_0(z,a)$. It follows that
\begin{equation}
T(a,b)T_0(b,a) = \One + \int_a^b \Delta(z)dz + \int_a^b \Delta(z)[\cP e^{\int_{z}^{b} \Delta dz'} - \One]dz,
\end{equation}
and thus that
\begin{equation}
T(a,b) = T_0(a,b) + \int_a^b T_0(a,z)(A - A_0)T_0(z,b)dz + \int_a^b \Delta(z)[\cP e^{\int_{z}^{b} \Delta dz'} - \One]dz\  T_0(a,b).
\end{equation}
The norm of $\Delta(z)$ is bounded by $e^{2|\int_a^z \norm{A_0} dz'|}\norm{A(z) - A_0(z)}$, so
\begin{equation}
\norm{T(a,b) - T_0(a,b) - \int_a^b T_0(a,z)(A - A_0)T_0(z,b)dz} 
\leq M^3 K (e^{M^2 K} - 1),
\end{equation}
with $M = e^{\int_I \norm{A_0} dz}$ and $K = \int_I \norm{A - A_0} dz$. As before, for all $\eg > 0$ there exists $\dg > 0$ such that
\begin{equation}
\norm{T(a,b) - T_0(a,b) - \int_a^b T_0(a,z)(A - A_0)T_0(z,b)dz} \leq \eg K,
\end{equation}
when $K < \dg$. Since $K$ is the $L^1(I)$ norm of $A - A_0$ this means that $T(a,b)$ is differentiable at $A_0$ with derivative determined by
$\dI T(a,b)\cont \dg A = \int_a^b T_0(a,z) \dg A(z) T_0(z,b) dz$. Again it is clear that this operator is bounded.
\end{proof}

\centerline{------}

\end{document}